\documentclass[twocolumn]{revtex4}

\usepackage[latin1]{inputenc}
\usepackage{amsmath}
\usepackage{amsfonts}
\usepackage{amssymb}
\usepackage{amsthm}
\usepackage{mathrsfs}
\usepackage{graphicx}

\newtheorem{Proposition}{Proposition}
\newtheorem{Lemma}{Lemma}
\newtheorem{Corollary}{Corollary}
\newtheorem{Theorem}{Theorem}
\newtheorem{Definition}{Definition}

\begin{document}

\title{Truly work-like work extraction}
\author{Johan {\AA}berg}
\email{johan.aberg@physik.uni-freiburg.de}
\affiliation{Institute for Physics, University of Freiburg, Hermann-Herder-Strasse 3, D-79104 Freiburg, Germany}
\affiliation{Institute for Theoretical Physics, ETH Zurich, 8093 Zurich, Switzerland}

\begin{abstract}
The work content of non-equilibrium systems in relation to a heat bath is often analyzed in terms of expectation values of an underlying random work variable. However, we show that when optimizing the expectation value of the extracted work, the resulting extraction process is subject to intrinsic fluctuations, uniquely determined by the Hamiltonian and the initial distribution of the system. These fluctuations can be of the same order as the expected work content per se, in which case the extracted energy is unpredictable, thus intuitively more heat-like than work-like. This raises the question of the `truly' work-like energy that can extracted. Here we consider an alternative that corresponds to an essentially fluctuation-free extraction.  We show that this quantity can be expressed in terms of a non-equilibrium generalization of the free energy, or equivalently in terms of a one-shot relative entropy measure introduced in information theory.
\end{abstract}

\maketitle

The amount of useful energy that can be harvested from non-equilibrium systems not only characterizes practical energy extraction and storage, but is also a fundamental thermodynamic quantity. Intuitively, we wish to extract ordered and predictable energy, i.e., `work', as opposed to  disordered random energy in the form of `heat'. The catch is that, in statistical systems, the work cost or yield  of a given transformation is typically a random variable \cite{ReviewFluctThm}. This raises the question of a quantitative notion of work content that truly reflects the idea of work as  ordered  energy. Here we show that standard expressions  for the work content \cite{Procaccia,Lindblad1983,Takara,Esposito} can correspond to a very noisy and thus heat-like energy, but we also introduce an alternative that quantifies the amount of ordered energy that can be extracted.  The latter can be expressed in terms of a non-equilibrium generalization of the free energy, or equivalently in terms of a one-shot relative entropy introduced in information theory.  
The work extraction problem is linked to information theory via concepts like Szilard engines, Landauer's principle, and Maxwell's demon \cite{LeffRexI,LeffRexII}, with  recent contributions in connection to one-shot information theory \cite{Dahlsten, delRio, Faist}. 
A direct consequence of the present investigation is that the latter is brought into a more physical setting, allowing, e.g., systems with non-trivial Hamiltonians, proof of near-optimality, as well as a connection to fluctuation theorems \cite{ReviewFluctThm}.
Similar results as in this study have been obtained independently in \cite{Horodecki11}. See also recent results in \cite{Egloff} based on ideas in \cite{EgloffThesis}.

The amount of work that a system can perform while it equilibrates with respect to an environment of temperature $T$ is often \cite{Procaccia,Lindblad1983,Takara,Esposito} expressed as
\begin{equation}
\label{OptimalExpectedExtractionMain}
\mathcal{A}(q,h) = kT\ln(2)D\boldsymbol{(}q\Vert G(h)\boldsymbol{)}.
\end{equation}
Here $q$ is the state of the system, $G(h)$ its equilibrium state, $h$ the system Hamiltonian,  and $k$ Boltzmann's constant. For the simple model we employ here, $q$ is a probability distribution over a finite set of energy levels, and $D(q\Vert p) = \sum_{n}q_n\log_2 q_n -\sum_{n}q_n\log_2p_n$ is the relative Shannon entropy  (Kullback-Leibler divergence) \cite{CoverThomas}, and $\log_2$ denotes the base $2$ logarithm.

The quantity $\mathcal{A}(q,h)$, and the closely related cost of information erasure (Landauer's principle), is often understood as an expectation value of an underlying random work yield (see e.g.~\cite{Procaccia,Takara,Shizume,Piechocinska}). However, this tells us very little about the fluctuations, and thus the `quality' of the extracted energy. Here we show that  optimizing the expected gain leads to intrinsic fluctuations. These can be of the same order as the expected work content $\mathcal{A}(q,h)$ per se, in which case the work extraction does not act as a truly ordered energy source.
As an alternative, we introduce the \emph{$\epsilon$-deterministic work content}, which quantifies the maximal amount of energy that can be extracted if we demand to always get precisely this energy each single time we run the extraction process, apart from a small probability of failure $\epsilon$. This quantity formalizes the idea of an almost perfectly ordered energy source.

Our analysis is based on a very simple model of a system interacting with a heat bath of fixed temperature $T$. Akin to, e.g., \cite{Crooks98,Piechocinska,delRio}, we model the Hamiltonian of the system as finite set of energy levels $h = (h_{1},\ldots, h_{N})$, and the state $q$ as a probability distribution over these. We can raise or lower the energy levels at will, which we refer to as \emph{level transformations} (LT). (For a quantum system this would essentially correspond to adiabatic evolution with respect to some external control parameters.)  Via the LTs we define what `work' is in our model. If we perform an LT that changes $h$ to $h'$, and  if the system is in state $n$, then this results in a work gain $h_n-h'_n$ (or work cost $h'_n-h_n$). 
To model the thermalization, we put the system into the random state $\mathcal{N}$ described by the Gibbs distribution, $P(\mathcal{N}=n) = G_{n}(h)$, where $G_{n}(h) = e^{-\beta h_{n}}/Z(h)$, $\beta = 1 / (kT)$, and $Z(h) = \sum_{n}e^{-\beta h_{n}}$ is the partition function. It is furthermore assumed that the state (regarded as a random variable) after a thermalization is independent of the state before. We combine sequences of LTs and thermalizations to construct processes. An example is given in Fig.~\ref{fig1}, where we construct the analogue of isothermal reversible (ITR) processes, which serve as a building block in our analysis. As opposed to other processes we will consider, the ITRs have essentially fluctuation-free work costs.

 \begin{figure}[h]
 \includegraphics[width= 8.9cm]{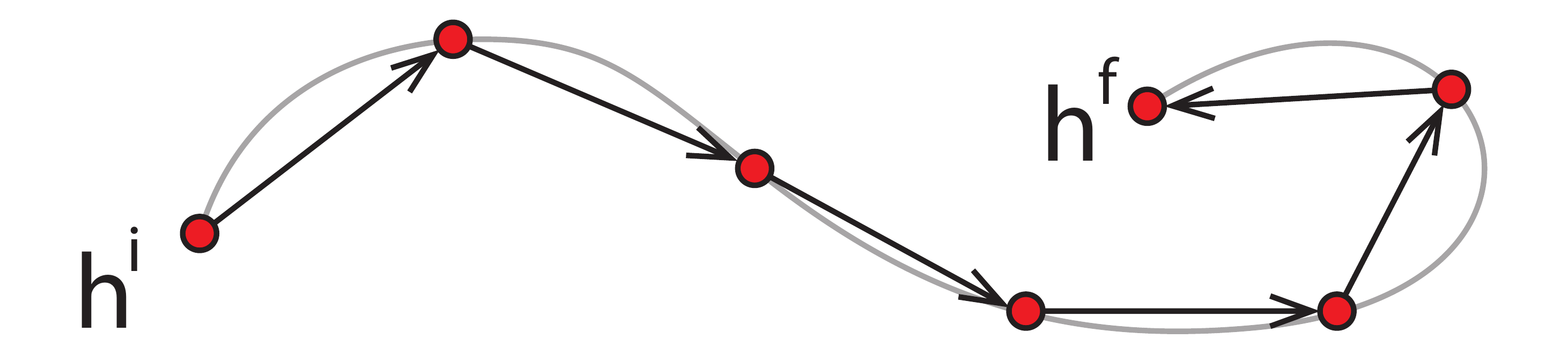} 
   \caption{\label{fig1} {\bf Isothermal reversible processes.} 
   In the space of energy level configurations we connect an initial configuration $h^{i}\in\mathbb{R}^{N}$ with the final $h^{f}\in\mathbb{R}^{N}$ by a smooth path (gray line).
    Given an $L$-step discretization of this path, we construct a sequence of LTs (arrows) sandwiched by thermalizations (circles).
   This process has the random work cost $W = \sum_{l=0}^{L-1}(h^{l+1}_{\mathcal{N}^{l}} - h^{l}_{\mathcal{N}^{l}})$, where $\mathcal{N}^{l}$ is the state at the $l$-th step, which is Gibbs distributed $G(h^{l})$. In the limit of an infinitely fine discretization, the expected work cost is $\lim_{L\rightarrow \infty}\langle W\rangle = F(h^{f})-F(h^{i})$. The independence of the work costs of the subsequent LTs, yields $\lim_{L\rightarrow \infty}(\langle W^2\rangle -\langle W\rangle^2) = 0$, i.e.,  the work cost is essentially deterministic. 
   }
\end{figure}

Given an initial state $\mathcal{N}$ with distribution $q$, we can reproduce Eq.~(\ref{OptimalExpectedExtractionMain}) within our model.
A cyclic three-step process, as described in Fig.~(\ref{fig2}), gives the random work yield
\begin{equation}
\label{OptYieldVariable}
W_{\textrm{yield}} = kT\ln q_{\mathcal{N}}-kT\ln G_{\mathcal{N}}(h).
\end{equation}  
By taking the expectation value we obtain Eq.~(\ref{OptimalExpectedExtractionMain}). 
The positivity of relative entropy, $D(q\Vert p)\geq 0$, can be used to show that no process can give a better expected work yield (Proposition \ref{OptimalExpectedWork} in Appendix \ref{Sec:OptimalExpected}).

 \begin{figure}[h]
 \includegraphics[width= 8.9cm]{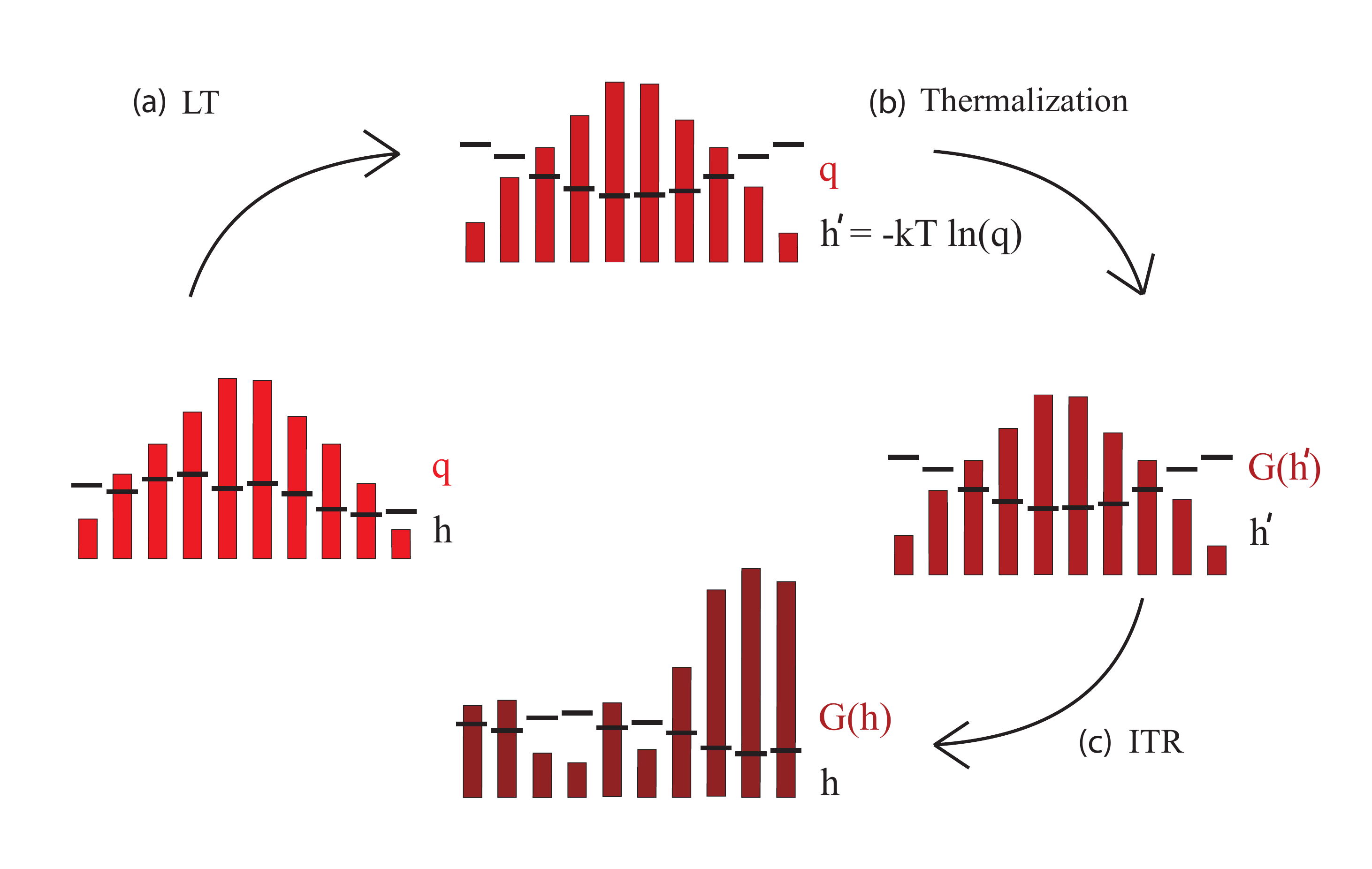} 
   \caption{\label{fig2} {\bf Expected work extraction.} 
    For an initial state $\mathcal{N}$ with distribution $q$ (bars) and energy levels $h$ (horizontal lines), the expected work content $\mathcal{A}(q,h)$ is obtained by a cyclic three-step process. The idea is to avoid unnecessary dissipation when the system is put in contact with the heat bath. To this end we make an LT to a set of energy levels $h'$ for which $G(h')=q$. The total process is:
       (a) LT that transforms $h_n$ into $h'_{n} = -kT\ln q_n$.  (b) Thermalization, resulting in the Gibbs distribution $G(h') = q$. (c) ITR process from $h'$ back to $h$. 
       The resulting random work yield is $W_{\textrm{yield}} = kT\ln q_{\mathcal{N}}-kT\ln G_{\mathcal{N}}(h)$, with expectation value $\langle W_{\textrm{yield}}\rangle = kT\ln(2)D\boldsymbol{(}q\Vert G(h)\boldsymbol{)}$. }
\end{figure}

How large are the fluctuations for a process that maximizes the expected work extraction, and thus achieves $\mathcal{A}(q,h)$? Equation (\ref{OptYieldVariable}) determines the noise of the specific process in Fig.~\ref{fig2}, but it turns out that it actually specifies the fluctuations for all processes that optimize the expected work extraction. We can phrase this result more precisely as follows.
For a process $\mathcal{P}$ that operates on an initial state $\mathcal{N}$ with distribution $q$, we let $W_{\textrm{yield}}(\mathcal{P},\mathcal{N})$ denote the corresponding random work yield. We here consider cyclic processes that starts and ends in the energy levels $h$. If  $(\mathcal{P}_{m})_{m=1}^{\infty}$ is a family of processes such that $\lim_{m\rightarrow\infty }\langle W_{\textrm{yield}}(\mathcal{P}_m,\mathcal{N})\rangle =  \mathcal{A}(q,h)$,  then $W_{\textrm{yield}}(\mathcal{P}_m,\mathcal{N}) \rightarrow  kT\ln q_{\mathcal{N}}- kT\ln G_{\mathcal{N}}(h)$ in probability.  (For a proof see Sec.~\ref{fluctuations} in the Appendix.)
We can conclude that to analyze the noise in the optimal expected work extraction, it is enough to consider Eq.~(\ref{OptYieldVariable}). As we will confirm later, these fluctuations can be of the same order as $\mathcal{A}(q,h)$ itself.


Since the optimal expected work extraction suffers from fluctuations, a natural question is how much (essentially) noise-free energy can be extracted. We say that a random variable $X$ has the $(\epsilon,\delta)$-deterministic value $x$, if the probability to find $X$ in the interval $[x-\delta,x+\delta]$ is larger than $1-\epsilon$. Hence, $\delta$ is the precision by which  the value $x$ is taken, and $\epsilon$ the  largest probability by which this fails. We define $\mathcal{A}^{\epsilon}_{\delta}(q,h)$ as the highest possible $(\epsilon,\delta)$-deterministic work  yield among all processes that operate on the initial distribution $q$ with initial and final energy levels $h$. Next, we define the $\epsilon$-deterministic work content as $\mathcal{A}^{\epsilon}(q,h) = \lim_{\delta\rightarrow 0}\mathcal{A}^{\epsilon}_{\delta}(q,h)$, i.e., we take the limit of infinite precision, thus formalizing the idea of an energy extraction that is essentially free from fluctuations.

$\mathcal{A}^{\epsilon}(q,h)$  can be expressed in terms of the \emph{$\epsilon$-free energy}, which is defined via restrictions to sufficiently likely subsets of energy levels.  Given a subset $\Lambda$, we define $Z_{\Lambda}(h) = \sum_{n\in\Lambda}e^{-\beta h_n}$. We minimize $Z_{\Lambda}(h)$ among all subsets $\Lambda$ such that $q(\Lambda) = \sum_{n\in\Lambda}q_n > 1-\epsilon$. If $\Lambda^{*}$ is such a minimizing set, then the $\epsilon$-free energy is defined as $F^{\epsilon}(q,h) = -kT\ln Z_{\Lambda^{*}}(h)$. The concept of one-shot free energy has been  introduced independently in \cite{Horodecki11}.

The distribution of fluctuations is clearly important for determining the value of $\mathcal{A}^{\epsilon}(q,h)$. It is thus maybe not surprising that a fluctuation theorem plays an important role to show the following bound on the $\epsilon$-deterministic work content
\begin{equation}
\label{bound}
0 \leq \mathcal{A}^{\epsilon}(q,h) -  F^{\epsilon}(q,h) + F(h) \leq   -kT\ln(1-\epsilon).
\end{equation}
In other words, for small $\epsilon$ we have 
\begin{equation*}
\mathcal{A}^{\epsilon}(q,h) \approx F^{\epsilon}(q,h) - F(h). 
\end{equation*}
In the case of completely degenerate energy levels $h = (r,\ldots,r)$, Eq.~(\ref{bound}) reduces to the result in \cite{Dahlsten}.
We obtain the lower bound in Eq.~(\ref{bound}) by the process  described in Fig.~\ref{fig3}.  
The upper bound is obtained by a combination of a version (Lemma \ref{Crooks}) on Crook's fluctuation theorem \cite{CrooksTheorem} and a work bound for  LTs (Sec.~\ref{ProofofMain} in the Appendix). See also Sec.~\ref{Sec:AlmDetEr} for the $\epsilon$-deterministic work cost of information erasure.

One can show (Sec.~\ref{thermfluct}) that $-kT\ln(1-\epsilon)$ is an upper bound to the $\epsilon$-deterministic work content of equilibrium systems.  Equation (\ref{bound}) thus determines the value of  $\mathcal{A}^{\epsilon}(q,h)$ up to an error with the size of a sufficiently probable equilibrium fluctuation.

The above result can be reformulated in terms of an $\epsilon$-smoothed relative R\'enyi $0$-entropy, defined as    
$D^{\epsilon}_{0}(q\Vert p) =  -\log_{2}\min_{q(\Lambda)> 1-\epsilon}\sum_{j\in \Lambda} p_{j}$. 
This relative entropy was (up to some technical differences) introduced in \cite{Wang1} in the context of one-shot information theory. (See \cite{Datta,Wang2} for quantum versions.) 
One can see that 
\begin{equation}
\label{relation}
F^{\epsilon}(q,h)-F(h) = kT\ln(2)D^{\epsilon}_{0}\boldsymbol{(}q\Vert G(h)\boldsymbol{)}.
\end{equation}

\begin{figure}[h]
\includegraphics[width= 8.9cm]{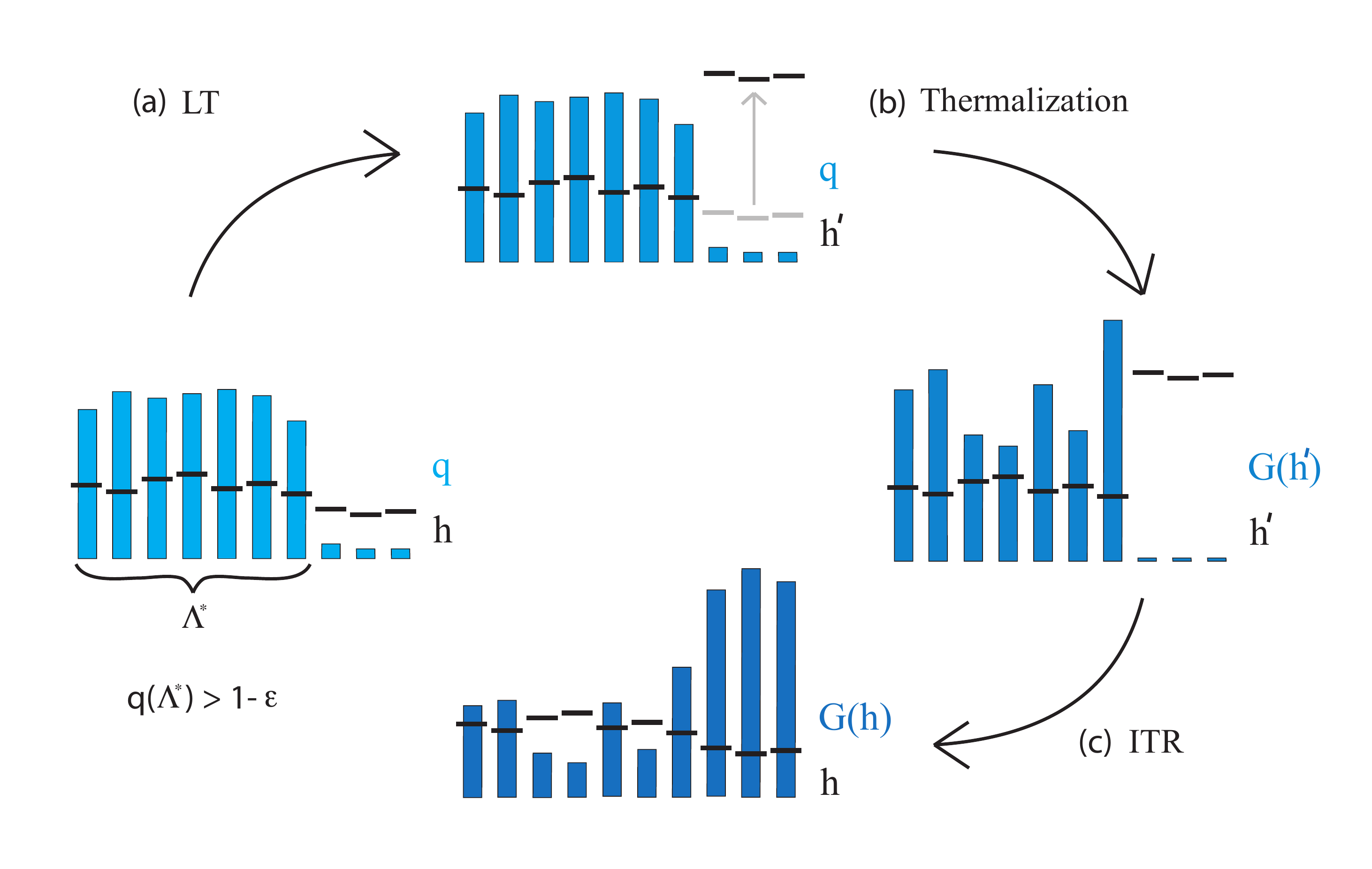} 
   \caption{\label{fig3} {\bf $\epsilon$-deterministic work extraction.} 
      For a state distribution $q$ and energy levels $h$, let $\Lambda^{*}$ be a subset of the energy levels such that $F^{\epsilon}(q,h) = -kT\ln Z_{\Lambda^*}(h)$.
    (a) LT that lifts all energy levels not in $\Lambda^{*}$ to a very high value, i.e.,  $h'_{n} = h_{n}$ if $n\in \Lambda^{*}$, while $h'_{n} = h_{n}+ E$ if $n\notin \Lambda^{*}$. (b) Thermalization, resulting in the Gibbs distribution $G(h')$. (c) ITR process from $h'$ back to $h$, which gives the essentially deterministic work yield $F(h')-F(h)$.
    In the limit $E\rightarrow +\infty$ this process gives the work yield $F^{\epsilon}(q,h)-F(h)$ 
    with a probability larger than $1-\epsilon$.  This is a lower bound to $\mathcal{A}^{\epsilon}(q,h)$, but is also close to it for small $\epsilon$.
   }
\end{figure}


An immediate question is how $\mathcal{A}(q,h)$ compares with $\mathcal{A}^{\epsilon}(q,h)$, and with the fluctuations in the optimal expected work extraction. 
The latter we measure by the standard deviation of  $W_{\textrm{yield}}$ in  Eq.~(\ref{OptYieldVariable}), $\sigma(W_{\textrm{yield}})= (\langle W_{\textrm{yield}}^2\rangle -  \langle W_{\textrm{yield}}\rangle^2)^{1/2}$.
We compare how these three quantities scale with increasing system size (e.g., in number of spins, or other units).

Our first example is a collection of $m$ systems whose state distributions are independent and identical, $q^{m}(n_1,\ldots, n_m) = q(n_1)\cdots q(n_m)$, and which have non-interacting identical Hamiltonians, corresponding to energy levels $h^m(n_1,\ldots,n_m) = h(n_1) +\cdots + h(n_m)$. In this case $\mathcal{A}(q^{m},h^{m})= m\mathcal{A}(q, h)$, and $\sigma(W_{\textrm{yield}}^{m}) = \sqrt{m}kT\ln(2)\sigma\boldsymbol{(}q\Vert G(h)\boldsymbol{)}$, where $\sigma(q\Vert r)^2 = \sum_{n}q_n[\log_2(q_n/r_n)]^2 -D(q\Vert r)^2$. 
One can show that 
 $\mathcal{A}^{\epsilon}(q^{m},h^{m})$ is equal to $m\mathcal{A}(q, h)$ to the leading order in $m$ (see \cite{Brandao11} for a similar result in a resource theory framework). 
 More precisely,  
\begin{equation}
\label{expansion}
\begin{split}
\mathcal{A}^{\epsilon}(q^{m},h^{m}) = &  m \mathcal{A}(q,h) + \sqrt{m}kT\ln(2)\Phi^{-1}(\epsilon)\sigma\boldsymbol{(}q\Vert G(h)\boldsymbol{)}\\
& + o(\sqrt{m}),
 \end{split}
\end{equation}
where $o(\sqrt{m})$ is a correction term that grows slower than $\sqrt{m}$,
  and $\Phi^{-1}$ is the inverse of the cumulative distribution function of the standard normal distribution. The smaller our error tolerance $\epsilon$, the more the correction term lowers the value of $\mathcal{A}^{\epsilon}(q^m,h^m)$ as compared to $\mathcal{A}(q^m,h^m)$.  This expansion is proved via Berry-Esseen's theorem \cite{Berry,Esseen}, which determines the convergence rate in the central limit theorem (see Sec.~\ref{IndepIdentical} in the Appendix).  
  
 As seen from Eq.~(\ref{expansion}) the difference between the expected and the $\epsilon$-deterministic work content only appears at the next to leading order. Hence, in these particular systems the fluctuations are comparably small, and the dominant contribution to $\mathcal{A}^{\epsilon}(q^{m},h^{m})$ is $\mathcal{A}(q^{m},h^{m})$. It appears reasonable to expect similar results for non-equilibrium systems with sufficiently fast spatial decay of both correlations and interactions, which may explain why issues concerning  $\mathcal{A}$ as a measure of work content appear to have gone largely unnoticed.

A simple modification of the state distribution in the previous example results in a system with large fluctuations. With probability $1-\epsilon$ (independent of $m$) the system is prepared in the joint ground state $0,\ldots, 0$,  and with probability $\epsilon$  in the Gibbs distribution. This results in  $q^{m}_{l_1,\ldots, l_m}  = (1-\epsilon)\delta_{l_1,0}\cdots\delta_{l_m,0}  +  \epsilon G_{l_1}(h)\cdots G_{l_m}(h)$, and yields $\mathcal{A}(q^{m}, h^{m}) \sim - mkT\ln(2)(1-\epsilon ) \log_{2}G_{0}(h)$, $\sigma(W^{m}_{\textrm{yield}}) \sim  -mkT\ln (2)\sqrt{\epsilon(1  - \epsilon)}\log_{2}G_{0}(h)$, and $\mathcal{A}^{\epsilon}(q^{m},h^{m})\sim  -mkT\ln(2)\log_{2}G_{0}(h)$. Hence, all three quantities grow proportionally to $m$.

For second case of large fluctuations we choose the distribution $q^{m}_{l_1,\ldots,l_m} = d^{-m}$, for a collection of $d$-level systems. For large $m$ we assume that the energy levels are dense enough that they can be  replaced by a spectral density. One example  is Wigner's semi-circle law, where $f^{(m)}(x) = 2\sqrt{R(m)^2-x^2}/[\pi R(m)^2]$ for $|x|\leq R(m)$. With $R(m) = \sqrt{2}d^{m/2}$ this is the asymptotic energy level distribution of large random matrices from the Gaussian unitary ensemble \cite{Mehta}. For the semi-circle distribution   
$\mathcal{A}(q^{m},h^{m})\sim R(m)$, $\sigma(W^{m}_{\textrm{yield}}) \sim R(m)/2$, and $\mathcal{A}^{\epsilon}(q^{m},h^{m})\sim  c(\epsilon)R(m)$, where $c(\epsilon)$ is independent of $m$.


We have here employed what one could refer to as a discrete classical model.
Relevant extensions include a classical phase-space picture, as well as a quantum setting that  allows superpositions between different energy eigenstates (e.g., in the spirit of \cite{Janzing,Horodecki11,Brandao11}) and where the work-extractor can possess quantum information about the system \cite{delRio}. An operational approach, based on what `work' is supposed to achieve, rather than ad hoc definitions, may yield deeper insights to the question of the truly work-like energy content.

It is certainly justified to ask for the relevance of the effects we have considered here.
The evident role of fluctuations suggests that the noise in the expected work extraction should become noticeable in the same nano-regimes as where fluctuation theorems are relevant. The considerable experimental progress on the latter (see e.g.~\cite{Liphardt02, Collin05, Toyabe10}) should reasonably be applicable also to the former. Also the theoretical aspects of the link to fluctuation theorems merits further investigations.

In principle, the fluctuations in the expected work extraction can be large also outside the microscopic regime, as this only requires a sufficiently `violent' relation between the non-equilibrium state and the Hamiltonian of the system. As opposed to the expected work content, $\mathcal{A}^{\epsilon}$ retains its  interpretation as the ordered energy. It is no coincidence that this is much analogous to how single-shot information theory generalizes `standard' information theory \cite{Holenstein11,Tomamichel09}. In this spirit, the present study, along with \cite{Dahlsten,delRio,Horodecki11,Egloff,Faist}, can be viewed as the first glimpse of a `single-shot statistical mechanics'.

The author thanks L\'idia del Rio, Renato Renner, and Paul Skrzypczyk for useful comments. 
This research was supported by the Swiss National Science Foundation
through the National Centre of Competence in Research `Quantum Science
and Technology', by the European Research Council, grant no.~258932,
 and the Excellence Initiative of the German Federal and State Governments (grant ZUK 43).


\begin{widetext}
\end{widetext}

\begin{appendix}


\section{\label{Remarks}Some remarks}
 
Here we give a bit more background to the main motivations and goals of this project.

 \emph{The standard approaches as `expectation value settings'.}--
One of the main goals of this investigation is to relate and compare the `standard' approaches to work extraction and information erasure with the single-shot setting, and to argue that the latter quantifies the useful energy content of systems in a way that is closer to our intuitive notion of work as ordered energy. 
As mentioned in the main text, the standard approaches can in some sense be regarded as `expectation value settings'. By this we mean that 
work extraction, information erasure, as well as other transformations, are often directly or indirectly analyzed in terms of expectation values, and limits on the costs of these operations are often expressed in terms of expectation 
values \cite{Procaccia,Shizume,Piechocinska,Kawai,Horowitz,Zolfaghari,Takara,Vaikuntanathan}. In the quantum case, there are several investigations where work in some sense is expressed in terms of combinations of expectation values with respect to quantum states of the system  (see, e.g., \cite{Allahverdyan,Maroney,Alicki,Sagawa}). In these quantum cases it is of course not clear to what extent one can (or whether one should) associate an underlying random work variable to the expectation value, unless explicit measurements are included in the model. We will not consider this issue here, but rather use an essentially discrete classical model, as detailed in Sec.~\ref{model}.  It is maybe worth to point out that the above comments are not meant to imply that work can, or generally is, treated as an observable \cite{Talkner}. Rather, it is often defined in terms of differences or changes of expectation values, e.g., accumulated over a path, akin to what we do in Sec.~\ref{ITR}.

\emph{Non-trivial Hamiltonians in the single-shot setting.}-- For a general treatment of equilibrium and non-equilibrium statistical mechanics it is certainly vital to allow systems to carry non-trivial Hamiltonians. One goal of this investigation is to incorporate this component into the single-shot analysis of work extraction and information erasure. While fairly straightforward in the standard (expectation value) setting, this is more challenging in a single-shot analysis. The main reason why the case of completely degenerate Hamiltonians is easier to handle is because then one can argue that arbitrary unitary operations (in the quantum case) or arbitrary permutations (in the classical case) come `for free' from a thermodynamic point of view, as they do not change the energy distribution of the system. This freedom opens up the entire toolbox of classical and quantum information theory, which was utilized in previous single-shot analyses \cite{Dahlsten, delRio}. In the non-degenerate case it becomes less clear how to use such techniques (concepts like energy and Hamiltonians are somewhat alien to the typical information theoretic setting). The main problem is that we must account for the changes in energy induced by state transformations, as to not manipulate the overall thermodynamic balance sheet in some dubious manner. In this investigation we handle this issue by making sure that all changes of the state distribution of the system always goes towards thermal equilibrium. (We are even so brutal as to put the system directly in the equilibrium distribution at each contact with the heat bath. See Sec.~\ref{model}.) For alternative approaches to handle non-trivial Hamiltonians in the singe-shot setting, see \cite{Horodecki11} and \cite{Egloff}.

\emph{Work extraction and information erasure as optimization problems.}--
A further goal of this investigation is to formulate work extraction (and information erasure) as an optimization problem over a well defined physical model. In contrast to many heuristic approaches, we here formulate a model that clearly specifies the rules of the game, and we optimize over \emph{all} processes that are allowed within this thermodynamic toy universe. Needless to say, we have to strike a balance with tractability, why we settle with a rather simple model. Section~\ref{model} introduces this model and defines the set of processes over which we optimize.  To each such process is, by construction, associated a probability distribution of possible work costs. To obtain a well defined optimization problem it is of course not enough to define the set over which to optimize; we must also specify a cost function. In our case these cost functions are functionals on the space of probability distributions. In other words, we assign a value to the probability distribution of work costs of the process, rather than to particular outcomes in each single run.   (Furthermore, we do not strictly speaking minimize this cost function in the sense of finding a minimizing element, but we rather determine the infimum over all allowed processes.) Each choice of cost function potentially corresponds to a different formalization of what `work content' is supposed to be. In this investigation we compare two cost functions: the expectation value and the $(\epsilon,\delta)$-deterministic value (to be defined in Sec.~\ref{Properties}). The former case leads to standard results on work extraction and information erasure, while the latter defines the $\epsilon$-deterministic work yield or cost of these tasks.

\emph{$\epsilon$-deterministic energy vs.~other energy concepts.}--
As stated earlier, one of the main goals of this investigation is to characterize the essentially noise-free energy content of non-equilibrium states. It is maybe worth repeating that this study is not necessarily restricted to a nano-scale regime, but rather strives to generally define in a quantitative manner what we mean by `useful energy', irrespective of scale (although the effects certainly would be relevant in a microscopic regime). 
Moreover, in analogy with other ideal thermodynamic concepts, e.g., the Carnot efficiency, we do not here concern with questions of practical achievability, but rather regard the $\epsilon$-deterministic work content as ideal quantity to which all realistic implementations can be compared. 
In this idealized setting it appears natural to capture the notion of ordered energy by demanding that the energy source always produces a predefined energy with almost perfect certainty, as we do with the $\epsilon$-deterministic values. Undoubtedly there are many alternative definitions of work content that potentially could capture other relevant aspects of energy extraction. For example,  in some applications it might be sufficient to know that the energy yield is beyond a given threshold energy (e.g. to drive a chemical reaction). Such a threshold quantity has been considered in \cite{Dahlsten,Egloff}. However, since it is not a priori clear that such threshold quantities do capture the idea of almost perfectly ordered energy we do not consider that approach here (see Sec.~\ref{OtherCostFcns} for a discussion on this, where we also discuss the possibility that the $\epsilon$-deterministic work content might be close to the threshold quantities). 

Since we wish to characterize ordered energy it furthermore appears natural to focus on the regime of small failure probabilities $\epsilon$. 
However, technically speaking our results are valid in the regime $0<\epsilon \leq 1-1/\sqrt{2}$ (see Corollaries \ref{AlmostDetermWorkCont} and \ref{AlmostDetermEreasureCost}) although the maybe more relevant aspect is that we determine the $\epsilon$-deterministic work content up to an error of the size $-kT\ln(1-\epsilon)$. To go beyond our present focus of noise-free energy, it would certainly be interesting to pinpoint the exact value of  the $\epsilon$-deterministic work content for all $0< \epsilon < 1$, i.e., we could consider different risk tolerance regimes, akin to \cite{Dahlsten}.
However, as discussed above, the assumption of non-trivial Hamiltonians appears to prohibit a direct application of the techniques of the latter approach. To approach this question in the present optimization setting with non-trivial Hamiltonians goes beyond the tools developed in this investigation and will not be considered here. However, see recent results in \cite{Egloff} that uses other techniques to combine non-trivial Hamiltonians with arbitrary success probabilities.

As a technical side remark one may note that the error bound  $-kT\ln(1-\epsilon)$ goes to zero as the failure probability $\epsilon$ goes to zero. This could be compared to the (at the time of writing) typical single-shot information theoretic error terms, as in \cite{Dahlsten, Faist},  which are of the form $\ln(1/\epsilon)$ and thus diverge with a decreasing $\epsilon$. Furthermore, $-kT\ln(1-\epsilon)$ has an interpretation as the size of a thermal equilibrium fluctuation (see Sec.~\ref{thermfluct}) and is thus in a  thermodynamic sense `small'. However, beyond the purely aesthetic aspects, it is far from clear what significance such differences concerning error terms have in the thermodynamic setting. (In information theoretic applications, the divergent error terms are usually unproblematic.)

\emph{Single-shot vs.~the `multi-copy' iid setting.}--
As mentioned in the main text, the single-shot scenario can (like for the information-theoretic counterparts \cite{Holenstein11,Tomamichel09}) be regarded as more `fundamental' than the `multi-copy' scenario of iid states and identical non-interacting Hamiltonians, in the sense that the latter can be derived form the former as a special case. The crucial question is maybe rather to what extent the more general single-shot scenario is relevant. In the multi-copy case we have seen that the $\epsilon$-deterministic work content is to the leading order equal to the expected work content. (Note, however, the difference between the single-shot iid case and the expectation value setting. See Sec.~\ref{expecvssingleshotiid}.) In view of standard equilibrium statistical mechanics one could thus suspect that many realistic systems would have states and Hamiltonians close to this regime. However, one should keep in mind that we here consider \emph{non-equilibrium} statistical mechanics, where we allow states arbitrarily far from equilibrium. There is thus no particular reason why we should assume the states to be near iid, irrespective of the structure of the underlying Hamiltonian. 

On a more broad level one can view this investigation as a step towards a better understanding of the foundations of statistical mechanics. For such a purpose it appears more satisfying with a formalism that has the capacity to handle arbitrary states and Hamiltonians; not being restricted to special cases like iid assumptions or non-interacting Hamiltonians. Furthermore, since realistic physical systems in general neither are perfectly iid nor perfectly non-interacting, this immediately spurs the question of the quality of the approximation we implicitly invoke by assuming an analysis based on a multi-copy setting. The single-shot setting can be used to answer such questions.

\emph{Not a study on the emergence of thermodynamics.}--
As mentioned, this investigation, together with \cite{Dahlsten,delRio,Horodecki11,Egloff,Faist} can be regarded as the first steps toward a single-shot statistical mechanics. It is certainly a relevant question how the issues considered here relate to the countless of studies on the foundations of thermodynamics and statistical mechanics, the emergence of irreversibility, the second law, and other aspects of  standard thermodynamics and statistical mechanics. As an illustrative (but very small) selection one can mention studies of emergence of thermodynamics in closed quantum systems \cite{GemmerMahler}, the efforts to understand the existence of canonical equilibrium states via entanglement \cite{Goldstein, Popescu06}, or the relation between  entanglement and the thermodynamic arrow of time \cite{Partovi,Jennings}.
In contrast to these studies we do not here strive to analyze the very emergence of thermodynamics per se. We essentially put in irreversibility, the second law, and canonical thermal states by hand when we model the contact with the heat bath as replacement maps that put the system in the Gibbs distribution (see Sec.~\ref{model}). This investigation should rather be understood as an approach towards a more fine-tuned quantitative characterization of the consequences of the second law, where we, e.g., ask how ordered energy should be characterized in a consistent manner for arbitrary distributions and Hamiltonians.

\emph{To avoid confusion: Auxiliary comments on the technical scope and terminology.}-- 
For the sake of simplicity and conceptual clarity we have in the main text focused on the work extraction problem. However, due to the close relation between these two tasks, we here also treat the work cost of information erasure.  Moreover, rather than assuming that these processes begin and end in the very same collection of energy levels $h$ (as we did in the main text) we allow an initial set of energy levels $h^{i}$ and final set of energy levels $h^f$. 
 This makes it easier to use these processes as building blocks in a composition of processes. Furthermore, for convenience, and to underline the similarities between work extraction and information erasure, we will often express the former in terms of its \emph{work cost} rather than its work yield as we did in the main text. For example, in Sec.~\ref{Sec:OptimalExpected} we introduce $\mathcal{C}^{\textrm{extr}}(q,h^{i},h^{f})$ as the minimal expected work cost of transforming the system from the set of energy levels $h^{i}$ to the new levels $h^{f}$, assuming the initial state is distributed $q$. Hence, the expected work content, as introduced in the main text, is $\mathcal{A}(q,h) = -\mathcal{C}^{\textrm{extr}}(q,h,h)$. 
 
\emph{Structure of this Appendix.}--
The structure of this Appendix is as follows. In Sec.~\ref{model} we introduce the model we employ throughout this investigation.
In Sec.~\ref{ITR} we consider a class of processes in our model that correspond to isothermal reversible processes.
We consider the expectation value as cost function for work extraction in  Sec.~\ref{Sec:OptimalExpected}, and for the information erasure in Sec.~\ref{Sec:OptimalWorkCostErasure}.
In  Sec.~\ref{fluctuations} we prove that the optimal expected work extraction has an intrinsic randomness associated to it.
Section~\ref{Properties} introduces the alternative cost function, the $(\epsilon,\delta)$-deterministic value. Section~\ref{Sec:FandD0} defines the $\epsilon$-free energy and a smooth relative R\'enyi $0$-entropy. 
These concepts are applied to $\epsilon$-deterministic work extraction in Sec.~\ref{Sec:AlmDetWorkExtr}, with proofs in Sec.~\ref{ProofofMain}.
Section \ref{thermfluct}  concerns a brief clarification on the $\epsilon$-deterministic work  extraction from thermal equilibrium systems.
We turn to the question of the $\epsilon$-deterministic  work cost of information erasure in Sec.~\ref{Sec:AlmDetEr}, with proofs in Sec.~\ref{ProofAlmDetErasure}.
In Sec.~\ref{comparisons}  we compare the expected work content with the $\epsilon$-deterministic work content for some simple examples. We also compare the expected work content with the fluctuations in processes that achieve the optimal expected work content. 
We end with Sec.~\ref{OtherCostFcns}, where we make a brief comment on an alternative type of cost function, and its relation to the $\epsilon$-deterministic setting.


\section{\label{model}  The model}

The choice of model in this investigation is a compromise between simplicity for tractable optimization problems, and the need to capture some essential aspects of the effects of a heat bath. 
Similar types of models have bee considered in \cite{Crooks98,Piechocinska,Alicki,delRio}.

\subsection{Model assumptions}

\emph{The setting: Probability distributions over energy levels.}--
We assume that the system can be in  a finite set of states $\{1,\ldots,N\}$, where $N$ is a fixed number. To each such state $n$ we associate an energy level $h_n$. In other words, the system can be found in any of the energy levels $h = (h_1,\ldots, h_N)\in \mathbb{R}^{N}$. In general, we will view  the state of the system as a random variable $\mathcal{N}$, with some probability distribution $q = (q_1,\ldots, q_N)\in \mathbb{P}(N)$.  Here, $\mathbb{P}(N)$ denotes the probability simplex over $N$ objects
\begin{equation*}
\mathbb{P}(N): = \{ (q_{1},\ldots, q_{N}): q_{1},\ldots, q_{N}\geq 0, \sum_{n=1}^{N}q_{n} =1\}.
\end{equation*}
By $\mathbb{P}^{+}(N)$ we denote the subset of all distributions with full support, i.e., $q_{n}>0$ for $n= 1,\ldots,N$.

\emph{Elementary operations: Energy level transformations and thermalizations.}-- 
Our model includes two elementary operations that allow us to change  the energy levels and the state of the system, respectively. 

The first type of elementary operation allows us to change the collection of energy levels $h= (h_1,\ldots, h_N)\in\mathbb{R}^{N}$ into a new  configuration $h' = (h'_1,\ldots,h'_N)\in\mathbb{R}^{N}$ of our choice. We refer to this as a \emph{level transformation} (LT). Note that we assume that an LT always transforms  an element $h\in\mathbb{R}^{N}$ to an element $h'\in \mathbb{R}^{N}$. In particular, the underlying number of states does not change, and we do not allow `infinite energies', like $h_n = +\infty$ or $h_n= -\infty$. (The latter is not a particularly severe restriction as we can use limits to essentially the same effect.) The LTs do not affect the state of the system, and thus the random variable $\mathcal{N}'$ describing the state after the LT is identical to the state $\mathcal{N}$ before the transformation, i.e., $\mathcal{N}' = \mathcal{N}$. 

Via the LTs we furthermore define what `work' is in this model.
Given an LT that takes $h$ to $h'$, operating on the initial state $\mathcal{N}$, we define the \emph{work cost} as
\begin{equation}
\label{ElementaryWork}
W := h'_{\mathcal{N}}-h_{\mathcal{N}}.
\end{equation}
We refer to $-W$ as the \emph{work yield} or \emph{work gain}.

The second elementary operation changes the state of the system, and models the thermalization by a heat bath of temperature $T$. We will throughout this investigation assume that this temperature is fixed and given. The thermalization does not change the energy levels $h$, but replaces the state $\mathcal{N}$ with a new \emph{independent} random variable $\mathcal{N'}$ that is Gibbs distributed with respect to $h$, i.e., 
\begin{equation}
P(\mathcal{N}' = n) = G_{n}(h),
\end{equation}
where 
\begin{equation*}
G_{n}(h) := \frac{e^{-\beta h_{n}}}{Z(h)},\quad Z(h) := \sum_{n=1}^{N}e^{-\beta h_{n}},\quad \beta := \frac{1}{kT},
\end{equation*}
and where $k$ is Boltzmann's constant. 
We denote $G(h) := \boldsymbol{(}G_{1}(h),\ldots, G_N(h)\boldsymbol{)}$.
That $\mathcal{N}'$ is `independent' is to be understood such that if we make a sequence of thermalizations, the resulting family of random variables are all independent of each other and of the initial state. The thermalization has no work cost. 

\emph{Processes as arbitrary combinations of elementary operations.}--
When we speak of a \emph{process} $\mathcal{P}$ we mean a finite sequence of LTs and thermalizations (at one fixed temperature $T$). The work cost of a process $\mathcal{P}$ is defined as the sum of the work costs of all the LTs in the process.  We denote the work cost of a  process $\mathcal{P}$ as $W(\mathcal{P},\mathcal{N})$, where $\mathcal{N}$ is the initial state. We let $\mathscr{P}(h^i,h^f)$ denote the collection of all processes that starts in the energy levels $h^i$ and ends with the energy levels $h^{f}$.

Note that the work cost of two LTs are independent only if they are separated by a thermalization. 
 Since two consecutive LT processes can be combined into one single (adding their work costs), and since two consecutive thermalizations have the same effect as one single, we may without loss of generality regard every process in $\mathscr{P}(h^i,h^f)$ as an alternating sequence of thermalizations and LTs. If the system initially is in state $\mathcal{N}$, with distribution $q$, and if  $h^0,\ldots, h^{L}$ is the sequence of energy level configurations, we may thus in general write the work cost of the process as 
\begin{equation}
W(\mathcal{P},\mathcal{N}) = \sum_{l=0}^{L-1}(h^{l+1}_{\mathcal{N}_l} -h^{l}_{\mathcal{N}_l}),
\end{equation}
 where each $\mathcal{N}_l$ is an independent random variable. Here $\mathcal{N}_{0} := \mathcal{N}$, and  $\mathcal{N}_l$ is Gibbs distributed $G(h^{l})$ for each $l\neq 0$.

We let $\mathcal{F}(\mathcal{P},\mathcal{N})$ denote the final state of the process $\mathcal{P}\in\mathscr{P}(h^i,h^f)$ that operates on the initial state $\mathcal{N}$. 
Note that $\mathcal{F}(\mathcal{P},\mathcal{N})$ depends on $\mathcal{N}$ only in the case that $\mathcal{P}$ does not contain any thermalization, due to the assumed independence of subsequent states separated by thermalizations. In the case that $\mathcal{P}$ does contain a thermalization, then $\mathcal{F}(\mathcal{P},\mathcal{N})$ is distributed according to Gibbs distribution of the last thermalization in the process.

\subsection{Brief discussions of the model}

Here we briefly consider possible physical interpretations of the elementary operations in the model, and also discuss some of the inherent limitations.

\emph{LT's as adiabatic transformations.}--
As suggested by our use of phrases such as `energy levels' the most immediate interpretation of this model would be in terms of a quantum system. The LTs would then correspond to adiabatic passage. By this we intend Hamiltonian evolution as a closed quantum system, where the Hamiltonian depends on external parameters (e.g, classical fields) that we change by a much slower rate than the characteristic time-scales of the Hamiltonian. (We may need to take some extra care at possible level crossings.)

\emph{Ideal complete thermalizations.}-- 
The application of the thermalization operation corresponds to turning on the interactions to a heat bath, let the system thermalize, and finally de-connect the system. Especially for small systems, it is certainly a relevant question to what extent it in practice is possible  to implement such procedures in a controlled fashion. However, this investigation aims at understanding the theoretical limitations of ideal thermodynamics, and we do not consider practical issues. Moreover, like in all ideal thermodynamic considerations we do not concern ourselves with questions about kinetics. In other words, we impose no constraints on how much time can be spent on implementing the thermalization procedure (or the adiabatic transformations implementing the LT's). 
This is much in spirit with standard thermodynamic considerations where optimal efficiency (e.g. in a Carnot cycle) typically is reached only in the limit of infinitely slow operations. Under these idealized assumptions, the thermalization model we employ appears a reasonable choice.

It is maybe worth noticing that the way we model the effect of a heat bath is a special case of the detailed balance condition as employed, e.g., in \cite{Crooks98} for a derivation of the Jarzynski equality, or in \cite{Piechocinska} for Landauer's principle. While the detailed balance condition allows a partial or gradual thermalization of the system, our model is somewhat  more brutal in that it directly puts the entire system in the Gibbs distribution.

\emph{Hidden costs of time-dependent operations?}
Our model includes time dependent transformations: the LTs as well as the connection and disconnection to the heat bath. The question is if these time-dependent operations implicitly require hidden resources for their implementations. This time-dependence  could be put in contrast to the time-independent constructions of, e.g., minimal refrigerators and heat pumps \cite{Palao01,Youssef09,Linden10,Skrzypczyk11}. However, those schemes operate in a steady state regime, quite the opposite to the single-shot setting we consider here. 
One could imagine to  `embed'  our time-dependent single-shot scheme into a time-independent construction. This would require us to explicitly model the control systems that implements the time-dependence, and one potentially relevant resource would be the quality and stability of the clock that times the evolution \cite{Peres80,Buzek}. It does not seem unreasonable that considerations of this type could add new layers to the work extraction and information erasure problem. However, it also appears reasonable that the essence of the findings of the present investigation would remain, e.g., in some ideal limit.

\emph{In essence a discrete classical model.}--
If we regard our system as a quantum system, our treatment corresponds to the case that the initial density operator is diagonal in an energy eigenbasis. (Strictly speaking, we should also include an energy measurement to obtain a random variable from the underlying quantum state at each LT. However, as the state already is assumed to be diagonal in an energy eigenbasis, this is more of a technicality.) Hence, even though our model indeed can be phrased in terms of a quantum system, it in essence is a classical discrete model.  As such we can neither analyze the effects of superpositions of energy levels, nor quantum correlations to an observer with a quantum memory as in \cite{delRio}. Such a quantum generalization is likely to require an explicit treatment of the degrees of freedom that carries the extracted energy. As a remark we note that the discrete classical setting of course allows us to investigate effects of purely classical correlations with an observer. However, to avoid further technical complications we do not consider this question here.


\section{\label{ITR}  Isothermal reversible processes}
Isothermal reversible process takes an equilibrium configuration to another, where the work cost is given by the free energy difference of the final and initial system. These processes are generally quasi-static (i.e., at each point along the process, the system is essentially in equilibrium with the heat bath). Here we consider the counterpart of this in our model, which will serve as an important building block in the rest of this investigation.

Given an initial configuration of energy levels $h^i$ and a final set of energy levels $h^f$ we consider a path in the space of energy level configurations along which we move by incremental steps of LT processes sandwiched between thermalizations. In the limit of infinitely small steps we find that the expected work cost converges toward the free energy difference of the final and initial energy configuration.  However, by an argument very similar to the weak law of large numbers, we make the observation that the work cost essentially becomes deterministic. 

Consider a bounded and sufficiently smooth path $h: [0,1]\rightarrow\mathbb{R}^{N}$, such that $h(0) := h^i$ and $h(1) := h^f$.
We make an $L$-step discretization of this path, as $h(l\Delta x)$, for $l= 0,\ldots, L-1$, where $\Delta x = 1/L$.
Given this discretization, we construct a process $\mathcal{P}^{(L)}$ that consists of the sequence of LTs that takes $h(l\Delta x)$ to $h\boldsymbol{(}(l+1)\Delta x\boldsymbol{)}$, sandwiched with thermalizations. The work cost of the process is
\begin{equation}
W(\mathcal{P}^{(L)},\mathcal{N}_0) = \sum_{l=0}^{L-1}\left[h_{\mathcal{N}_{l}}\boldsymbol{(}(l+1)\Delta x\boldsymbol{)} -h_{\mathcal{N}_{l}}(l\Delta x) \right],
\end{equation}
where $\mathcal{N}_{l}$ is the state of the system at the $l$-th step, which has distribution $G\boldsymbol{(}h(l\Delta x)\boldsymbol{)}$. 

If we define
\begin{equation}
\Delta h_{n}(l) := \frac{h_{n}\boldsymbol{(}(l+1)\Delta x\boldsymbol{)} -h_{n}(l\Delta x)}{\Delta x},
\end{equation}
one can see that the expectation value of the work cost is 
\begin{equation}
\label{expectation}
\begin{split}
\langle W(\mathcal{P}^{(L)},\mathcal{N}_0)\rangle = & \sum_{l=0}^{L-1}\sum_{n=1}^{N} \Delta h_{n}(l) G_n\boldsymbol{(}h(l\Delta x)\boldsymbol{)}\Delta x\\
\overset{L\rightarrow\infty}{\rightarrow} & \int_{0}^{1}\sum_{n=1}^{N}  \frac{dh_{n}}{dx} \frac{e^{-\beta h_{n}(x)}}{Z\boldsymbol{(}h(x)\boldsymbol{)}}dx\\
= & F(h^f)-F(h^i),
\end{split}
\end{equation}
where 
\begin{equation}
F(h) = -kT\ln Z(h)
\end{equation}
denotes the free energy of $h$.

By using the assumed independence of the state variables $\mathcal{N}_l$, we can calculate the variance $\sigma\boldsymbol{(}W(\mathcal{P}^{(L)},\mathcal{N}_0)\boldsymbol{)}^2 := \langle W(\mathcal{P}^{(L)},\mathcal{N}_0)^2\rangle -\langle W(\mathcal{P}^{(L)},\mathcal{N}_0)\rangle^2$, which yields
\begin{equation}
\label{variance}
\begin{split}
& \sigma\boldsymbol{(}W(\mathcal{P}^{(L)},\mathcal{N}_0)\boldsymbol{)}^2 \\
 & = \sum_{l=0}^{L-1}\sum_{n=1}^{N} [\Delta h_{n}(l)]^{2} 
 G_n\boldsymbol{(}h(l\Delta x)\boldsymbol{)} \Delta x^2 \\
& \quad-\sum_{l=0}^{L-1}\Big[\sum_{n=1}^{N} \Delta h_{n}(l)  G_{n}\boldsymbol{(}h(l\Delta x)\boldsymbol{)} \Big]^{2}\Delta x^2.
\end{split}
\end{equation}
Assuming the function $h$ to be bounded and sufficiently smooth, one can see that $\sigma\boldsymbol{(}W(\mathcal{P}^{(L)},\mathcal{N}_0)\boldsymbol{)}$ tends to zero as $L\rightarrow\infty$. By combining this with Eq.~(\ref{expectation}), and Chebyshev's inequality (see e.g.~\cite{Gut}), it follows that 
\begin{equation*}
\lim_{L\rightarrow\infty}P\boldsymbol{(}|W(\mathcal{P}^{(L)},\mathcal{N}_0)-F(h^f)+F(h^i)|> \delta\boldsymbol{)} = 0,
\end{equation*}
for all $\delta>0$.
In other words, $W(\mathcal{P}^{(L)},\mathcal{N}_0)$ converges in probability \cite{Gut} to the free energy difference.
We can conclude that ITR processes yield an essentially deterministic work cost.

By the above discussion we can conclude the following lemma:
\begin{Lemma}
\label{LemmaITR}
Let $h^{i},h^{f}\in\mathbb{R}$, and let $\mathcal{N}$ be a random variable which is distributed $G(h^{i})$. Let $\delta>0$ and $0<\epsilon \leq 1$.  Then there exists a process $\mathcal{P}\in\mathscr{P}(h^{i},h^{f})$ such that
\begin{equation}
P\boldsymbol{(}|W(\mathcal{P},\mathcal{N})- F(h^{f})+F(h^{i})|\leq \delta\boldsymbol{)} > 1-\epsilon.
\end{equation}
\end{Lemma}
In a terminology that we shall introduce later, this lemma states that for every $\epsilon$ and $\delta$, there exists a process $\mathcal{P}$ for which $F(h^f)-F(h^{i})$ is an $(\epsilon,\delta)$-deterministic value of the random variable $W(\mathcal{P},\mathcal{N})$.


\section{ \label{Sec:OptimalExpected}Optimal expected work extraction}

We shall here derive the minimal expected work cost (and thus the maximal expected work yield) of any process that transforms an energy level configuration $h^i$ into $h^f$, with the state of the system initially distributed as $q$. (We make no restrictions on the final state, nor its distribution.) In other words, we use the expectation value as a cost function, and we search over all elements $\mathcal{P}$ in $\mathscr{P}(h^i,h^f)$  in order to minimize $\langle W(\mathcal{P},\mathcal{N})\rangle$, where the initial state $\mathcal{N}$ is distributed $q$.
The resulting infimum can be expressed in terms of the relative Shannon entropy \cite{CoverThomas} (also called the Kullback-Liebler divergence \cite{KLdiv})
\begin{equation}
D(q\Vert p) := \sum_{n}q_n \log_2 q_n - \sum_{n}q_n\ln p_n,
\end{equation}
 which in some sense measures the difference between two probability distributions $q,p\in\mathbb{P}(N)$. Here, $\log_2$ denotes the base-2 logarithm. (The relative entropy should not to be confused with the \emph{conditional} Shannon/von Neumann entropy, as used in e.g.~\cite{delRio}, which emerges in settings where we have access to side-information.)

\begin{Definition}
\label{nbydlfnyd}
Let $h^i,h^f\in \mathbb{R}^{N}$, and let $\mathcal{N}$ be a random variable with distribution $q\in\mathbb{P}(N)$, then we define
\begin{equation}
\mathcal{C}^{\mathrm{extr}}(q,h^{i},h^{f}) :=\inf_{\mathcal{P}\in\mathscr{P}(h^i,h^f)}\langle W(\mathcal{P},\mathcal{N})\rangle.
\end{equation}
and in the special case $h:= h^i = h^f$ we define the expected work content as
\begin{equation}
\begin{split}
\mathcal{A}(q,h):= & -\mathcal{C}^{\mathrm{extr}}(q,h,h) \\
= & \sup_{\mathcal{P}\in\mathscr{P}(h^i,h^f)}\langle -W(\mathcal{P},\mathcal{N})\rangle.
\end{split}
\end{equation}
\end{Definition}

\begin{Proposition}
\label{OptimalExpectedWork}
Let $h^i,h^f\in \mathbb{R}^{N}$, and let $q\in\mathbb{P}(N)$, then
\begin{equation}
\begin{split}
\mathcal{C}^{\mathrm{extr}}(q,h^{i},h^{f})  = & F(h^f)-F(h^i) \\
& -kT\ln(2)D\boldsymbol{(}q\Vert G(h^i)\boldsymbol{)}.
\end{split}
\end{equation}
\end{Proposition}
Note that since $h^i\in\mathbb{R}^N$ it follows that $G(h^i)$ has full support, i.e., there is no $n$ for which $G_{n}(h^i) = 0$. Hence, $D\boldsymbol{(}q\Vert G(h^i)\boldsymbol{)}< +\infty$.

A direct corollary of Proposition \ref{OptimalExpectedWork} is
\begin{equation}
\label{OptimalExpectedExtraction}
\begin{split}
\mathcal{A}(q,h) = kT\ln(2)D\boldsymbol{(}q\Vert G(h)\boldsymbol{)}.
\end{split}
\end{equation}
Within our model we can thus re-derive this well known result concerning the work content of non-equilibrium systems \cite{Procaccia,Lindblad1983,Takara,Esposito}.

In the case of  complete degeneracy,  $h = (r,\ldots, r)$, Eq.~(\ref{OptimalExpectedExtraction}) reduces to $\mathcal{A}(q,h)  =   [\log_2 N- H(q)]kT\ln 2$, where $H(q) := -\sum_{n=1}^{N}q_n\log_2 q_n$ is the Shannon entropy \cite{CoverThomas}. 

As a side remark we note that the quantity
\begin{equation}
\label{NonequiFreeEnergy}
\begin{split}
F(q,h) := &  F(h) + \frac{1}{\beta}\ln(2)D\boldsymbol{(}q\Vert G(h)\boldsymbol{)} \\
= & -\frac{1}{\beta}\ln(2)H(q)  +\sum_{n=1}^{N}q_n  h^i_{n}
\end{split}
\end{equation}
can be viewed as a non-equilibrium generalization of the free energy. It is not uncommon to refer to this more general quantity as `free energy'. However, in this investigation we reserve the term `free energy' for the equilibrium quantity $F(h) = -kT\ln Z(h)$.

\begin{proof}[proof of Proposition \ref{OptimalExpectedWork}]

We first shall show that 
\begin{equation}
\label{ncxmnxn}
\begin{split}
\langle W(\mathcal{P},\mathcal{N})\rangle \geq & F(h^f)-F(h^i) \\
& -kT\ln(2)D\boldsymbol{(}q\Vert G(h^i)\boldsymbol{)},
\end{split}
\end{equation}
for all $\mathcal{P}\in \mathscr{P}(h^i,h^f)$, where  $\mathcal{N}$ has distribution $q$.
Define
\begin{equation}
\label{overlWdef}
\tilde{W} : =  F(h^f)-F(h^i) -kT\ln\frac{q_{\mathcal{N}}}{G_{\mathcal{N}}(h^i)}.
\end{equation}
One can see that 
\begin{equation}
\label{idealworkcost}
\langle \tilde{W}\rangle  =  F(h^f)-F(h^i) -kT\ln(2) D\boldsymbol{(}q\Vert G(h^i)\boldsymbol{)}.
\end{equation}
Next, let $\mathcal{P}\in \mathscr{P}(h^i,h^f)$. Without loss of generality we may assume that this is a sequence of alternating LTs and thermalizations, beginning with an LT. Hence, we have a sequence of sets of energy levels $h^{0},h^{1},\ldots,h^{L}$, where $h^0 := h^i$ and $h^L := h^f$. The process proceeds with an LT that takes $h^{l}$ to $h^{l+1}$, operating on the state $\mathcal{N}_{l}$, followed by a thermalization, resulting in the new state $\mathcal{N}_{l+1}$ which is distributed $G(h^{l+1})$. Furthermore $\mathcal{N}_{0}:= \mathcal{N}$.
The work cost of this process is 
\begin{equation}
\label{ndvlkavlnk}
W(\mathcal{P},\mathcal{N}) = h^1_{\mathcal{N}}-h^{0}_{\mathcal{N}} + \sum_{l=1}^{L-1}(h^{l+1}_{\mathcal{N}_l} -h^{l}_{\mathcal{N}_l}).
\end{equation}
If we now make use of the general relation 
\begin{equation}
\label{mfnxkdjng}
h_{n} = F(h)-\frac{1}{\beta}\ln G_{n}(h), 
\end{equation}
in combination with Eqs.~(\ref{overlWdef}) and (\ref{ndvlkavlnk}) the result is 
\begin{equation}
\label{WminusWt}
\begin{split}
& W(\mathcal{P},\mathcal{N}) - \tilde{W} =  \frac{1}{\beta}\ln q_{\mathcal{N}}  -\frac{1}{\beta}\ln G_{\mathcal{N}}(h^1)\\
& \quad \quad  +\frac{1}{\beta} \sum_{l=1}^{L-1}\Big(\ln G_{\mathcal{N}_l}(h^{l}) -\ln G_{\mathcal{N}_l}(h^{l+1}) \Big).
\end{split}
\end{equation}
Hence,
\begin{equation}
\label{myvcmb}
\begin{split}
 \langle W(\mathcal{P},\mathcal{N}) - \tilde{W} \rangle 
= & \frac{1}{\beta}D\boldsymbol{(}q\Vert G(h^1)\boldsymbol{)} \\
& + \frac{1}{\beta}\sum_{l=1}^{L-1}D\boldsymbol{(}G(h^{l})\Vert G(h^{l+1})\boldsymbol{)}.
\end{split}
\end{equation}
Due to the fact that the relative Shannon entropy is non-negative, $D(p\Vert r)\geq 0$, we thus find  $\langle W (\mathcal{P},\mathcal{N})\rangle \geq \langle  \tilde{W} \rangle$. Combined with Eq.~(\ref{idealworkcost}) this yields Eq.~(\ref{ncxmnxn}).

Next we show that there exists a sequence of processes $(\mathcal{P}_{m})_{m\in\mathbb{N}}$ with $\mathcal{P}_{m}\in \mathscr{P}(h^i,h^f)$, such that 
\begin{equation}
\label{dtzkjdk}
\begin{split}
\lim_{m\rightarrow\infty }\langle W(\mathcal{P}_m,\mathcal{N})\rangle = & F(h^f)-F(h^i)\\
& -kT\ln(2)D\boldsymbol{(}q\Vert G(h^i)\boldsymbol{)}. 
\end{split}
\end{equation}
 For each $m\in\mathbb{N}$, we  
  define $h'_{n} := -kT\ln q_{n}$ for all $n$ such that $q_n\neq 0$, and $h'_n := m$ otherwise. Let $\mathcal{P}^{(1)}_m$ be the LT that takes $h^i$ to $h'$. This process has the expected work cost $\langle W(\mathcal{P}^{(1)}_m,\mathcal{N})\rangle  = kT\ln(2) H(q) -\sum_{n}q_n h^i_n$. The free energy of $h'$ is $F(h') = -kT\ln(1 + N_0e^{-\beta m})$, where $N_0$ is the number of energy levels for which the initial distribution assigns zero probability, $q_k=0$. After the initial LT the system is thermalized, leading to a new state $\mathcal{N}'$, distributed $G(h')$. By Lemma \ref{LemmaITR} there exists a finite process $\mathcal{P}^{(2)}_m$ that takes $h'$ to $h^f$,  and is such that $|\langle W(\mathcal{P}^{(2)}_m,\mathcal{N}')\rangle-F(h^f)+F(h')|\leq 1/m$, where $\mathcal{N}'$ is distributed $G(h')$. We let $\mathcal{P}_{m}$ be the concatenation of the initial LT $\mathcal{P}_{m}^{(1)}$, the thermalization, and $\mathcal{P}^{(2)}_m$. One can see (e.g., using Eq.~(\ref{NonequiFreeEnergy})) that $\lim_{m\rightarrow\infty}\langle W(\mathcal{P}_m,\mathcal{N})\rangle = F(h^f)-F(h^i)    -kT\ln 2D\boldsymbol{(}q\Vert G(h^i)\boldsymbol{)}$.

\end{proof}


\section{ \label{Sec:OptimalWorkCostErasure}Minimal expected work cost of erasure}
To erase a system is to put it in a well defined and pre-determined state. In our model this corresponds to transforming the system into a specific selected state $s \in\{1,\ldots,N\}$. According to Landauer's erasure principle \cite{Landauer61, Bennett82, LeffRexI, LeffRexII, PlenioVitelli01, Maruyama09} there is a minimal work cost associated with this erasure. We shall here  re-derive the  `standard' work cost of erasure, via the expectation value as a cost function. 

Similarly as for the work extraction, we assume  initial energy levels $h^i$,  final levels $h^f$, and an initial state $\mathcal{N}$ with distribution $q$. However, in addition we require the process to put the system into the selected state $s$. Formulated like this, one realizes that the work extraction task and the information erasure task are closely related. 
The erasure problem is nothing but the work extraction setup, but with an added constraint on the final distribution of the state. 

It is a  bit too strict to demand that the erasure process puts the system in state $s$ with certainty. Such a process exists within our model only if the initial distribution is  $q_n = \delta_{n,s}$. The reason for this is that the only method by which we can change the state of the system is by thermalizing it, in which case the system is put in the Gibbs distribution $G(h)$ for some collection of energy levels $h$. However, there exists no $h\in\mathbb{R}^{N}$ such that $G_{n}(h) = \delta_{n,s}$.  For this reason we only require the process $\mathcal{P}$ to result in final states $\mathcal{F}(\mathcal{P},\mathcal{N})$ such that $P\boldsymbol{(}\mathcal{F}(\mathcal{P},\mathcal{N}) = s\boldsymbol{)}\geq 1-\tau$, for $\tau>0$. After the optimization we take the limit $\tau\rightarrow 0$.
We formalize this idea in terms of the following set of operations:
\begin{Definition}
Let $h^i,h^f\in\mathbb{R}^{N}$ and let $\mathcal{N}$ be distributed $q\in\mathbb{P}(N)$, let $s\in\{1,\ldots, N\}$, and $0<\tau<1$. Define
\begin{equation}
\begin{split}
& \mathscr{P}^{\tau}_{s}(q, h^i,h^f)\\
&  := \{\mathcal{P}\in\mathscr{P}(h^i,h^f): P\boldsymbol{(}\mathcal{F}(\mathcal{P},\mathcal{N})=s\boldsymbol{)} \geq 1-\tau\}.
\end{split}
\end{equation}
\end{Definition}
The set $\mathscr{P}^{\tau}_{s}(q,h^i,h^f)$ depends on the initial distribution $q$ only in a very weak sense. The only aspect of $q$ that matters is if $q_{n} = \delta_{ns}$, or not. For the sake of completeness we nevertheless keep $q$ in the notation. 

\begin{Definition}
\label{nsfdbkvnbakn}
Let $h^i, h^f\in \mathbb{R}^{N}$, $q\in\mathbb{P}(N)$, and $s\in\{1,\ldots,N\}$. Let $\mathcal{N}$ be a random variable that is distributed $q$.
Then we define
\begin{equation}
\label{nadfklnbadfb}
\mathcal{C}^{\mathrm{erase}}(q,h^i,h^f,s) := \lim_{\tau \rightarrow 0^+}\inf_{\mathcal{P}\in\mathscr{P}^{\tau}_{s}(q,h^i,h^f) } \langle W(\mathcal{P},\mathcal{N})\rangle.
\end{equation}
\end{Definition}
Note that $\tau \geq \tau'>0$ implies $\mathscr{P}^{\tau}_{s}(q,h^i,h^f) \supseteq\mathscr{P}^{\tau'}_{s}(q,h^i,h^f)$. This in turn yields
\begin{equation*}
\inf_{\mathcal{P}\in\mathscr{P}^{\tau}_{s}(q,h^i,h^f) } \langle W(\mathcal{P},\mathcal{N})\rangle \leq \inf_{\mathcal{P}\in\mathscr{P}^{\tau'}_{s}(q,h^i,h^f) } \langle W(\mathcal{P},\mathcal{N})\rangle.
\end{equation*}
In other words, for decreasing $\tau$ the function $\inf_{\mathcal{P}\in\mathscr{P}^{\tau}_{s}(q,h^i,h^f) } \langle W(\mathcal{P},\mathcal{N})\rangle$ increases monotonically. Hence, the limit $\tau\rightarrow 0$ in Eq.~(\ref{nadfklnbadfb}) is well defined (possibly with the value $+\infty$).

\begin{Proposition}
\label{OptimalExpectedWorkCost}
Let $h^i, h^f\in \mathbb{R}^{N}$ and $q\in\mathbb{P}(N)$, and $s\in\{1,\ldots,N\}$. 
Then
\begin{equation}
\label{InfimumErasureCost}
\begin{split}
 \mathcal{C}^{\mathrm{erase}}(q,h^i,h^f,s)  = &   h^f_s  -F(h^i) \\
 &    -\frac{1}{\beta}\ln(2)D\boldsymbol{(}q\Vert G(h^i)\boldsymbol{)}.
\end{split}
\end{equation}
\end{Proposition}
In the special case of $h:=h^{i}= h^{f}$ for a completely degenerate set of energy levels, $h_n:= r$, we find the expected work cost of erasure to be 
\begin{equation}
\mathcal{C}^{\mathrm{erase}}(q,h,h,s) = kT\ln(2)H(q),
\end{equation}
which is a standard result concerning Landauer's erasure principle \cite{Piechocinska, Landauer61, Maruyama09, Benett03, Shizume, Shulman99,Alicki}.

\begin{proof}[proof of Proposition \ref{OptimalExpectedWorkCost}]
Regarding the erasure process as an alternating sequence of LTs and thermalizations, we first distinguish the special case that the process does not contain any thermalization.
In this case the process only consists of one single LT.  This LT must transform $h^i$ to $h^f$. Furthermore, an LT does not change the state of the system. Hence, for this LT to be an admissible process it follows that the initial distribution $q$ must be such that $q_{s}\geq 1-\tau$. Hence, in the limit $\tau\rightarrow 0$ we can only accept the initial distribution $q_n = \delta_{sn}$. In this limit, the resulting energy cost, $h^{f}_s-h^{i}_{s}$, agrees with Eq.~(\ref{InfimumErasureCost}).

We next consider the case that the process $\mathcal{P}\in\mathscr{P}^{\tau}_{s}(q,h^i,h^f)$ does contain at least one thermalization. Denote  
 $R: = h^f_s  -F(h^i)  -kT\ln(2)D\boldsymbol{(}q\Vert G(h^i)\boldsymbol{)}$.
First we shall prove that the left hand side of Eq.~(\ref{InfimumErasureCost}) is lower bounded by $R$. 
Let us divide $\mathcal{P}$ into two parts: The first part, $\mathcal{P}_1$, is the whole process up to the very last thermalization.  We let $\overline{h}$ be the configuration of energy levels at this final thermalization. The second part, $\mathcal{P}_2$, consists only of the final LT that transforms $\overline{h}$ into $h^f$. By Proposition~\ref{OptimalExpectedWork} we know that $\langle W(\mathcal{P}_1,\mathcal{N})\rangle \geq  F(\overline{h})-F(h^i)    -kT\ln(2)D\boldsymbol{(}q\Vert G(h^i)\boldsymbol{)}$.
The expected work cost of $\mathcal{P}_2$ is $\langle W(\mathcal{P}_2,\overline{\mathcal{N}})\rangle = \langle h^f_{\overline{\mathcal{N}}} - \overline{h}_{\overline{\mathcal{N}}}\rangle$,  where the state $\overline{\mathcal{N}}$  is Gibbs distributed $G(\overline{h})$. 
By combining the above observations we find
 \begin{eqnarray}
\langle W(\mathcal{P},\mathcal{N}) \rangle & =  &  \langle W(\mathcal{P}_1,\mathcal{N})\rangle + \langle W(\mathcal{P}_2,\overline{\mathcal{N}})\rangle \nonumber \\
& \geq & F(\overline{h})-F(h^i)    -\frac{\ln 2}{\beta}D\boldsymbol{(}q\Vert G(h^i)\boldsymbol{)} \nonumber\\
 \label{asjkd}& & + \langle h^f_{\overline{\mathcal{N}}}\rangle- \langle \overline{h}_{\overline{\mathcal{N}}}\rangle     
\end{eqnarray} 
Since $\mathcal{P}\in\mathscr{P}^{\tau}_{s}(q,h^i,h^f)$ we must have $G_s(\overline{h})\geq 1-\tau$. 
It follows that $\lim_{\tau\rightarrow 0}\langle h^f_{\overline{\mathcal{N}}}\rangle = h^f_s$. Furthermore,  $F(\overline{h}) - \langle \overline{h}_{\overline{\mathcal{N}}}\rangle = -kT\ln(2)H\boldsymbol{(}G(\overline{h})\boldsymbol{)}$, which goes to zero as $\tau\rightarrow 0$. We can conclude that 
$\lim_{\tau \rightarrow 0}\inf_{\mathcal{P}\in\mathscr{P}^{\tau}_{s}(q,h^i,h^f) } \langle W(\mathcal{P},\mathcal{N})\rangle \geq R$.

Next we shall find a sequence of processes $\mathcal{P}_m$ such that $\lim_{m\rightarrow\infty}P\boldsymbol{(}\mathcal{F}(\mathcal{P}_m,\mathcal{N})=s\boldsymbol{)} = 1$, and $\lim_{m\rightarrow\infty}\langle W(\mathcal{P}_m,\mathcal{N})\rangle = R$. Together with the lower bound we proved above, this implies Eq.~(\ref{InfimumErasureCost}).
 We construct each $\mathcal{P}_m$  as a concatenation of an initial LT $\mathcal{P}^{(1)}$, a thermalization, a process $\mathcal{P}_{m}^{(2)}$, and a final LT $\mathcal{P}^{(3)}_{m}$. 
 
We begin by constructing $\mathcal{P}_m^{(1)}$ as the LT that takes $h^i$ to a configuration of energy levels $h'$. The latter we define as $h'_{n} := -(\ln q_{n})/\beta$ for all $n$ such that $q_n\neq 0$, and $h'_n = m$ otherwise.   The expected work cost of this process acting on the initial state $\mathcal{N}$ is $\langle W(\mathcal{P}^{(1)}_m,\mathcal{N})\rangle  = \sum_nq_n(h'_n-h^{i}_n ) =  -F(h^i) -kT\ln(2)D\boldsymbol{(}q\Vert G(h^i)\boldsymbol{)}$. 

Next, the system is thermalized, and we let $\mathcal{N}'$ denote the state of the system after this thermalization. Define $h''_n := -m\delta_{n,s} + h^f_n$.
 By Lemma \ref{LemmaITR} there exists a process $\mathcal{P}^{(2)}_m\in \mathscr{P}(h',h'')$ such that $|\langle W(\mathcal{P}^{(2)}_m,\mathcal{N}')\rangle-F(h'')+F(h')|\leq 1/m$.  
Note that $F(h') = -kT \ln( 1 + N_0e^{-\beta m})$, where $N_0$ is the number of energy levels $k$ for which $q_k=0$. 
The final state of process $\mathcal{P}^{(2)}_m$ is $\mathcal{N}'' := \mathcal{F}(\mathcal{P}^{(2)}_m,\mathcal{N}')$, which is Gibbs distributed $G(h'')$. 
 
Finally, we let $\mathcal{P}_m^{(3)}$ be the LT that takes $h''$ to $h^f$. This process has the expected work cost $\langle W(\mathcal{P}^{(3)}_m,\mathcal{N}'')\rangle  =  m G_s(h'')$. 

Combining the above results we find $|\langle W(\mathcal{P}_m,\mathcal{N})\rangle -R|= |\langle W(\mathcal{P}_m^{(2)}, \mathcal{N}')\rangle -h^{f}_s  + mG_s(h'')| \leq m^{-1} + |F(h'')- F(h') -h^{f}_s  + mG_s(h'')|$. By writing out the terms explicitly, and using $\lim_{m\rightarrow \infty}m[1-G_s(h'')] = 0$, one can show that the right hand side of the above inequality converges to zero as $m\rightarrow \infty$.
Furthermore, $P\boldsymbol{(}\mathcal{F}(\mathcal{P}_m,\mathcal{N}) = s\boldsymbol{)} = G_s(h'')$. Since $\lim_{m\rightarrow\infty}G_s(h'')=1$, the proposition is proved. 
\end{proof}


\section{ \label{fluctuations}Intrinsic fluctuations in optimal expected work extraction}

Here we show that for any sequence of processes for which the expected work cost approaches the minimal value, as given by Proposition \ref{OptimalExpectedWork}, the resulting work cost variable converges in probability to a specific function of the initial state.

In the main text (and in Sec.~\ref{Sec:OptimalExpected}) we considered a specific process (or rather limit process) that yields the optimal expected work extraction. Recall that this process proceeds by first changing the Hamiltonian $h$ to a new $h'$ such that the initial distribution $q$ becomes the Gibbs distribution of this new hamiltonian $q = G(h')$ (for the moment disregarding the special case of distributions that do not have full support). The random work cost of this initial step takes the value $F(h')  -   F(h) -  kT\ln\boldsymbol{(}q_n/G_n(h)\boldsymbol{)}$ with probability $q_n$. We can next find a family of processes that arbitrarily well approximates an ITR that brings back the new Hamiltonian to the original. This last step adds a random variable that converges in probability to the constant value $F(h)-F(h')$.
For this specific choice of process we can thus conclude that the work cost is clustered around the values $-kT\ln\boldsymbol{(} q_{n}/G_n(h)\boldsymbol{)}$, each carrying a `cloud' of values around them. The purpose of this section is to prove that this feature is not restricted to this specific choice of process, but is generally true for any sequence of processes that achieves optimal expected work extraction.

\subsection{Convergence in probability}

Recall that a sequence of random variables $X_k$ converges to $X$ in probability if for each $\delta>0$ it is true    that $\lim_{k\rightarrow \infty}P(|X_k-X| > \delta) = 0$ \cite{Gut}.

Given a real valued random variable $X$, the cumulative distribution function is defined as $F_{X}(x) := P(X\leq x)$. The moment generating function of $X$, if it exists, is defined as $M_{X}(t) := \langle e^{tX}\rangle = \int_{-\infty}^{\infty}e^{tx}dF_{X}(x)$.

Given a random variable $X$ and a sequence of random variables $X_k$, a well known result by Curtis \cite{Curtis} states that if the moment generating function $M_{X_k}(x)$ exists and converges pointwise to $M_{X}(x)$ in a neighborhood of $0$ along the real axis, then $X_k$ converge to $X$ in distribution, i.e., $F_{X_k}(x)$ converges to $F_{X}(x)$ for each $x$ where $F_{X}$ is continuous. This standard result is unfortunately not quite enough for our purpose. We need a generalization \cite{Mukherjea,Ushakov} where the interval does not contain $0$. 
\begin{Proposition}[\cite{Mukherjea,Ushakov}]
\label{Indistribution}
Let $0< a<b$. If 
\begin{equation}
\lim_{k\rightarrow\infty}M_{X_k}(x) = M_{X}(x),\quad \forall x\in (a,b),
\end{equation}
then
\begin{equation}
\lim_{k\rightarrow\infty}F_{X_k}(x) = F_{X}(x),
\end{equation}
for each $x$ where $F_{X}$ is continuous.
\end{Proposition}

Given distributions $q\in\mathbb{P}(N)$ and $r\in \mathbb{P}^{+}(N)$ and a real number $\alpha>0$, $\alpha \neq 1$, the relative R\'enyi  $\alpha$-entropy is defined as \cite{Renyi}
\begin{equation}
D_{\alpha}(q\Vert r) := \frac{1}{\alpha-1}\log_{2}\sum_{j=1}^{N}\frac{q_j^{\alpha}}{r_j^{\alpha-1}}.
\end{equation}
Furthermore, we can let $D_{1}(q\Vert r):= D(q\Vert r)$, since $\lim_{\alpha\rightarrow 1}D_{\alpha}(q\Vert r)  = D_{1}(q\Vert r)$ \cite{Renyi}. It is furthermore the case that $\lim_{\alpha\rightarrow 0^{+}}D_{\alpha}(q\Vert r) = D_{0}(q\Vert r) := -\log_2\sum_{n:q_n>0}r_n$ \cite{Erven}.
The quantity $D_{\alpha}(q\Vert r)$ is non-decreasing in $\alpha$, and furthermore non-negative, 
\begin{equation}
\label{increase}
0 \leq D_{\alpha}(q\Vert r) \leq D_{\beta}(q\Vert r),\quad 0<\alpha \leq \beta.
\end{equation}
The latter can be seen by the fact that $D_{\beta}$ can be expressed in terms of the power mean (also called H\"older mean). If $w\in\mathbb{P}(N)$, and $x =(x_1,\ldots,x_N)$ and $s\neq 0$, then the power mean is defined as $M_{s}(w,x): = (\sum_{k}w_kx_k^{s})^{1/s}$ \cite{Bullen}. Furthermore,  $M_0(w,x) := \Pi_{k}x_k^{w_k}$. One can see that $D_{\alpha}(q\Vert r) =  \log_2 M_{\alpha-1}(q,q/r)$, where $(q/r)_{n} := q_{n}/r_n$. Since the power mean is monotonically increasing in $s$ \cite{Bullen}, Eq.~(\ref{increase}) follows.

\begin{Proposition}
\label{IntrinsicFluctuations}
Let $h^{i},h^{f}\in\mathbb{R}^{N}$, and let $\mathcal{N}$ be a random variable with distribution $q\in\mathbb{P}(N)$.
If  $(\mathcal{P}_{m})_{m\in\mathbb{N}}$ with $\mathcal{P}_{m}\in \mathscr{P}(h^i,h^f)$ is such that 
\begin{equation}
\begin{split}
\lim_{m\rightarrow\infty }\langle W(\mathcal{P}_m,\mathcal{N})\rangle = & F(h^f)-F(h^i)\\
& -kT\ln(2)D\boldsymbol{(}q\Vert G(h^i)\boldsymbol{)}, 
\end{split}
\end{equation}
then 
\begin{equation*}
W(\mathcal{P}_m,\mathcal{N}) \rightarrow F(h^f)-F(h^i) -kT\ln\frac{q_{\mathcal{N}}}{G_{\mathcal{N}}(h^i)}
\end{equation*}
in probability.
\end{Proposition}
In the special case $h := h^{f} = h^{i}$, we thus find that the work yield, $-W(\mathcal{P}_m,\mathcal{N})$, converges in probability to 
\begin{equation}
\label{fluctyield}
W_{\textrm{yield}}(h,\mathcal{N}) := kT\ln q_{\mathcal{N}} - kT\ln G_{\mathcal{N}}(h).
\end{equation}

\begin{proof}
In the following we let 
\begin{equation}
\label{overlWdef2}
\tilde{W} : =  F(h^f)-F(h^i) -kT\ln\frac{q_{\mathcal{N}}}{G_{\mathcal{N}}(h^i)}.
\end{equation}
Let $(\mathcal{P}_{m})_{m\in\mathbb{N}}$ be any sequence in  $\mathscr{P}(h^i,h^f)$ such that 
$\lim_{m\rightarrow\infty }\langle W(\mathcal{P}_m,\mathcal{N})\rangle = F(h^f)-F(h^i)-kT\ln(2)D\boldsymbol{(}q\Vert G(h^i)\boldsymbol{)}$. Define 
\begin{equation*}
\begin{split}
a_m:=  & \beta \langle W(\mathcal{P}_{m},\mathcal{N}) \rangle \\
& -\beta F(h^f)+\beta F(h^i)+\ln(2)D\boldsymbol{(}q\Vert G(h^i)\boldsymbol{)}.
\end{split}
\end{equation*}
Hence, $\lim_{m\rightarrow\infty}a_m = 0$.
Furthermore, by Eqs.~(\ref{idealworkcost}) and (\ref{WminusWt}) we know that 
\begin{equation*}
\begin{split}
a_m 
= & \frac{1}{\beta}D\boldsymbol{(}q\Vert G(h^{m,1})\boldsymbol{)} \\
& + \frac{1}{\beta}\sum_{l=1}^{L_m-1}D\boldsymbol{(}G(h^{m,l})\Vert G(h^{m,l+1})\boldsymbol{)},
\end{split}
\end{equation*}
where, for each $m$,  $(h^{m,l})_{l=0}^{L_m}$ is the sequence of energy level configurations in $\mathcal{P}_{m}$. Define the random variable 
\begin{equation}
X_m := -\beta[W(\mathcal{P}_{m},\mathcal{N})-  \tilde{W} ].
\end{equation}
It follows that the cumulant generating function of $X_m$ is 
\begin{equation*}
\begin{split}
\log_2 M_{X_m}(t)  = & -tD_{1-t}\boldsymbol{(}q\Vert G(h^{m,1})\boldsymbol{)}\\
  & -t\sum_{l=1}^{L_m} D_{1-t}\boldsymbol{(}G(h^{m,l})\Vert G(h^{m,l+1})\boldsymbol{)}, 
\end{split}
\end{equation*}
where we know that $M_{X_m}(t)$ exists, since $X_m$ only can take a finite number of values for each $m$. By the monotonic increase of $D_{\alpha}$ with respect to $\alpha$ (Eq.~(\ref{increase})), and $D_{\alpha}\geq 0$, it follows that
\begin{equation}
0\geq \log_2 M_{X_m}(t) \geq -ta_m,\quad \forall t \in (0,1).
\end{equation}
Thus
\begin{equation}
\label{nlkvn}
\lim_{m\rightarrow \infty} M_{X_m}(t)=1,\quad \forall t \in (0,1).
\end{equation}
The constant $X \equiv 0$ has the moment generating function $M_{X}(t) \equiv 1$.
Hence, according to Proposition \ref{Indistribution}, Eq.~(\ref{nlkvn}) implies that $X_n$ converges in distribution to $X$.  Since the cumulative distribution function of $X$ is the step function, the convergence in distribution yields 
\begin{equation}
\label{xcnvkmlm}
\lim_{m\rightarrow\infty}P(X_{m} \leq x) = \left\{\begin{matrix}
0,\quad x < 0,\\
1,\quad x > 0.
\end{matrix} \right.
\end{equation}
(We do not care about $x=0$ as it is a point of discontinuity.)
A direct consequence of Eq.~(\ref{xcnvkmlm}) is
\begin{equation*}
\lim_{m\rightarrow \infty}P\boldsymbol{(}|W(\mathcal{P}_{m},\mathcal{N})-  \tilde{W}| > \delta\boldsymbol{)} = 0,\quad \forall  \delta>0.
\end{equation*}
In other words, $W(\mathcal{P}_{m},\mathcal{N})$ converges in probability to $\tilde{W}$.
\end{proof}


\subsection{\label{StandardDeviation} Standard deviation}
In Sec.~\ref{comparisons} we will compare the expected work content $\mathcal{A}(q,h)$ with the fluctuations in the optimal expected work extraction. Since we just have proved that the random work yield of any family of processes that achieves the optimal expected work extraction, converges in probability to the variable
\begin{equation}
\tilde{W} := F(h^f)-F(h^i) -kT\ln q_{\mathcal{N}} + kT\ln G_{\mathcal{N}}(h^i),
\end{equation}
it would be very convenient to use the standard deviation of $\tilde{W}$ (or of $W_{\textrm{yield}}$ in Eq.~(\ref{fluctyield})) for the comparison. 

Here we prove that for all sequences of processes that achieves the optimal expected work extraction, the minimal amount of noise is determined by the standard deviation of $\tilde{W}$. 
To this end we define $\sigma(X) := \sqrt{\langle X^2\rangle-\langle X\rangle^2}$, and note that 
\begin{equation}
\sigma(\tilde{W}) = kT\ln(2)\sigma\boldsymbol{(}q\Vert G(h^{i})\boldsymbol{)}
\end{equation}
where
\begin{equation}
\sigma(q\Vert r)^2 : = \sum_{n}q_n(\log_{2}\frac{q_n}{r_n})^2 - (\sum_{n}q_n\log_2\frac{q_n}{r_n})^2.
\end{equation}
(We discus this quantity further in Sec.~\ref{comparisons}.)

Here we prove (Proposition \ref{MinStdv}) that for all sequences of processes for which the expected work cost converges to $C(q,h^i,h^f)$, the resulting variance can in the limit not be smaller than $\sigma(\tilde{W})$. Furthermore (Proposition \ref{ExistMinStdv}) there exists a sequence of processes that simultaneously achieves the expected work cost $C(q,h^i,h^f)$ and the variance $\sigma(\tilde{W})$. 

To prove this we make use of the following two theorems, taken from \cite{Gut}.
\begin{Theorem}[Theorem 10.3 in \cite{Gut}.]
\label{ithysitkt}
Let $X$ and $(X_n)_{n\in\mathbb{N}}$ be random variables such that $X_n\rightarrow X$ in probability, and suppose that $h$ is a continuous function. Then $h(X_n)\rightarrow h(X)$ in probability. 
\end{Theorem}

\begin{Theorem}[Theorem 5.3 in \cite{Gut}.]
\label{xbnnbx}
Let $X$ and $(X_n)_{n\in\mathbb{N}}$ be random variables, and suppose that $X_{n}\rightarrow X$ in probability, then 
\begin{equation}
\langle |X|\rangle \leq \liminf_{n\rightarrow \infty} \langle |X_n|\rangle. 
\end{equation}
\end{Theorem}

\begin{Proposition}
\label{MinStdv}
Let $h^{i},h^{f}\in\mathbb{R}^{N}$, and let $\mathcal{N}$ be a random variable with distribution $q\in\mathbb{P}(N)$.
If  $(\mathcal{P}_{m})_{m\in\mathbb{N}}$ with $\mathcal{P}_{m}\in \mathscr{P}(h^i,h^f)$ is such that 
\begin{equation}
\label{yncykn}
\begin{split}
\lim_{m\rightarrow\infty }\langle W(\mathcal{P}_m,\mathcal{N})\rangle = & F(h^f)-F(h^i)\\
& -kT\ln(2)D\boldsymbol{(}q\Vert G(h^i)\boldsymbol{)}, 
\end{split}
\end{equation}
then 
\begin{equation}
\label{wqwefee}
\liminf_{m\rightarrow\infty } \sigma\boldsymbol{(}W(\mathcal{P}_m,\mathcal{N})\boldsymbol{)} \geq  kT\sigma\boldsymbol{(}q\Vert G(h^i)\boldsymbol{)}.
\end{equation}
\end{Proposition}
\begin{proof}
By combining Eq.~(\ref{yncykn}), Proposition \ref{IntrinsicFluctuations}, and Theorem \ref{ithysitkt} we can conclude that 
\begin{equation*}
W(\mathcal{P}_m,\mathcal{N})^2 \rightarrow \Big( F(h^f)-F(h^i) -kT\ln\frac{q_{\mathcal{N}}}{G_{\mathcal{N}}(h^i)}\Big)^2
\end{equation*}
in probability. Theorem \ref{xbnnbx} yields
\begin{equation*}
\Big\langle  \Big[ F(h^f)-F(h^i) -kT\ln\frac{q_{\mathcal{N}}}{G_{\mathcal{N}}(h^i)}\Big]^2 \Big\rangle \leq \liminf_{n\rightarrow \infty} \langle W(\mathcal{P}_m,\mathcal{N})^2 \rangle. 
\end{equation*}
(The variances $\sigma\boldsymbol{(}W(\mathcal{P}_m,\mathcal{N})\boldsymbol{)}$  and $\sigma(\tilde{W})$ exist, as these random variables only take a finite number of values.)
Combined with Eq.~(\ref{yncykn}) it follows that Eq.~(\ref{wqwefee}) holds.
\end{proof}

\begin{Proposition}
\label{ExistMinStdv}
Let $h^{i},h^{f}\in\mathbb{R}^{N}$, and let $\mathcal{N}$ be a random variable with distribution $q\in\mathbb{P}(N)$.
Then there exists  a sequence of processes $(\mathcal{P}_{m})_{m\in\mathbb{N}}$ with $\mathcal{P}_{m}\in \mathscr{P}(h^i,h^f)$ such that 
\begin{equation}
\begin{split}
\lim_{m\rightarrow\infty }\langle W(\mathcal{P}_m,\mathcal{N})\rangle = & F(h^f)-F(h^i)\\
& -kT\ln(2)D\boldsymbol{(}q\Vert G(h^i)\boldsymbol{)}, 
\end{split}
\end{equation}
and 
\begin{equation}
\lim_{m\rightarrow\infty } \sigma\boldsymbol{(}W(\mathcal{P}_m,\mathcal{N})\boldsymbol{)} =  kT\sigma\boldsymbol{(}q\Vert G(h^i)\boldsymbol{)}.
\end{equation}
\end{Proposition}
The proof of the above proposition uses the sequence of processes constructed in the proof of Proposition \ref{OptimalExpectedWork} (more precisely the proof of Eq.~(\ref{dtzkjdk})), with the addition that we use the fact that we always can find an approximate ITR with arbitrarily small standard deviation.


\section{\label{Properties} $(\epsilon,\delta)$-deterministic values}

Here we construct a cost function that favors work cost variables that have a sufficiently narrow distribution, i.e, for which we are more or less certain of what the work cost will be.
\begin{Definition}
Given a  real-valued random variable $X$, and $0 < \epsilon \leq 1$ and $0 \leq \delta<+\infty$, we say that $x\in\mathbb{R}$ is an $(\epsilon,\delta)$-deterministic value of $X$ if $P(|X -x| \leq \delta)> 1-\epsilon$. 
We denote the set of all $(\epsilon,\delta)$-deterministic values of $X$ as 
\begin{equation}
\Delta_{\delta}^{\epsilon}(X) := \{x\in\mathbb{R}: P( |X-x| \leq  \delta)> 1-\epsilon\}.
\end{equation}
\end{Definition}
Note that it may very well be the case $X$ does not have any $(\epsilon,\delta)$-deterministic value, or that it has more than one  $(\epsilon,\delta)$-deterministic value.

As we apply this concept, the random variable $X$ will  typically be the work cost of a given process, i.e., $W(\mathcal{P},\mathcal{N})$. The set $\Delta_{\delta}^{\epsilon}\boldsymbol{(}W(\mathcal{P},\mathcal{N})\boldsymbol{)}$ corresponds to work values around which the distribution is sufficiently peaked. Since our goal is to find the minimal work cost, we select the `smallest' of these sufficiently concentrated work costs, or more precisely, the infimum $\inf\!\Delta_{\delta}^{\epsilon}\boldsymbol{(}W(\mathcal{P},\mathcal{N})\boldsymbol{)}$. The quantity $\inf\!\Delta_{\delta}^{\epsilon}(X)$ will serve as the cost function that defines what, e.g.,  $(\epsilon,\delta)$-deterministic work content is (Sec~\ref{Sec:AlmDetWorkExtr}).
It is maybe worth pointing out that  $\inf\!\Delta_{\delta}^{\epsilon}\boldsymbol{(}W(\mathcal{P},\mathcal{N})\boldsymbol{)}$ is only the cost of one fixed process $\mathcal{P}$. To obtain the cost of extraction and information erasure we still need to minimize over the set of allowed processes.

In the following we establish some properties of the quantity $\inf\!\Delta_{\delta}^{\epsilon}(X)$ that will prove useful in the subsequent derivations.

We first note that $\inf\!\Delta_{\delta}^{\epsilon}(X) = +\infty$ if and only if $\Delta_{\delta}^{\epsilon}(X) = \emptyset$.  With the assumptions $0 < \epsilon <1$ and $0 \leq \delta <+\infty$ it follows that $-\infty <\inf\!\Delta_{\delta}^{\epsilon}(X)$.

The following two lemmas show that $ \inf\!\Delta_{\delta}^{\epsilon}(X)$ decreases monotonically with increasing $\epsilon$ and $\delta$.
\begin{Lemma}
\label{decreaseineps}
Let $X$ be a real-valued random variable, and let $0 < \epsilon \leq \epsilon' \leq 1$ and $0\leq \delta <+\infty$, then 
\begin{equation}
\label{nsdklvn}
\Delta_{\delta}^{\epsilon}(X) \subseteq \Delta_{\delta}^{\epsilon'}(X),
\end{equation}
and thus 
\begin{equation}
\label{ynkfdb}
 \inf\!\Delta_{\delta}^{\epsilon}(X) \geq  \inf\!\Delta_{\delta}^{\epsilon'}(X). 
\end{equation}
\end{Lemma}
\begin{proof}
If $ \inf\!\Delta_{\delta}^{\epsilon}(X) = +\infty$ the lemma is trivially true. Hence, without loss of generality we assume $\inf\!\Delta_{\delta}^{\epsilon}(X) < +\infty$, and thus $\Delta_{\delta}^{\epsilon}(X)$ is non-empty.  If $x\in \Delta_{\delta}^{\epsilon}(X)$, then $P(|X-x|\leq \delta) \geq 1- \epsilon$. Since  $\epsilon \leq \epsilon'$, this implies $P(|X-x|\leq \delta) \geq 1- \epsilon'$, and thus Eq.~(\ref{nsdklvn}) holds. This immediately implies Eq.~(\ref{ynkfdb}).
\end{proof}

\begin{Lemma}
\label{decreaseindelta}
Let $X$ be a real-valued random variable, and let  $0< \epsilon \leq 1$, and $0\leq \delta \leq \delta'<+\infty$, then 
\begin{equation}
\label{ndslvakn}
  \Delta_{\delta}^{\epsilon}(X)\subseteq  \Delta_{\delta'}^{\epsilon}(X),
\end{equation}
and thus 
\begin{equation}
\label{dfdlkfb}
 \inf\!\Delta_{\delta}^{\epsilon}(X) \geq  \inf\!\Delta_{\delta'}^{\epsilon}(X). 
\end{equation}
\end{Lemma}
\begin{proof}
If $ \inf\!\Delta_{\delta}^{\epsilon}(X) = + \infty$, the lemma is trivially true. We thus assume $ \inf\!\Delta_{\delta}^{\epsilon}(X) < +\infty$, and hence $\Delta_{\delta}^{\epsilon}(X)$ is non-empty. If $x\in \Delta_{\delta}^{\epsilon}(X)$, then $P(|X-x|\leq \delta) \geq 1- \epsilon$. Since $ 0\leq\delta \leq \delta' $ it follows that  $P(|X -x|\leq \delta') \geq P(|X-x| \leq \delta)$ for every $x\in\mathbb{R}$. From this we can conclude that Eq.~(\ref{ndslvakn}) holds, which also proves Eq.~(\ref{dfdlkfb}).
\end{proof}

\begin{Lemma}
\label{lowerbound}
Let $X$ be a real valued random variable,  and let $0 < \epsilon \leq \frac{1}{2}$ and $0 \leq\delta<+\infty$.
If there exists a real number $x$ such that 
\begin{equation}
\label{sndfklbvnfdb}
P( |X-x| \leq \delta)> 1-\epsilon
\end{equation}
then
\begin{equation}
\label{dfmngkl}
x-2\delta \leq  \inf\!\Delta_{\delta}^{\epsilon}(X) 
\end{equation}
\end{Lemma}
\begin{proof}
First note that Eq.~(\ref{sndfklbvnfdb}) implies $\Delta_{\delta}^{\epsilon}(X) \neq \emptyset$.
Suppose there is an $x'\in \Delta_{\delta}^{\epsilon}(X)$ such that  $|x-x'| > 2\delta$.  It follows that $\{z\in\mathbb{R}: |z-x|\leq \delta\} \cap \{z\in\mathbb{R}: |z-x'|\leq \delta\} = \emptyset$. Hence, $P(\{z\in\mathbb{R}: |z-x|\leq \delta\} \cap \{z\in\mathbb{R}: |z-x'|\leq \delta\}) = P(|X-x|\leq \delta) + P(|X-x'|\leq \delta) > 2-2\epsilon \geq 1$, which is a contradiction. 
Thus we must conclude  $|x-x'| \leq 2\delta$. Since this is true for all $x'\in \Delta_{\delta}^{\epsilon}(X)$ we can conclude that Eq.~(\ref{dfmngkl}) holds.
\end{proof}

\begin{Lemma}
\label{finite}
Let $X$ and $Y$ be two independent real-valued random variables, and let $0 < \epsilon < 1$ and $0\leq \delta<+\infty$. Then 
\begin{equation*}
  \Delta_{\delta}^{\epsilon}(X+Y) \neq \emptyset \,\,\,\Rightarrow\,\,\, \Delta_{\delta}^{\epsilon}(X) \neq \emptyset ,\,\,  \Delta_{\delta}^{\epsilon}(Y) \neq \emptyset.  
\end{equation*}
\end{Lemma}
\begin{proof}
The function $Q(X,\delta):= \sup_{s\in\mathbb{R}}P(|X-s|\leq \delta)$ is sometimes referred to as  Levy's  concentration function \cite{Petrov}. For two independent random variables $X$ and $Y$ it can be shown (see Lemma 1.11 in \cite{Petrov}) that  
\begin{equation}
\label{concentrfcn}
Q(X+Y,\delta) \leq \min[Q(X,\delta),Q(Y,\delta)].
\end{equation}
If we assume  $\Delta_{\delta}^{\epsilon}(X+Y) \neq \emptyset$ it implies that there exists a $z$ such that 
$P(|X+Y-z|\leq \delta)>1-\epsilon$. By Eq.~(\ref{concentrfcn}) we can thus conclude that 
$1-\epsilon < Q(X,\delta)$.
By the properties of the supremum, it follows that for every $\xi>0$ there exists an $x'$ such that $Q(X,\delta)-\xi < P(|X-x'|\leq \delta)$. Since $1-\epsilon <Q(X,\delta)$ it follows that we can find a $\xi>0$, and a corresponding $x'$, such that $1-\epsilon < Q(X,\delta)-\xi < P(|X-x'|\leq \delta)$. Thus, 
$x'\in  \Delta_{\delta}^{\epsilon}(X)$, and hence $ \Delta_{\delta}^{\epsilon}(X)\neq \emptyset$.
By an equivalent argument $ \Delta_{\delta}^{\epsilon}(Y)\neq \emptyset$.
\end{proof}

\begin{Lemma}
\label{mindetaddition}
Let $X$ and $Y$ be two independent real-valued random variables. Let $0 < \epsilon \leq 1-\frac{1}{\sqrt{2}}$ and $0\leq\delta <+\infty$. Then 
\begin{equation}
\label{nklfdb}
 \inf\!\Delta_{\delta}^{\epsilon}(X) +  \inf\!\Delta_{\delta}^{\epsilon}(Y) -4\delta\leq   \inf\!\Delta_{\delta}^{\epsilon}(X+Y).
\end{equation}
\end{Lemma}

\begin{proof}
First of all we note that the statement of the lemma is trivially true if $ \inf\!\Delta_{\delta}^{\epsilon}(X+Y) = +\infty$. Thus, without loss of generality we assume $ \inf\!\Delta_{\delta}^{\epsilon}(X+Y) < +\infty$. By  Lemma \ref{finite} this implies that that $ \Delta_{\delta}^{\epsilon}(X)\neq \emptyset$ and $ \Delta_{\delta}^{\epsilon}(Y)\neq \emptyset$. Thus there exist $x\in  \Delta_{\delta}^{\epsilon}(X)$ and $y\in  \Delta_{\delta}^{\epsilon}(Y)$. Since $X$ and $Y$ are independent it follows that 
\begin{equation*}
\begin{split}
& P(|X + Y-x-y| \leq  2\delta) \\
& \geq P(|X-x|\leq \delta)P(|Y-y|\leq \delta) \geq (1-\epsilon)^{2}.
\end{split}
\end{equation*}
Hence, $x+y$ is an $(2\epsilon - \epsilon^{2},2\delta)$-deterministic value of $X+Y$.
By assumption $\epsilon \leq 1-1/\sqrt{2}$ and thus $2\epsilon-\epsilon^{2}\leq  1/2$. Lemma \ref{lowerbound} yields  
\begin{equation}
\label{nvflnv}
x+y -4\delta \leq  \inf\!\Delta_{2\delta}^{2\epsilon-\epsilon^{2}}(X+Y).
\end{equation}
Since $x\in \Delta_{\delta}^{\epsilon}(X)$ it follows that $\inf\!\Delta_{\delta}^{\epsilon}(X) \leq x$, and similarly $\inf\!\Delta_{\delta}^{\epsilon}(Y) \leq y$. Combined with Eq.~(\ref{nvflnv}) this yields
\begin{equation}
\label{sfbjd}
 \inf\!\Delta_{\delta}^{\epsilon}(X) +  \inf\!\Delta_{\delta}^{\epsilon}(Y) -4\delta \leq  \inf\!\Delta_{2\delta}^{2\epsilon-\epsilon^{2}}(X+Y).
\end{equation}
Since $1\geq 2\epsilon-\epsilon^{2} \geq \epsilon>0$, Lemma \ref{decreaseineps} yields $\inf\!\Delta_{2\delta}^{2\epsilon-\epsilon^{2}}(X+Y) \leq  \inf\!\Delta_{2\delta}^{\epsilon}(X+Y)$.
By Lemma \ref{decreaseindelta} we furthermore find $ \inf\!\Delta_{2\delta}^{\epsilon}(X+Y)\leq  \inf\!\Delta_{\delta}^{\epsilon}(X+Y)$.
By combining Eq.~(\ref{sfbjd}) with the above observations we obtain Eq.~(\ref{nklfdb}). 
\end{proof}


\section{ \label{Sec:FandD0}The $\epsilon$-free energy and the smoothed relative R\'enyi $0$-entropy}
As we have seen in Sec.~\ref{Sec:OptimalExpected}, the minimal expected work cost of the extraction process can be expressed in terms of the generalized free energy $F(q,h)$, or equivalently, the relative Shannon entropy. In the $\epsilon$-deterministic setting, the roles of these measures are, as we shall see in Sec.~\ref{Sec:AlmDetWorkExtr}, taken over by two other quantities: the $\epsilon$-free energy and a smoothed relative R\'enyi $0$-entropy.

Given a distribution $q\in\mathbb{P}(N)$, and an event $\Lambda \subseteq \{1,\ldots, N\}$ we denote the probability of the event $\Lambda$ with respect to $q$ as
\begin{equation}
q(\Lambda):= \sum_{n\in\Lambda}q_n.
\end{equation} 

\begin{Definition}
Let $h\in\mathbb{R}^{N}$ and $\Lambda \subseteq \{1,\ldots, N\}$. We define 
the truncated partition function with respect to $\Lambda$ as
\begin{equation}
Z_{\Lambda}(h) := \sum_{n\in\Lambda}e^{-\beta h_n}. 
\end{equation}
Given $0 < \epsilon \leq 1$, and  $q\in\mathbb{P}(N)$ we define the $\epsilon$-free energy as 
\begin{equation}
F^{\epsilon}(q,h) := -\frac{1}{\beta}\ln\inf_{\Lambda: q(\Lambda)>1-\epsilon}Z_{\Lambda}(h)
\end{equation}
\end{Definition}
In other words, we make the free energy as large as possible by finding the smallest partition function truncated to a sufficiently likely event. Since the underlying set is finite, the infimum can be replaced by a minimum, i.e., there exists a sufficiently likely subset subset $\Lambda^{*}$ such that $F^{\epsilon}(q,h) = -kT\ln Z_{\Lambda^*}(h)$. Note that the concept of one-shot free energy has been introduced independently in \cite{Horodecki11}.

The  `standard'  relative R\'enyi $0$-entropy (see e.g.~\cite{Erven}) between two distributions $q$ and $p$  over $\{1,\ldots,N\}$  is defined as
\begin{equation}
D_{0}(q\Vert p) := -\log_{2}\sum_{j: q_{j}>0}p_{j},
\end{equation}
where we sum $p$ only over the support of $q$. 
We can obtain an $\epsilon$-smoothed relative R\'enyi $0$-entropy in a manner very similar to how we defined the $\epsilon$-free energy; we sum $p$ over the `best' sufficiently likely support: 
\begin{Definition}
Let $q,p\in\mathbb{P}(N)$ and $0 < \epsilon \leq 1$. We define
\begin{equation}
\label{relRzeroeps}
D^{\epsilon}_{0}(q\Vert p) :=  -\log_{2}\inf_{\Lambda: q(\Lambda)> 1-\epsilon}p(\Lambda).
\end{equation}
\end{Definition}
Up to some purely technical differences concerning the smoothing,  this entropy measure was introduced in \cite{Wang1}. (See also \cite{Datta,Wang2} for related measures in the quantum setting.) 

The relative entropy $D_{0}^{\epsilon}$ and the $\epsilon$-free energy are related as 
\begin{equation}
\label{RelationDeps0andFeps}
F^{\epsilon}(q,h) = F(h) + \frac{1}{\beta}\ln(2)D^{\epsilon}_{0}\boldsymbol{(}q\Vert G(h)\boldsymbol{)}.
\end{equation}
In what follows we shall switch freely between $F^{\epsilon}$ and $D^{\epsilon}_{0}$ without comment.
\begin{Lemma}
Let $h\in\mathbb{R}^N$, $0<\epsilon'\leq \epsilon \leq 1$, and $q\in\mathbb{P}(N)$, then
\begin{equation}
F^{\epsilon'}(q,h) \leq  F^{\epsilon}(q,h).
\end{equation}
\end{Lemma}
\begin{proof}
Let $\Lambda^*\subseteq\{1,\ldots,N\}$ be such that $F^{\epsilon'}(q,h)  = -kT\ln Z_{\Lambda^*}(h)$. Hence,
$q(\Lambda^*)  >1-\epsilon'\geq 1-\epsilon$. Thus, $Z_{\Lambda^*}(h) \geq \inf_{q(\Lambda)> 1-\epsilon}Z_{\Lambda}(h) $.
\end{proof}

\begin{Lemma}
\label{Fcontinuity}
Let $h\in \mathbb{R}^{N}$, and $q\in\mathbb{P}(N)$. Then the function $\epsilon \mapsto F^{\epsilon}(q,h)$ is left-continuous on $(0,1]$.
\end{Lemma}
\begin{proof}
Take an $\epsilon\in(0,1]$. We know that there exists a set $\Lambda^{*}\subseteq \{1,\ldots,N\}$ such that $F^{\epsilon}(q,h) = -kT\ln Z_{\Lambda^*}(h)$. By definition $q(\Lambda^*)>1-\epsilon$. Hence, there exists an $\overline{\epsilon}$ such that $q(\Lambda^*)> 1-\overline{\epsilon}>1-\epsilon$.
As a consequence, $\Lambda^*$ is also a minimizing set for $F^{\overline{\epsilon}}(q,h)$, i.e., $F^{\overline{\epsilon}}(q,h)= -kT\ln Z_{\Lambda^*}(h)$. Moreover, this is true for all $\overline{\epsilon}\in (1-q(\Lambda^*),\epsilon]$. Hence, for each $\epsilon$ there exists a left neighborhood to $\epsilon$, where $F^{\overline{\epsilon}}(q,h)$ is constant. This proves the lemma.
\end{proof}

One may wonder why we have defined $F^{\epsilon}(q,h)$ as $\sup_{q(\Lambda)>1-\epsilon}F_{\Lambda}$ and not $\sup_{q(\Lambda)\geq 1-\epsilon}F_{\Lambda}$, i.e., why a strict inequality rather than just an inequality? This is due to technicalities concerning the proofs of the upper bounds in Propositions \ref{LowerMax} and \ref{AlmDetErasure} (via Lemmas \ref{Fcontinuity}, \ref{ExistenceMindet}, and \ref{nklfbanlkb}). This also the reason why we similarly define $\Delta_{\delta}^{\epsilon}$ in terms of a strict inequality.


\section{ \label{Sec:AlmDetWorkExtr} Optimal $\epsilon$-deterministic work extraction}

In Secs.~\ref{Sec:OptimalExpected} and \ref{Sec:OptimalWorkCostErasure} we minimized the expectation value of the work cost variable for the given task. Instead of using the expectation value as the cost function, we here use the cost function $\inf\!\Delta_{\delta}^{\epsilon}$ introduced in Sec.~\ref{Properties}. In other words, we demand that the random work cost variable, or work yield variable, have a very high degree of predictability.

\begin{Definition}
\label{vnalksdjvn}
Let $h^{i},h^{f}\in\mathbb{R}^{N}$, $q\in\mathbb{P}(N)$, and let $\mathcal{N}$ be a random variable with distribution $q$. Let $\epsilon,\delta\in\mathbb{R}$ be such that $0 < \epsilon < 1$ and $0 \leq \delta <+\infty$. Define the $(\epsilon,\delta)$-deterministic cost of work extraction as
\begin{equation}
\label{DefCextrepsdelta}
\begin{split}
\mathcal{C}^{\mathrm{extr}}_{\epsilon,\delta}(q,h^{i},h^{f}) := & \inf\bigcup_{\mathcal{P}\in\mathscr{P}(h^{i},h^{f})}\Delta_{\delta}^{\epsilon}\boldsymbol{(}W(\mathcal{P},\mathcal{N})\boldsymbol{)}\\
= &   \inf_{\mathcal{P}\in\mathscr{P}(h^{i},h^{f})}\inf\!\Delta_{\delta}^{\epsilon}\boldsymbol{(}W(\mathcal{P},\mathcal{N})\boldsymbol{)}.
\end{split}
\end{equation} 
Furthermore, define the $\epsilon$-deterministic work cost of work extraction as
\begin{equation}
\label{DefCextr}
\begin{split}
\mathcal{C}^{\mathrm{extr}}_{\epsilon}(q,h^{i},h^{f})  : = \lim_{\delta \rightarrow 0^+}\mathcal{C}^{\mathrm{extr}}_{\epsilon,\delta}(q,h^{i},h^{f}). 
\end{split}
\end{equation} 
\end{Definition}
The two lines in Eq.~(\ref{DefCextrepsdelta}) merely reflects the fact that it is a matter of convenience whether we wish to view the problem as finding the infimum of the set of all $(\epsilon,\delta)$-deterministic values generated by all allowed processes, $\cup_{\mathcal{P}\in\mathscr{P}(h^{i},h^{f})}\Delta_{\delta}^{\epsilon}\boldsymbol{(}W(\mathcal{P},\mathcal{N})\boldsymbol{)}$, or whether we wish to view it as a minimization of the cost function $\inf\!\Delta_{\delta}^{\epsilon}$ over the set of processes $\mathscr{P}(h^{i},h^{f})$.

By combining the first line of Eq.~(\ref{DefCextrepsdelta}) in the above definition, with Eq.~(\ref{nsdklvn}) in  Lemma \ref{decreaseineps}, and Eq.~(\ref{ndslvakn}) in Lemma \ref{decreaseindelta}, we obtain the following:
\begin{Lemma}
For fixed $q\in\mathbb{P}(N)$, and $h^{i},h^{f}\in\mathbb{R}^{N}$, the quantity $\mathcal{C}^{\mathrm{extr}}_{\epsilon,\delta}(q,h^{i},h^{f})$ increases monotonically with decreasing $\epsilon$, as well as with decreasing $\delta$.
\end{Lemma}
As we would expect,  the work cost thus increases if we demand a lower risk of failure (smaller $\epsilon$), or if we require a higher precision in the extraction (smaller $\delta$). Note that a direct consequence of the monotonicity of $\mathcal{C}^{\mathrm{extr}}_{\epsilon,\delta}(q,h^{i},h^{f})$ in $\delta$ is that the limit in Eq.~(\ref{DefCextr}) is well defined. 

\begin{Definition}
Let $h^{i},h^{f}\in\mathbb{R}^{N}$, $q\in\mathbb{P}(N)$, and let $\mathcal{N}$ be a random variable with distribution $q$. Let $\epsilon,\delta\in\mathbb{R}$ be such that $0 < \epsilon < 1$ and $0 \leq \delta <+\infty$. We define the $(\epsilon,\delta)$-deterministic work content of $q$ relative to $h$ as
\begin{equation}
\begin{split}
\mathcal{A}^{\epsilon}_{\delta}(q,h): = &  -\mathcal{C}^{\mathrm{extr}}_{\epsilon,\delta}(q,h,h)\\
= & \sup\bigcup_{\mathcal{P}\in\mathscr{P}(h^{i},h^{f})}\Delta_{\delta}^{\epsilon}\boldsymbol{(}-W(\mathcal{P},\mathcal{N})\boldsymbol{)},
\end{split}
\end{equation}
 and the $\epsilon$-deterministic work content as
\begin{equation}
\mathcal{A}^{\epsilon}(q,h): = 
-\mathcal{C}^{\mathrm{extr}}_{\epsilon}(q,h,h).
\end{equation} 
\end{Definition}

\begin{Proposition}
\label{LowerMax}
Let $h^i,h^f\in\mathbb{R}^{N}$, and $q\in\mathbb{P}(N)$. Let $0< \epsilon\leq 1-\frac{1}{\sqrt{2}}$, and $0<\delta < +\infty$.  Then
\begin{equation}
\label{akdfnlb}
\begin{split}
&   F(h^f)- F^{\epsilon}(q,h^i) +\frac{1}{\beta}\ln(1-\epsilon)  -6\delta  \\
& \leq  \mathcal{C}^{\mathrm{extr}}_{\epsilon,\delta}(q,h^{i},h^{f}) \\
& \leq    F(h^f)- F^{\epsilon}(q,h^i). 
\end{split}
\end{equation}
\end{Proposition}
The proof of Proposition \ref{LowerMax} is presented in Sec.~\ref{ProofofMain}.
Note that due to Eq.~(\ref{RelationDeps0andFeps}) we have 
\begin{equation*}
\begin{split}
 F(h^f)- F^{\epsilon}(q,h^i) = & F(h^{f})-F(h^{i})\\
&   -kT\ln(2)D^{\epsilon}_{0}\boldsymbol{(}q\Vert G(h^{i})\boldsymbol{)},
 \end{split}
 \end{equation*}
  which puts Eq.~(\ref{akdfnlb}) in a form more similar to the results concerning the expected work extraction in Proposition \ref{OptimalExpectedWork}.

\begin{Corollary}
\label{AlmostDetermWorkCont}
Let $h\in\mathbb{R}^{N}$, $q\in\mathbb{P}(N)$, and let $\mathcal{N}$ be a random variable with distribution $q$. Let $\epsilon\in\mathbb{R}$ be such that $0 < \epsilon \leq 1-\frac{1}{\sqrt{2}}$, then 
\begin{equation*}
0  \leq \mathcal{A}^{\epsilon}(q,h)- F^{\epsilon}(q,h)+F(h) \leq  -kT\ln(1-\epsilon).
\end{equation*}
\end{Corollary}
As considered in more detail in Sec.~\ref{thermfluct}, the quantity $-kT\ln(1-\epsilon)$ is the largest $\epsilon$-deterministic work that can be extracted from a thermal equilibrium state.

If we in Corollary \ref{AlmostDetermWorkCont} take the special case of a completely degenerate set of energy levels, we find $\mathcal{A}^{\epsilon}(q,h)\approx kT[\ln(N)-H_{0}^{\epsilon}(q)]$.
Here $H_{0}^{\epsilon}(q)$ is a smoothed R\'enyi $0$-entropy, defined as $H_{0}^{\epsilon}(q) := \min_{q(\Lambda)>1-\epsilon}\log_2|\Lambda|$. Up to technical differences this is essentially the result  obtained in \cite{Dahlsten}. Apart from a difference in terms of the choice of smoothing, the result in \cite{Dahlsten} is stated in terms of a smoothed R\'enyi $1/2$-entropy, rather than the smoothed R\'enyi $0$-entropy we use. However, results in, e.g., \cite{RennerWolf} suggest that these entropies can be substituted at the cost of error terms constant in the system size. A perhaps more relevant difference is that \cite{Dahlsten} does not employ the same type of cost function as we do here, but rather the type of threshold function  briefly described in Sec.~\ref{OtherCostFcns}.


\section{ \label{ProofofMain} Proof of Proposition~\ref{LowerMax}}

The proof idea of the lower bound in Prop.~\ref{LowerMax} is to decompose the total process into two parts, where the first part is the initial LT of the process, and where the second part (if any) begins with a thermalization. We proceed by finding a lower bound on the work cost for a single LT.
 In Sec.~\ref{epsextracthermal} we prove a variant of Crook's fluctuation theorem to find a similar bound for any process that operates on an equilibrium distribution. In Sec.~\ref{kmfbl} we combine these two bounds to obtain the lower bound in Eq.~(\ref{akdfnlb}). To obtain the upper bound in Prop.~\ref{LowerMax} we construct a sequence of processes for which the $\epsilon$-deterministic work cost converges to the upper bound.

\subsection{\label{LowerBoundLT}Lower bound on the work cost of a single LT}

\begin{Lemma}
\label{singleLT}
Let $h^{i},h'\in\mathbb{R}^{N}$, and let $\mathcal{P}$ be a single LT that takes $h^i$ to $h'$. Let $\mathcal{N}$ be distributed $q\in\mathbb{P}(N)$. Let $w$ be an $(\epsilon,\delta)$-deterministic value of $W(\mathcal{P},\mathcal{N})$, then
\begin{equation}
\label{tpopui}
w \geq F(h')- F^{\epsilon}(q,h^{i}) -\delta.
\end{equation}
\end{Lemma}
\begin{proof}
We first note that with probability $q_j$ the work cost of the single LT is $h'_j -h^{i}_j$. 
Define $\Lambda := \{j\in\{1,\ldots, N\}: |h'_{j}-h^{i}_{j} -w|\leq \delta \}$. 
Since $w$ is an $(\epsilon,\delta)$-deterministic value of $W(\mathcal{P},\mathcal{N})$ it means, by definition, that $q(\Lambda)>1-\epsilon$.
\begin{equation}
\begin{split}
F(h') \leq & -\frac{1}{\beta}\ln\sum_{j\in\Lambda}e^{-\beta h'_{j}}\\
\leq & -\frac{1}{\beta}\ln\sum_{j\in\Lambda }e^{-\beta h^{i}_{j} - \beta (\delta+w)}\\
= & \delta+w  -\frac{1}{\beta}\ln\sum_{j\in\Lambda }e^{-\beta h^{i}_{j}}\\
\leq & \delta + w + F^{\epsilon}(q,h^{i}).
\end{split}
\end{equation}
\end{proof}

\subsection{\label{epsextracthermal} A variation on Crook's theorem}

Crook's theorem \cite{CrooksTheorem} relates the probability distribution of  the entropy production, or work cost, of a process and its reversal. We shall here derive a slight variation of Crook's theorem within our model. Since the thermalization in our model is a special case of the detailed balance condition, it is maybe not particularly surprising that we can derive Crook's theorem. (Note, e.g., that a similar type of model, based on detailed balance, was employed in \cite{Crooks98} to derive the Jarzynski equality \cite{Jarzynski}.) The reason for this derivation is more on a technical level; we need a particular version that fits well within our framework.

Given a process $\mathcal{P}\in\mathscr{P}(h^{i},h^{f})$ we define the process $\mathcal{P}^{\textrm{rev}}\in\mathscr{P}(h^{f},h^{i})$ as the reversal of the sequence of thermalizations and LTs in $\mathcal{P}$. To be more precise, let $\mathcal{P}$ is a sequence of alternating thermalizations and LTs along a sequence of energy level configurations $(h^{0},h^{1},\ldots,h^{L-1}, h^{L})$, with $h^{0} := h^{i}$ and $h^{L}:=h^f$. Then $\mathcal{P}^{\textrm{rev}}$ is the sequence of alternating thermalizations and LTs along the sequence $(h^{\textrm{rev},0},\ldots, h^{\textrm{rev},L-1}) := (h^{L},h^{L-1},\ldots, h^{1}, h^{0})$. Hence $h^{\textrm{rev},l} = h^{L-l}$. 

 In the following lemma we assume that the forward process $\mathcal{P}$  operates on an initial state $\mathcal{N}^{i}$ that is Gibbs distributed $G(h^i)$. Similarly, the reversed process $\mathcal{P}^{\textrm{rev}}$ operates on the initial state $\mathcal{N}^{f}$, which has the Gibbs distribution $G(h^{f})$. Note that $\mathcal{N}^{f}$ should not be confused with the final state $\mathcal{F}(\mathcal{P},\mathcal{N}^i)$. These are not necessarily equal, and may not have the same distribution (unless $\mathcal{P}$ ends with a thermalization). Similarly, $\mathcal{F}(\mathcal{P}^{\textrm{rev}},\mathcal{N}^f)$ does not necessarily have to have the same distribution as $\mathcal{N}^{i}$.

\begin{Lemma}
\label{Crooks}
Let $h^{i},h^{f}\in \mathbb{R}^{N}$, let $\mathcal{P}\in \mathscr{P}(h^{i},h^{f})$, and $\delta >0$. 
Denote $\Delta F := F(h^{f})-F(h^{i})$.
 Then, for all $w\in\mathbb{R}$,
\begin{equation*}
\begin{split}
e^{\beta (w-\Delta F)}e^{-\beta \delta} \leq & \frac{P\boldsymbol{(} |W(\mathcal{P},\mathcal{N}^i) -w | \leq \delta \boldsymbol{)}}{P\boldsymbol{(} |W(\mathcal{P}^{\textrm{rev}},\mathcal{N}^f) +w | \leq \delta \boldsymbol{)}}\\
\leq & e^{\beta (w-\Delta F)}e^{\beta \delta},
\end{split}
\end{equation*}
where $\mathcal{N}^i$ is distributed $G(h^i)$, and $\mathcal{N}^{f}$ is distributed $G(h^{f})$. 
\end{Lemma}

\begin{proof}
Consider a specific path $\overline{n} := (n_0,\ldots, n_{L-1})\in \{1,\ldots, N\}^{\times L}$ of the forward process $\mathcal{P}$, i.e., initially the system is in state $n_{0}$, next in state $n_{1}$, etc. The work cost of this path is  $w_{\overline{n}} := \sum_{l=0}^{L-1}(h^{l+1}_{n_l}- h^{l}_{n_l})$.
 Analogously, the work cost of a path $\overline{m} = (m_0,\ldots, m_{L-1})\in \{1,\ldots, N\}^{\times L}$ of the reversed process $\mathcal{P}^{\textrm{rev}}$ is
 \begin{equation}
 w^{\textrm{rev}}_{\overline{m}} :=  \sum_{l=0}^{L-1}(h^{\textrm{rev}, l+1}_{m_l}- h^{\textrm{rev},l}_{m_l}).
 \end{equation}
 Let us now define the operation
 $\textrm{rev}(\overline{n})_{l} := n_{L-l-1}$, for $l= 0,\ldots,L-1$, 
which is a bijection on $\{1,\ldots, N\}^{\times L}$. 

 With these definitions, one can check that
  \begin{equation}
  \label{reverse}
 w^{\textrm{rev}}_{\overline{m}}  =  -w_{\textrm{rev}(\overline{m})}.
 \end{equation}
 Let us again consider the work cost $w_{\overline{n}}$ of the path $\overline{n}$ of the forward process. By using Eq.~(\ref{mfnxkdjng}) it follows that
\begin{equation}
\label{doitfhj}
\begin{split}
w_{\overline{n}}   = & F(h^L)-F(h^0)   - \frac{1}{\beta}\ln \Pi_{l=0}^{L-1}  G_{n_l}(h^{l+1})\\
&  +\frac{1}{\beta}\ln\Pi_{l=0}^{L-1}G_{n_l}(h^l) .
\end{split}
\end{equation}
Furthermore, the probability to get path $\overline{n}$ in the forward process $\mathcal{P}$ is  $P(\overline{n}) :=  \Pi_{l=0}^{L-1}G_{n_l}(h^l)$.
Hence, we can rewrite Eq.~(\ref{doitfhj}) as 
\begin{equation}
\label{dlfbys}
\begin{split}
P(\overline{n}) = e^{\beta(w_{\overline{n}}  -\Delta F)} \Pi_{l=0}^{L-1}  G_{n_l}(h^{l+1}).\\
\end{split}
\end{equation}
Let us define the set of paths 
\begin{equation}
\label{flbnmfsl}
\begin{split}
A := \{ \overline{n}\in \{1,\ldots,N\}^{\times L}:  |w_{\overline{n}} -w  |\leq \delta\}.
\end{split}
\end{equation}
In other words, $A$ consists of those  paths for which the work cost of the forward process differs at most  $\delta$ from $w$. Hence, using Eq.~(\ref{dlfbys}), this results in 
\begin{equation}
\begin{split}
& P(|W(\mathcal{P},\mathcal{N}^{i})-w|\leq \delta) \\
& = \sum_{\overline{n}\in A}P(\overline{n})\\
& =  \sum_{\overline{n}\in A}e^{\beta(w_{\overline{n}}  -\Delta F)} \Pi_{l=0}^{L-1}  G_{n_l}(h^{l+1}).
\end{split}
\end{equation}
Utilizing the defining condition in Eq.~(\ref{flbnmfsl}), i.e., $w-\delta \leq w_{\overline{n}} \leq w + \delta$, yields the inequalities
\begin{equation}
\label{kvbjirut}
\begin{split}
 e^{\beta(w  -\Delta F  - \delta)}  \leq &  \frac{P(|W(\mathcal{P},\mathcal{N}^{i})-w|\leq \delta)}{
 \sum_{\overline{n}\in A} \Pi_{l=0}^{L-1}  G_{n_l}(h^{l+1}) }\\
 \leq & e^{\beta(w  -\Delta F+ \delta )}.
 \end{split}
\end{equation}
Let us now define the set $B$ of paths such the reversed process $\mathcal{P}^{\textrm{rev}}$ gives a work cost that differs at most $\delta$ from $-w$.
\begin{equation}
\begin{split}
B := \{ \overline{m}\in \{1,\ldots,N\}^{\times L}:  |w^{\textrm{rev}}_{\overline{m}} + w  |\leq \delta\}.
\end{split}
\end{equation}
By Eq.~(\ref{reverse}) and the fact that $\overline{n}\mapsto \textrm{rev}(\overline{n})$ is a bijection it follows that $\textrm{rev}(B)=  A$.

The probability of a path $\overline{m}$ of the reversed process $\mathcal{P}^{\textrm{rev}}$ is 
\begin{equation}
\begin{split}
\mathcal{P}^{\textrm{rev}}(\overline{m}) 
= & \Pi_{l=0}^{L-1}G_{m_l}(h^{L-l})\\
= & \Pi_{l'=0}^{L-1}G_{\textrm{rev}(\overline{m})_{l'}}(h^{l'+1}).
\end{split}
\end{equation}
(Note that $\mathcal{P}^{\textrm{rev}}(\overline{m}) \neq \mathcal{P}\boldsymbol{(}\textrm{rev}(\overline{m})\boldsymbol{)}$.)

By construction of the set $B$, we can conclude that
\begin{equation*}
\begin{split}
& P(|W(\mathcal{P}^{\textrm{rev}},\mathcal{N}^{f}) + w|\leq \delta ) \\
& \quad  =  \sum_{\overline{m}\in B}\Pi_{l=0}^{L-1}G_{\textrm{rev}(\overline{m})_{l}}(h^{l+1})\\
& \quad =  \sum_{\overline{n}\in A}\Pi_{l=0}^{L-1}G_{\overline{n}_{l}}(h^{l+1}).
\end{split}
\end{equation*}
Inserting the above equation into Eq.~(\ref{kvbjirut}) yields the statement of the lemma.
\end{proof}
\begin{Corollary}
\label{MaxThermal}
Let $h^{i},h^{f}\in \mathbb{R}^{N}$, let $\mathcal{P}\in \mathscr{P}(h^{i},h^{f})$, and let $0< \epsilon <1$ and $0 <\delta <+\infty$. If $\mathcal{N}$ is distributed $G(h^{i})$ then 
\begin{equation*}
\inf\!\Delta_{\delta}^{\epsilon}\boldsymbol{(}W(\mathcal{P},\mathcal{N})\boldsymbol{)} \geq \frac{1}{\beta}\ln(1-\epsilon) +F(h^f) -F(h^i) -\delta.
\end{equation*}
\end{Corollary}

\subsection{\label{kmfbl} Proof of the lower bound in Eq.~(\ref{akdfnlb})}

\begin{proof}
Without loss of generality the process $\mathcal{P}$ can be regarded as an alternating sequence of LT processes and thermalizations, beginning with an LT that takes $h^i$ to some configuration of energy levels $h'$. (We can let $h' = h^{i}$.)  We let $\mathcal{P}^{(1)}$ denote this initial LT of the process. Next, the system is thermalized, and the remaining part of the process is denoted $\mathcal{P}^{(2)}$. 
The process $\mathcal{P}^{(2)}$ begins in the state $\mathcal{N}'$ which is distributed $G(h')$. Since $\mathcal{P}^{(1)}$ and $\mathcal{P}^{(2)}$ are separated by a thermalization, it follows that $W(\mathcal{P}^{(1)},\mathcal{N}^{i})$ and $W(\mathcal{P}^{(2)},\mathcal{N}')$ are independent.
Furthermore, by assumption $0\leq \epsilon \leq 1-1/\sqrt{2}$, and $0<\delta<+\infty$. Hence, by Lemma \ref{mindetaddition} we find that
\begin{equation*}
\begin{split}
\inf\!\Delta_{\delta}^{\epsilon}\boldsymbol{(}W(\mathcal{P},\mathcal{N}^{i})\boldsymbol{)} \geq & \inf\!\Delta_{\delta}^{\epsilon}\boldsymbol{(}W(\mathcal{P}^{(1)},\mathcal{N}^{i})\boldsymbol{)} \\
& +  \inf\!\Delta_{\delta}^{\epsilon}\boldsymbol{(}W(\mathcal{P}^{(2)},\mathcal{N}')\boldsymbol{)} -4\delta.
\end{split}
\end{equation*} 
We know from Lemma \ref{singleLT} that
\begin{equation*}
\inf\!\Delta_{\delta}^{\epsilon}\boldsymbol{(}W(\mathcal{P}^{(1)},\mathcal{N}^{i})\boldsymbol{)} \geq  F(h')- F^{\epsilon}(q,h^i)-\delta.
\end{equation*}
Since $\mathcal{P}^{(2)}$ starts in the Gibbs distribution $G(h')$, Corollary \ref{MaxThermal} yields
\begin{equation*}
 \inf\!\Delta_{\delta}^{\epsilon}\boldsymbol{(}W(\mathcal{P}^{(2)},\mathcal{N}')\boldsymbol{)}  \geq \frac{1}{\beta}\ln(1-\epsilon) +F(h^f) -F(h') -\delta.
\end{equation*}
Combining the above inequalities yields the lower bound in Eq.~(\ref{akdfnlb}).
\end{proof}

\subsection{\label{NearOptimal}Proof of the upper bound in Eq.~(\ref{akdfnlb})}

To prove the upper bound in Proposition \ref{LowerMax} we shall here construct an explicit family of processes whose $(\epsilon,\delta)$-deterministic work values approach the bound. 

However, shall first prove a lemma. In the proof of this lemma make use of Lemma \ref{LemmaITR}, which can be rephrased as asserting the existence of a process $\mathcal{P}\in \mathscr{P}(h^i,h^f)$, such that $F(h^f)- F(h^i)$ is an $(\epsilon,\delta)$-deterministic value of $W(\mathcal{P},\mathcal{N})$, where $\mathcal{N}$ is distributed $G(h^i)$. 

\begin{Lemma}
\label{ExistenceMindet}
Let $h^i,h^f\in\mathbb{R}^{N}$ and let  $0< \epsilon \leq 1$, $0<\delta<+\infty$, and $0< \xi < 1$. Let $\mathcal{N}^{i}$ have the distribution $q\in\mathbb{P}(N)$. Then there exists a $\mathcal{P}\in\mathscr{P}(h^{i},h^{f})$ and a $w\in \mathbb{R}$, such that $w$ is an 
   an $(\epsilon,\delta)$-deterministic value of $W(\mathcal{P},\mathcal{N}^{i})$, and
   \begin{equation}
   \label{nadvjs}
   w \leq \xi + F(h^f)-F^{\epsilon}(q,h^i).
   \end{equation}
\end{Lemma}

\begin{proof}
As noted in Sec.~\ref{Sec:FandD0} there exists a $\Lambda^*\subseteq \{1,\ldots,N\}$ such that $F^{\epsilon}(q,h^i) = -kT\ln\sum_{j\in \Lambda^*}e^{-\beta h^i_n}$.  Due to the strict inequality, $q(\Lambda^*) >1-\epsilon$, there exists a number  $1> r>0$ such that $q(\Lambda^*) \geq rq(\Lambda^*) > 1-\epsilon$.
We construct a process $\mathcal{P}$ as a concatenation of a process $\mathcal{P}^{(1)}$, a thermalization, and a process $\mathcal{P}^{(2)}$.
Let $E\in\mathbb{R}$, and let $\mathcal{P}^{(1)}$ be the LT that takes $h^i$ to the new configuration of energy levels $h'$, defined by 
\begin{equation}
h'_{k} := \left\{ \begin{matrix} 
h^i_{k} & \textrm{if} & k\in \Lambda^*,\\
h^i_{k}+ E & \textrm{if} & k\notin \Lambda^*. 
\end{matrix}\right.
\end{equation}
As seen, $W(\mathcal{P}^{(1)},\mathcal{N}^{i})$ takes the value $0$ with probability $q(\Lambda^*)$. Next we thermalize the system, thus putting it in a state $\mathcal{N}'$ distributed $G(h')$. 

By Lemma \ref{LemmaITR} we know that there exists a  process $\mathcal{P}^{(2)}\in\mathscr{P}(h',h^f)$, such that $F(h^f)-F(h')$ is an $(1-r,\delta)$-deterministic value of $W(\mathcal{P}^{(2)},\mathcal{N}')$.
Since $W(\mathcal{P}^{(1)},\mathcal{N}^{i})$ and $W(\mathcal{P}^{(2)},\mathcal{N}')$ are independent it follows that  $P\boldsymbol{(}|W(\mathcal{P}^{(1)},\mathcal{N}^{i}) + W(\mathcal{P}^{(2)},\mathcal{N}') -F(h^f)+F(h') | \leq \delta \boldsymbol{)} >   q(\Lambda^*)r >  1-\epsilon$.
Hence, $w = F(h^f)-F(h')$ is an $(\epsilon,\delta)$-deterministic value of $W(\mathcal{P},\mathcal{N}^{i})$.
Note that $F(h')$ is a monotonically increasing function in increasing $E$, and $\lim_{E\rightarrow +\infty}F(h') = F^{\epsilon}(q,h^i)$. Hence, for each $\xi >0$ there exists an $E$ such that $-F(h') \leq \xi -F^{\epsilon}(q,h^i)$. 
For such a choice of $E$, the corresponding process $\mathcal{P}$ thus have $w = F(h^f)-F(h')$ as an $(\epsilon,\delta)$-deterministic value of $W(\mathcal{P},\mathcal{N}^i)$, which moreover  satisfies  Eq.~(\ref{nadvjs}).
\end{proof}

\begin{proof}[Proof of the upper bound in Eq.~(\ref{akdfnlb}).]
By Lemma \ref{ExistenceMindet} we can conclude that there exists a sequence of  processes $\mathcal{P}_m \in \mathscr{P}(h^i,h^f)$ and a sequence of real numbers $w_m$, such that $w_m$ is an $(\epsilon,\delta)$-deterministic value of $W(\mathcal{P}_m,\mathcal{N})$, and $w_m \leq \frac{1}{m} + F(h^f)-F^{\epsilon}(q,h^i)$. This proves the upper bound in Eq.~(\ref{akdfnlb}).
\end{proof}


\section{\label{thermfluct}  $\epsilon$-deterministic work extraction from thermal equilibrium}

Let $h\in\mathbb{R}^{N}$. A direct consequence of Corollary \ref{MaxThermal} is that 
\begin{equation}
\label{fjnbdfajk}
\mathcal{A}^{\epsilon}\boldsymbol{(}G(h),h\boldsymbol{)} \leq -kT\ln(1-\epsilon) .
\end{equation}
In other words $-kT\ln(1-\epsilon)$ is an upper bound to the $\epsilon$-deterministic work content of a system that is in equilibrium with a heat bath of temperature $T$. 
(Note that Eq.~(\ref{fjnbdfajk}) implies that the upper bound in Corollary \ref{AlmostDetermWorkCont} is not sharp for all initial distributions and configurations of energy levels.) Corollary \ref{AlmostDetermWorkCont} furthermore provides a lower bound:
\begin{equation}
\label{mgfkkfsfg}
\mathcal{A}^{\epsilon}\boldsymbol{(}G(h),h\boldsymbol{)} \geq  -kT \ln\inf_{G_{\Lambda}(h)>1-\epsilon}G_{\Lambda}(h).
\end{equation}
This lower bound implies that with suitable configurations of energy levels $h$ we can extract $\epsilon$-deterministic work arbitrarily close to the upper bound in Eq.~(\ref{fjnbdfajk}). (For example, if we let one energy level, $1$ say, be sufficiently much lower in energy than all others, we can have $G_1(h)>1-\epsilon$ with $G_1(h)$ arbitrarily close to $1-\epsilon$.)

Note that the above statement can be rephrased as the probability of success of the extraction being exponentially small in the extracted energy. This is in agreement with the standard fluctuation theorems, where thermal fluctuations can violate the macroscopic notion of the second law, but with a probability that is exponentially small in the size of the violation \cite{ReviewFluctThm}. 

From Sec.~\ref{Sec:OptimalExpected} we know that the expected work content of the equilibrium is zero. Hence, the gain in the successful case in the $\epsilon$-deterministic work extraction has to be compensated by a corresponding loss in the unsuccessful case. (In the specific process used in the proof of  Lemma \ref{ExistenceMindet} the cost of failure approaches infinity.)


\section{ \label{Sec:AlmDetEr} Optimal $\epsilon$-deterministic work cost of erasure}
Like for the expected erasure cost in Sec.~\ref{Sec:OptimalWorkCostErasure}, we do not require the erasure process to establish the selected state $s$ perfectly, but rather allow an error that in the end is taken to zero.
\begin{Definition}
Let $h^i,h^f\in\mathbb{R}^{N}$, let $\mathcal{N}$ be a random variable with distribution $q\in\mathbb{P}(N)$, and let $s\in\{1,\ldots,N\}$. Let $0 < \epsilon < 1$ and $0 \leq \delta <+\infty$. Define the $(\epsilon,\delta)$-deterministic work cost of erasure as
\begin{equation}
\label{defepsdelterase}
\begin{split}
 & \mathcal{C}^{\mathrm{erase}}_{\epsilon,\delta}(q,h^i,h^f,s) \\
 & :=  \lim_{\tau\rightarrow 0^+}\inf \bigcup_{\mathcal{P}\in\mathscr{P}^{\tau}_{s}(q,h^i,h^f)} \Delta_{\delta}^{\epsilon}\boldsymbol{(}W(\mathcal{P},\mathcal{N})\boldsymbol{)}\\
 & =  \lim_{\tau\rightarrow 0^+}\inf_{\mathcal{P}\in\mathscr{P}^{\tau}_{s}(q,h^i,h^f)} \inf\!\Delta_{\delta}^{\epsilon}\boldsymbol{(}W(\mathcal{P},\mathcal{N})\boldsymbol{)},
\end{split}
\end{equation} 
and the $\epsilon$-deterministic work cost of erasure as
\begin{equation}
\label{bfdnmyny}
\mathcal{C}^{\mathrm{erase}}_{\epsilon}(q,h^i,h^f,s) := \lim_{\delta \rightarrow 0^+} \mathcal{C}^{\mathrm{erase}}_{\epsilon,\delta}(q,h^i,h^f,s).
\end{equation}
\end{Definition}
Analogously as for the definition of the expected work cost of erasure, Def.~\ref{nsfdbkvnbakn} in Sec.~\ref{Sec:OptimalWorkCostErasure}, the quantity $\inf \bigcup_{\mathcal{P}\in\mathscr{P}^{\tau}_{s}(q,h^i,h^f)} \Delta_{\delta}^{\epsilon}\boldsymbol{(}W(\mathcal{P},\mathcal{N})\boldsymbol{)}$ increases monotonically with decreasing $\tau$. Hence, the limit $\tau\rightarrow 0$ in Eq.~(\ref{defepsdelterase}) is well defined.

As a direct consequence of Eq.~(\ref{nsdklvn}) in Lemma \ref{decreaseineps} and Eq.~(\ref{ndslvakn}) in  Lemma \ref{decreaseindelta} it follows that, for a fixed $\tau$, the quantity $\inf \bigcup_{\mathcal{P}\in\mathscr{P}^{\tau}_{s}(q,h^i,h^f)} \Delta_{\delta}^{\epsilon}\boldsymbol{(}W(\mathcal{P},\mathcal{N})\boldsymbol{)}$ decreases monotonically with increasing $\epsilon$, as well as with increasing $\delta$. This remains true in the limit $\tau \rightarrow 0$. We can thus conclude:
\begin{Lemma}
For fixed $q\in\mathbb{P}(N)$, $h^{i},h^{f}\in\mathbb{R}^{N}$, and $s\in\{1,\ldots,N\}$ the quantity $\mathcal{C}^{\mathrm{erase}}_{\epsilon,\delta}(q,h^{i},h^{f},s)$ increases monotonically with decreasing $\epsilon$, as well as with decreasing $\delta$.
\end{Lemma}
Due to the monotonicity of $\mathcal{C}^{\mathrm{erase}}_{\epsilon,\delta}(q,h^{i},h^{f},s)$ with respect to $\delta$, the limit $\delta\rightarrow 0$ in Eq.~(\ref{bfdnmyny}) is well defined.

\begin{Proposition}
\label{AlmDetErasure}
Let $h^i,h^f\in\mathbb{R}^{N}$, and let $q\in \mathbb{P}(N)$.  Let $0  < \epsilon \leq 1-\frac{1}{\sqrt{2}}$,  $0<\delta<+\infty$, and $s\in \{1,\ldots, N\}$.  
Then
\begin{equation}
\label{ndfjb}
\begin{split}
 & h^f_s  - F^{\epsilon}(h^i) +  \frac{1}{\beta}\ln(1-\epsilon) -8\delta\\
& \leq \mathcal{C}^{\mathrm{erase}}_{\epsilon,\delta}(q,h^i,h^f,s)\\
& \leq  h^f_s  - F^{\epsilon}(h^i).
\end{split}
\end{equation}
\end{Proposition}

\begin{Corollary}
\label{AlmostDetermEreasureCost}
Let $h\in\mathbb{R}^{N}$, $q\in\mathbb{P}(N)$, and let $\mathcal{N}$ be a random variable with distribution $q$. Let $\epsilon\in\mathbb{R}$ be such that $0 < \epsilon \leq 1-\frac{1}{\sqrt{2}}$, then 
\begin{equation*}
  \frac{1}{\beta}\ln(1-\epsilon) \leq \mathcal{C}^{\mathrm{erase}}_{\epsilon}(q,h,h,s) -h_s +F^{\epsilon}(h) \leq 0.
\end{equation*}
\end{Corollary}
As a special case of Corollary \ref{AlmostDetermEreasureCost} it follows that the $\epsilon$-deterministic work cost of erasure for the case of a completely degenerate set of energy levels is bounded as $ kT\ln(1-\epsilon) \leq \mathcal{C}^{\mathrm{erase}}_{\epsilon}(q,h,h,s) -kTH_{0}^{\epsilon}(q) \leq 0$.


\section{\label{ProofAlmDetErasure} Proof of Proposition \ref{AlmDetErasure}}

The proof idea of the lower bound in Proposition \ref{AlmDetErasure} is to divide the total erasure process into two parts. The first part is almost the entire process apart from the very last LT.  For the first part we can apply our results on work extraction in Proposition \ref{LowerMax} to find a bound on the work cost. We next observe that the very last LT is very constrained by the requirement that the system with high probability should end up in  state $s$. This leads to a bound on the work cost. To prove the upper bound in  Proposition \ref{AlmDetErasure} we define a specific sequence of processes for which the $\epsilon$-deterministic erasure cost converge to the upper bound in Proposition \ref{AlmDetErasure}.

\subsection{\label{BoundErasure} Proof of the lower bound in Eq.~(\ref{ndfjb})}

\begin{Lemma}
\label{anvvaks}
Let $h^i,h^f\in\mathbb{R}^{N}$, let $\mathcal{N}$ be a random variable distributed $q\in\mathbb{P}(N)$, and let $s\in \{1,\ldots, N\}$.  Let $0 < \tau < \epsilon \leq 1-\frac{1}{\sqrt{2}}$, and $0<\delta<+\infty$.  
Then
\begin{equation}
\label{amvaa}
\begin{split}
 \inf\!\Delta_{\delta}^{\epsilon}\boldsymbol{(}W(\mathcal{P},\mathcal{N})\boldsymbol{)} \geq &  h^f_s  - F^{\epsilon}(q,h^i) -8\delta\\
  & + \frac{1}{\beta}\ln[(1-\epsilon)(1-\tau)], 
 \end{split}
\end{equation}
for all $\mathcal{P}\in\mathscr{P}^{\tau}_{s}(q,h^i,h^{f})$.
\end{Lemma}

\begin{proof}
First we note that Eq.~(\ref{amvaa}) is trivially true for processes $\mathcal{P}$ such that $\inf\!\Delta_{\delta}^{\epsilon}\boldsymbol{(}W(\mathcal{P}),\mathcal{N}\boldsymbol{)} = +\infty$. Hence, without loss of generality we may in the following restrict to processes $\mathcal{P}$ is such that $\inf\!\Delta_{\delta}^{\epsilon}\boldsymbol{(}W(\mathcal{P}),\mathcal{N}\boldsymbol{)} < +\infty$. Since $\mathcal{P}$ can be regarded as an alternating sequence of LTs and thermalizations, we can distinguish two cases:  $\mathcal{P}$ contains no thermalization, and thus effectively consists only of a single LT, or $\mathcal{P}$ contains at least one thermalization.

In the first case $\mathcal{P}$ consists only of a single LT. Hence, this LT must transform $h^i$ to $h^f$. Since an LT does not change the distribution of the state, we must have $q_s  \geq 1-\tau$. Since $\tau<\epsilon \leq 1-1/\sqrt{2} <1/2$ it means that $q_s > 1/2$. Hence, any subset $\Lambda\subseteq\{1,\ldots,N\}$, with $q(\Lambda)>1-\epsilon \geq 1/2$ \emph{must} contain $s$. Moreover, since $\tau<\epsilon$ implies $q_s \geq 1-\tau > 1-\epsilon$, it is enough if $\Lambda = \{s\}$ for $q(\Lambda) > 1-\epsilon$ to hold. One can thus realize that  $F^{\epsilon}(q,h^i) = -\ln\inf_{q(\Lambda)>1-\epsilon}\sum_{n\in\Lambda}e^{-\beta h^i_n} = h^i_{s}$.  Furthermore, since $q_s \geq 1-\tau > 1-\epsilon > 1/2$, it follows that $h^f_s-h^i_s$ is an $(\epsilon,\delta)$-deterministic value of $W(\mathcal{P})$. By Lemma \ref{lowerbound}, we thus have $h^f_s-h^i_s -2\delta\leq  \inf\!\Delta_{\delta}^{\epsilon}\boldsymbol{(}W(\mathcal{P}),\mathcal{N}\boldsymbol{)}$. We can conclude that the inequality in Eq.~(\ref{amvaa}) is satisfied.

The second case is that the process contains at least one thermalization. We may thus decompose $\mathcal{P}$ into two parts.
The first part, $\mathcal{P}^{(1)}$, is the entire process up to (and including) the last thermalization. This thermalization is done with respect to some set of energy levels $h'$ and leads to the state $\mathcal{N}'$, which is distributed $G(h')$.  The second part, $\mathcal{P}^{(2)}$, consists only of a single LT that takes $h'$ to $h^f$. Due to the thermalization at the end of $\mathcal{P}^{(1)}$, it follows that $W(\mathcal{P}^{(1)},\mathcal{N})$ and $W(\mathcal{P}^{(2)},\mathcal{N}')$ are independent. Since $\inf\!\Delta_{\delta}^{\epsilon}\boldsymbol{(}W(\mathcal{P}^{(1)},\mathcal{N}) + W(\mathcal{P}^{(2)},\mathcal{N}')\boldsymbol{)} < +\infty$, by assumption, it follows by Lemma \ref{finite} that 
$ \inf\!\Delta_{\delta}^{\epsilon}\boldsymbol{(}W(\mathcal{P}^{(1)},\mathcal{N})\boldsymbol{)} < +\infty$. Hence, there exists a $w\in \mathbb{R}$ such that 
\begin{equation}
\label{Aepsdet}
P\boldsymbol{(}|W(\mathcal{P}^{(1)},\mathcal{N})-w|\leq \delta\boldsymbol{)} > 1-\epsilon.
\end{equation}
Let us now consider the process $\mathcal{P}^{(2)}$. Since  $\mathcal{P}\in\mathscr{P}^{\tau}_{s}(q,h^{i},h^{f})$ it follows that $h'$ must be such that $G_{s}(h') \geq 1-\tau$, which we can rewrite as
\begin{equation}
\label{condonhs}
h'_s   \leq -\frac{1}{\beta}\ln(1-\tau)+ F(h').
\end{equation}
Note that due to the assumption $\tau >0$ there exists a $h'\in\mathbb{R}^{N}$ that satisfies this condition. 

Furthermore, $P\boldsymbol{(} W(\mathcal{P}^{(2)},\mathcal{N}') = h^f_s-h'_s \boldsymbol{)} \geq 1-\tau$.
 By combining this observation with Eq.~(\ref{Aepsdet}) we find (using the independence of $W(\mathcal{P}^{(1)},\mathcal{N})$ and $W(\mathcal{P}^{(2)},\mathcal{N}')$) that
\begin{equation*}
\begin{split}
& P\boldsymbol{(}|W(\mathcal{P},\mathcal{N}) -w -h^f_s + h'_s|\leq \delta\boldsymbol{)} \\
 \geq &  P\boldsymbol{(}|W(\mathcal{P}^{(1)},\mathcal{N})-w|\leq \delta\boldsymbol{)}P\boldsymbol{(}W(\mathcal{P}^{(2)},\mathcal{N}')= h^f_s-h'_s\boldsymbol{)}\\
 > & (1-\epsilon)(1-\tau).
\end{split}
\end{equation*}
The conditions $0 < \tau < \epsilon \leq 1- 1/\sqrt{2}$ implies  $0< \epsilon + \tau -\epsilon\tau \leq 1/2$. 
This enables us to apply Lemma \ref{lowerbound} to the above inequality, with the result 
\begin{equation}
w + h^f_s-h'_s -2\delta \leq  \inf\!\Delta_{\delta}^{\epsilon + \tau-\epsilon\tau}\boldsymbol{(}W(\mathcal{P}),\mathcal{N}\boldsymbol{)}.
\end{equation}
Since $\epsilon + \tau-\epsilon\tau \geq \epsilon$
 it follows by Lemma \ref{decreaseineps}, that 
\begin{equation}
w + h^f_s-h'_s -2\delta \leq  \inf\!\Delta_{\delta}^{\epsilon}\boldsymbol{(}W(\mathcal{P},\mathcal{N})\boldsymbol{)}.
\end{equation}
Since $w$ is an $(\epsilon,\delta)$-deterministic value of $W(\mathcal{P}^{(1)},\mathcal{N})$ we can conclude that $ \inf\!\Delta_{\delta}^{\epsilon}\boldsymbol{(}W(\mathcal{P}^{(1)},\mathcal{N})\boldsymbol{)} \leq w$. Since $\epsilon \leq 1- 1/\sqrt{2}$, Proposition \ref{LowerMax} yields 
\begin{equation*}
\begin{split}
\inf\!\Delta_{\delta}^{\epsilon}\boldsymbol{(}W(\mathcal{P},\mathcal{N})\boldsymbol{)} \geq &  h^f_s  - F^{\epsilon}(q,h^i)  -h'_s \\
& + F(h')+ \frac{1}{\beta}\ln(1-\epsilon)  -8\delta.
\end{split}
\end{equation*}
By combining this with Eq.~(\ref{condonhs}) we find Eq.~(\ref{amvaa}).
\end{proof}

\begin{proof}[Proof of the lower bound in Eq.~(\ref{ndfjb}).]
The lower bound in Eq.~(\ref{ndfjb}) follows from Lemma \ref{anvvaks} if we first take the infimum over  all processes in $\mathscr{P}^{\tau}_{s}(q,h^i,h^f)$, then take the limit $\tau\rightarrow 0$.
\end{proof}

\subsection{Proof of the upper bound in Eq.~(\ref{ndfjb})} 
 \begin{Lemma}
\label{nklfbanlkb}
Let $h^i,h^f\in\mathbb{R}^{N}$, let $\mathcal{N}$ be a random variable with distribution $q\in\mathbb{P}(N)$, and let $s\in \{1,\ldots, N\}$.  Let $0 < \tau < \epsilon \leq 1$, $0<\delta<+\infty$, $0<\xi<1$. Then there exists a  $\mathcal{P}\in\mathscr{P}^{\tau}_{s}(q,h^i,h^{f})$  and $w\in\mathbb{R}$ such that $w$ is an $(\epsilon,\delta)$-deterministic value of $W(\mathcal{P},\mathcal{N})$, and
\begin{equation}
\label{utrihie}
\begin{split}
w \leq & \xi +  h^f_s +  \frac{1}{\beta}\ln(1-\tau)\\
& -F^{(\epsilon-\tau)/(1-\tau)}(q,h^i).
\end{split}
\end{equation}
\end{Lemma}

\begin{proof} 
We construct $\mathcal{P}$ as a concatenation of a process $\mathcal{P}^{(1)}$, a thermalization, and a process $\mathcal{P}^{(2)}$. 

We begin by constructing the process $\mathcal{P}^{(1)}$.
For a given real number $E$, define $h'_{s} := h^f_s-E$ and $h'_{n} := h^f_n$ for all $n\neq s$. 
Let us choose 
\begin{equation}
 E  :=  \frac{1}{\beta}\ln\Big( \frac{1}{\tau}-1\Big) +\frac{1}{\beta}\ln\Big( e^{\beta h^f_s} \sum_{n\neq s}e^{-\beta h^f_n}\Big).
\end{equation}
This choice yields $G_s(h') = 1-\tau$ and
\begin{equation}
\label{ymvbmy}
F(h') = -E + h_{s}^{f} + \frac{1}{\beta}\ln(1-\tau).
\end{equation}
Define $\overline{\epsilon} := (\epsilon-\tau)/(1-\tau)$.  By the assumptions on $\epsilon$ and $\tau$  it follows that $0< \overline{\epsilon}< 1$. 
By Lemma  \ref{ExistenceMindet} we thus know that there exist a process $\mathcal{P}^{(1)}\in \mathcal{P}(h^{i},h')$ and an $(\overline{\epsilon},\delta)$-deterministic value $w^{(1)}$ of  $W(\mathcal{P}^{(1)},\mathcal{N})$ such that  
\begin{equation}
\label{bgjkabd}
w^{(1)} \leq \xi + F(h')-F^{\overline{\epsilon}}(q,h^i).
\end{equation}
We next turn to the process $\mathcal{P}^{(2)}$, and let it be the LT that takes $h'$ to $h^f$.  For this process $W(\mathcal{P}^{(2)},\mathcal{N}') =  h^f_{\mathcal{N}'} -h'_{\mathcal{N}'} = E\delta_{s,\mathcal{N}'}$, where $\mathcal{N}'$ is Gibbs distributed $G(h')$. Furthermore, $P\boldsymbol{(} W(\mathcal{P}^{(2)},\mathcal{N}')= E\boldsymbol{)}  = G_{s}(h') = 1-\tau$.

Let the total process $\mathcal{P}$ be a concatenation of $\mathcal{P}^{(1)}$, followed by a thermalization, and the process $\mathcal{P}^{(2)}$. Due to the thermalization, $W(\mathcal{P}^{(1)},\mathcal{N})$ and $W(\mathcal{P}^{(2)},\mathcal{N}')$ are independent, and thus 
$ P\boldsymbol{(}|W(\mathcal{P},\mathcal{N}) -w^{(1)} -E| \leq \delta\boldsymbol{)}
 \geq  P\boldsymbol{(}|W(\mathcal{P}^{(1)},\mathcal{N}) -w^{(1)} | \leq \delta\boldsymbol{)}P\boldsymbol{(}W(\mathcal{P}^{(2)},\mathcal{N}') = E\boldsymbol{)} >  1-\epsilon$.
Hence, $w : = w^{(1)} +E$ is an $(\epsilon,\delta)$-deterministic value of $W(\mathcal{P},\mathcal{N})$. By combining this with the inequalities in Eqs.~(\ref{ymvbmy}) and (\ref{bgjkabd}) we find the statement of the lemma.
\end{proof}

\begin{proof}[Proof of the upper bound in Eq.~(\ref{ndfjb}).]
Let $\xi_{m} := 1/m$ and let  $\tau_m :=\epsilon/m$ for each $m\in\mathbb{N}$ with $m\geq 2$. By Lemma \ref{nklfbanlkb} we know that for each $m$ there exists a process $\mathcal{P}_{m}\in\mathscr{P}^{\tau_m}_{s}(q,h^{i},h^{f})$ and $w_m\in\mathbb{R}$ such that $w_m$ is an $(\epsilon,\delta)$-deterministic value of $W(\mathcal{P}_{m},\mathcal{N})$, and satisfies the inequality
\begin{equation}
\label{dzkckz}
\begin{split}
w_m \leq & \frac{1}{m} +  h^f_s   +  \frac{1}{\beta}\ln(1-\frac{\epsilon}{m})\\
& -F^{\epsilon(1-m^{-1})/(1-\epsilon m^{-1})}(q,h^i).
\end{split}
\end{equation}
Note that $\epsilon(1-m^{-1})/(1-\epsilon m^{-1})$ increases monotonically to $\epsilon$ for increasing $m$. Hence, by the left-continuity of $F^{\epsilon}$ with respect to $\epsilon$, Lemma \ref{Fcontinuity}, it follows that $\lim_{m\rightarrow\infty}F^{\epsilon(1-m^{-1})/(1-\epsilon m^{-1})}(q,h^i) = F^{\epsilon}(q,h^i)$. Thus, the right hand side of Eq.~(\ref{dzkckz}) converges to $h^f_s  -F^{\epsilon}(q,h^i)$, which is the upper bound of Eq.~(\ref{ndfjb}). Furthermore $\tau_{m}$ goes to zero as $m$ increases, which thus proves the upper bound of Eq.~(\ref{ndfjb}).
\end{proof}


\section{ \label{comparisons} Comparisons}

Here we compare the expected work extraction with the $\epsilon$-deterministic work extraction for some simple examples.  For the sake of simplicity we  focus on the work yield quantities $\mathcal{A}(q,h)$ and $\mathcal{A}^{\epsilon}(q,h)$, rather than the more general work cost quantities $\mathcal{C}^{\textrm{extr}}(q,h^i,h^f)$ and $\mathcal{C}^{\textrm{extr}}_{\epsilon}(q,h^i,h^f)$. We will furthermore consider the  fluctuations (as described in Sec.~\ref{fluctuations}) in the optimal expected work extraction.
To quantify the size of these fluctuations we use the standard deviation of the work yield variable $W_{\textrm{yield}} : = W_{\textrm{yield}}(h,\mathcal{N})$ as defined by Eq.~(\ref{fluctyield}). In other words, we measure the size of the fluctuation by
\begin{equation}
\label{mvydlkm}
\sigma(W_{\textrm{yield}}) := \sqrt{\langle W_{\textrm{yield}}^2\rangle -\langle W_{\textrm{yield}} \rangle^2}.
\end{equation}
 As we have already discussed in Sec.~\ref{StandardDeviation}, this quantity gives the minimal standard deviation for any sequence of processes that achieves the maximal extracted work.

Similarly as $\mathcal{A}(q,h)$ is directly related to the relative Shannon entropy, the quantity $\sigma(W_{\textrm{yield}})$ can be related to an analogous quantity.
Given a random variable $X$ with distribution $q\in\mathbb{P}(N)$, the relative Shannon entropy between $q$  and another distribution $r\in\mathbb{P}(N)$ can be expressed as $D(q\Vert r) = \langle \log_{2}[q(X)/r(X)]\rangle$, i.e., as an expectation value of the random variable $\log_{2}[q(X)/r(X)]$.
In an analogous manner we define the standard deviation of the random variable $\log_{2}[q(X)/r(X)]$ as
\begin{equation}
\label{defrelsigma}
\sigma(q\Vert r) :=   \sqrt{\left\langle \left(\log_{2}\frac{q(X)}{r(X)}\right)^{2}\right\rangle - \left\langle \log_{2}\frac{q(X)}{r(X)}\right\rangle^{2}} 
\end{equation}
(This is the same quantity as we defined in Sec.~\ref{StandardDeviation}.)
By combining Eqs.~(\ref{fluctyield}) and (\ref{mvydlkm}) we find that
\begin{equation}
\label{stddeviation}
\sigma(W_{\textrm{yield}})  = kT\ln(2)\sigma\boldsymbol{(}q\Vert G(h)\boldsymbol{)}.
\end{equation}

In the following we shall  compare how $\mathcal{A}(q,h)$, $\sigma(W_{\textrm{yield}})$, and $\mathcal{A}^{\epsilon}(q,h)$ scale with the system size. For this purpose we will in the following consider systems that consist of $m$ `units' of some type (qubits, spins, etc). For each number $m$ we shall have an initial distribution $q^m$ and a collection of energy levels $h^{m}$. We compare the three quantities $\mathcal{A}(q^m,h^m)$, $\sigma(W^{m}_{\textrm{yield}})$, and $\mathcal{A}^{\epsilon}(q^m,h^m)$, in terms of their scalings in $m$. (Define $W^{m}_{\textrm{yield}} := W_{\textrm{yield}}(q^{m},h^{m})$.)  More precisely, we compare the leading order terms of these quantities in the limit of large $m$.

For these comparisons we use the asymptotic equivalence. Two functions $f(m)$ and $g(m)$ are asymptotically equivalent (with respect to $m\rightarrow +\infty$) denoted $f(m)\sim g(m)$, if $\lim_{m\rightarrow +\infty}[f(m)/g(m)] = 1$.
This means that $f$ and $g$ have the same leading order. 

We will also make use of an expansion up to the next to leading order in increasing $m$.
Let  $c_1$ and $c_1$ be constants, and $g_1$ and $g_2$ functions such that $g_2 = o(g_1)$ (where the latter means that $\lim_{m\rightarrow +\infty}[g_2(m)/g_1(m)]=0$), and suppose that
\begin{equation}
f(m) = c_1 g_1(m) + c_2g_2(m) + o\boldsymbol{(} g_2(m)\boldsymbol{)}.
\end{equation}
Then we say that we have a next to leading order expansion of $f$. In our case we will use $g_1(m): = m$ and $g_2(m) := \sqrt{m}$. (For a more general introduction to the notion of asymptotic expansions, see e.g.~\cite{Copson}.)

\subsection{\label{IndepIdentical}Independent, identical, and non-interacting systems}

We begin with an example where the fluctuations in the expected work extraction in some sense are small.

Consider $m$ copies of a system. These copies are independent and identical both in terms of their state distributions as well as their Hamiltonians. More precisely, we assume that the distribution $q^{m}$ of the totality of the $m$ systems is a product distribution, i.e., $q^m_{l_1,\cdots, l_m} = q_{l_1}\cdots q_{l_m}$, for some single-system distribution $q$. We will denote this $m$-fold product distribution as $q^{\otimes m}$. We furthermore assume that the systems do not interact, and that all of them have the same Hamiltonian. In terms of our model, this means that set of energy levels $h^m$ for the total system can be written $h^{m}_{l_1,\dots,l_m} = h_{l_1} +\cdots + h_{l_m}$, for some single-system set of energy levels $h$. We denote this $m$-fold direct sum by $h^{\oplus m}$.
Note that the Gibbs distribution corresponding to such a collection of identical non-interacting Hamiltonians  is a product distribution, $G(h^{\oplus m}) = G(h)^{\otimes m}$.

Due to the additivity of the relative Shannon entropy, the expected work content is
\begin{equation}
\mathcal{A}(q^{\otimes m},h^{\oplus m}) = mkT\ln(2) D\boldsymbol{(}q\Vert G(h)\boldsymbol{)}.
\end{equation}
Hence, the expected work content grows proportionally to the system size $m$.

Next, we determine the size of the fluctuations in the expected work extraction in terms of the quantity $\sigma(W_{\textrm{yield}}^m)$. By using the fact that
\begin{equation}
\sigma(q_aq_b\Vert r_ar_b)^2 = \sigma(q_a \Vert r_a)^2 +  \sigma(q_b \Vert r_b)^2,
\end{equation}
one finds 
\begin{equation}
\begin{split}
\sigma(W_{\textrm{yield}}^m) = \sqrt{m}kT\ln(2) \sigma\boldsymbol{(}q\Vert G(h)\boldsymbol{)}.
\end{split}
\end{equation}
Hence, as anticipated, the fluctuations only grow at the order of $\sqrt{m}$.

We furthermore wish to determine how $\mathcal{A}^{\epsilon}(q^{\otimes m},h^{\oplus m})$ scales with $m$.
For sufficiently small $\epsilon$ we know from Corollary \ref{AlmostDetermWorkCont} that this reduces to the question of how $D^{\epsilon}_{0}\boldsymbol{(}q^{\otimes m}\Vert G(h)^{\otimes m} \boldsymbol{)}$ scales with $m$. 
In Proposition \ref{nextorderasymptotics} (in Sec.~\ref{SecondOrderAEP}  below) we determine the next to leading order of the latter quantity, which yields
\begin{equation}
\label{ynvymvyn}
\begin{split}
\mathcal{A}^{\epsilon}(q^{\otimes m},h^{\oplus m}) = & mkT\ln(2)D\boldsymbol{(}q\Vert G(h)\boldsymbol{)} \\
& +  \sqrt{m}kT\ln(2)\Phi^{-1}(\epsilon)\sigma\boldsymbol{(}q\Vert G(h)\boldsymbol{)}\\
& + o(\sqrt{m}).
\end{split}
\end{equation}
Here $\Phi^{-1}$ denotes the inverse of the cumulative distribution function of the standard normal distribution $\Phi(x) := \int_{-\infty}^{x} e^{-x^2/2}dx/\sqrt{2\pi}$.
Hence, to the leading order, the $\epsilon$-deterministic work content is the same as the expected work content. The difference only shows up at the next to leading order, where the work yield is lowered by $\sqrt{m}kT\ln(2)\Phi^{-1}(\epsilon)\sigma\boldsymbol{(}q\Vert G(h)\boldsymbol{)}$. Note that $\Phi^{-1}(\epsilon)<0$ for $\epsilon<1/2$, and $\Phi^{-1}(\epsilon) \rightarrow -\infty$ for $\epsilon\rightarrow 0$.

A consequence of Eq.~(\ref{ynvymvyn}) is that 
\begin{equation*}
\lim_{m\rightarrow \infty} \frac{1}{m}\mathcal{A}^{\epsilon}(q^{\otimes m},h^{\oplus m}) = kT\ln(2)D\boldsymbol{(}q\Vert G(h)\boldsymbol{)}.
\end{equation*}
This can alternatively be obtained as a special case of the asymptotic rate of interconversion of quantum states that was proved in \cite{Brandao11} in a resource theory framework.


\subsection{\label{SecondOrderAEP} A next to leading order AEP for $D^{\epsilon}_{0}$}

Here we determine the asymptotic expansion in $m$ of $D^{\epsilon}_{0}(q^{\otimes m}\Vert r^{\otimes m})$ up to the next to leading order. (See \cite{Tomamichel09, Holenstein11} for  leading order expansions for conditional entropies.) Note also very recent results in \cite{Tomamichel12} concerning second order expansions of quantum entropies.

In classical information theory the concept of relative entropy typical sequences is introduced. This concept stems from the asymptotic equipartition property \cite{CoverThomas}, which in turn essentially is an application of the law of large numbers. As described in Sec.~11.8 of \cite{CoverThomas},  a sequence $(n_1,\ldots, n_m)\in \{1,\ldots, N\}^{\times m}$ is called relative entropy typical if $D(q\Vert r) -\epsilon \leq \log_2[q(n_1)\cdots q(n_m)/ r(n_1)\cdots r(n_m)] \leq D(q\Vert r) +\epsilon$. One can attempt to determine the expansion using this construction.  Properties 2 and 3 in Theorem 11.8.2 in \cite{CoverThomas} yields the upper bound, and Lemma 11.8.1 in \cite{CoverThomas} the lower bound in 
\begin{equation*}
\begin{split}
nD(q\Vert r) - n\epsilon < & D^{\epsilon}_{0}(q^{\otimes n}\Vert r^{\otimes n})\\
 < & -\log_2(1-2\epsilon) + nD(q\Vert r) + n\epsilon.
\end{split}
\end{equation*}
With these bounds one can prove that $\lim_{\epsilon\rightarrow 0}\lim_{n\rightarrow\infty} \frac{1}{n}D^{\epsilon}_{0}(q^{\otimes n}\Vert r^{\otimes n}) = D(q\Vert r)$. However, they are not strong enough to show that $D^{\epsilon}_{0}(q^{\otimes n}\Vert r^{\otimes n}) \sim nD(q\Vert r)$ for fixed $\epsilon$. As a second attempt, one can construct two sets of sequences (the ones in Def.~\ref{Defupperlower}) that are related to the central limit theorem rather than the law of large numbers.  Via the central limit theorem one can use these two sets to prove that the leading order in the expansion is $D^{\epsilon}_{0}(q^{\otimes n}\Vert r^{\otimes n}) \sim nD(q\Vert r)$. However, since this is not quite enough for our purposes we will not consider this proof here. To obtain also the next to leading order in the expansion, we take one step further, so to speak, and use Berry-Esseen's theorem, which bounds the rate of convergence in the central limit theorem.

We let $\Phi(y) := \int_{-\infty}^{y}e^{-x^2/2}/\sqrt{2\pi}dx$ denote the cumulative distribution function of the standard normal distribution.
Due to Berry \cite{Berry} and Esseen \cite{Esseen} we know the following:
\begin{Theorem}[Berry-Esseen \cite{Berry, Esseen}]
Let $Y_{1},\ldots, Y_{m}$ be iid random variables such that $\mu  :=  \langle Y\rangle$ exists, $\sigma^2 :=  \langle Y^{2}\rangle -\mu^2$ exists, with $\sigma>0$, and $\rho :=  \langle  |Y-\mu|^{3} \rangle <+\infty$.
Then,
\begin{equation*}
\Big| P\Big( \frac{\sqrt{m}}{\sigma} \Big[\frac{1}{m}\sum_{l=1}^{m}Y_{l}  -\mu\Big]  \leq y\Big) -\Phi(y)\Big| \leq \frac{C\rho}{\sigma^{3}\sqrt{m}},
\end{equation*}
for all $y\in \mathbb{R}$.
\end{Theorem}
Note that $C$ is a positive constant, independent of $y$ and independent of the distribution of $Y$. The exact value of this constant is to date not known, but there exist bounds \cite{Korolev} (see also, e.g., chapter 7 in \cite{Gut}). 

Given a random variable $X$ with distribution $q$, and given another distribution $r$, we define (analogous to $\sigma(q\Vert r)$ in Eq.~(\ref{defrelsigma})) the quantity
\begin{equation*}
\rho(q\Vert r) :=  \left\langle \left|\log_{2}\frac{q(X)}{r(X)}-\left\langle \log_{2}\frac{q(X)}{r(X)}\right\rangle\right|^{3}\right\rangle.
\end{equation*}
The following definition specifies two sets of sequences that take the role of the set of typical sequences described above. Note though, that in the way we will use these two sets, only one of them will  correspond to typical sequences, while the other set actually will correspond to very atypical sequences.
\begin{Definition}
\label{Defupperlower}
Let $q\in \mathbb{P}(N)$ and $r\in \mathbb{P}^{+}(N)$ be such that $\sigma(q\Vert r)>0$. For $m\in\mathbb{N}_{+}$ and $x\in\mathbb{R}$ define
\begin{equation}
\begin{split}
 \underline{\Lambda}^m_{x}  :=  & \Big\{ (n_{1},\ldots,n_{m})\in \{1,\ldots, N\}^{\times m}:\\
 &  2^{x \sigma \sqrt{m}  + m\mu}   <\frac{q(n_1)}{r(n_1)}\cdots \frac{q(n_{m})}{r(n_{m})} \Big\},
\end{split}
\end{equation}
where $\mu: = D(q\Vert r)$ and $\sigma:= \sigma(q\Vert r)$.
We furthermore denote the complementary set as
\begin{equation}
\begin{split}
\label{uppertypical}
\overline{\Lambda}^{m}_{x}  := & \{1,\ldots, N\}^{\times m}\setminus \underline{\Lambda}^{m}_{x}\\
= & \Big\{ (n_{1},\ldots,n_{m})\in \{1,\ldots, N\}^{\times m}:   \\
 & \frac{q(n_1)}{r(n_1)}\cdots \frac{q(n_{m})}{r(n_{m})}   \leq 2^{x \sigma \sqrt{m}  + m\mu}\Big\}.
\end{split}
\end{equation}
\end{Definition}

A direct application of Berry-Esseen's inequality, with the choice $Y_{l} = \log_2[q(X_l)/r(X_l)]$, with $X_{l}$ iid distributed $q$, yields
\begin{Lemma}
\label{ymvmvls}
Let $q\in \mathbb{P}(N)$ and $r\in \mathbb{P}^{+}(N)$ be such that $\sigma(q\Vert r)>0$. Let $m\in\mathbb{N}_+$. Then
\begin{equation}
\label{bnfgjksn}
|q^{\otimes m}(\underline{\Lambda}_y^m) -1+\Phi(y)| \leq \frac{C\rho}{\sigma^3\sqrt{m}},\quad \forall y\in\mathbb{R},
\end{equation}
\begin{equation}
\label{ycmvx}
|q^{\otimes m}(\overline{\Lambda}_x^m) -\Phi(x)| \leq \frac{C\rho}{\sigma^3\sqrt{m}},\quad \forall x\in\mathbb{R},
\end{equation}
where $\mu: = D(q\Vert r)$, $\sigma:= \sigma(q\Vert r)$, and $\rho:=\rho(q\Vert r)$.
\end{Lemma}

The general aim of this section is to prove the next to leading order expansion of $D^{\epsilon}_{0}(q^{\otimes m}\Vert r^{\otimes m})$ as in Eqs.~(\ref{leadingorder}) and (\ref{asymptotics}).
The proof is split into three parts: Lemmas \ref{lowerboundD0}, \ref{upperboundD0}, and Proposition \ref{nextorderasymptotics}.
The general idea is to construct  sequences of upper and  lower bounds to $D^{\epsilon}_{0}(q^{\otimes m}\Vert r^{\otimes m})$.
The lower bounds will be obtained from a sequence of sets $\underline{\Lambda}_{y(m)}^{m}$ such that  $q^{\otimes m}(\underline{\Lambda}_{y(m)}^{m})>1-\epsilon$. In other words $\underline{\Lambda}_{y(m)}^{m}$ is sufficiently likely with resect to $q^{\otimes m}$, and thus  
 $D^{\epsilon}_{0}(q\Vert r) = -\log_2\inf_{q^{\otimes m}(\Omega)>1-\epsilon}r^{\otimes m}(\Omega) \geq -\log_2 r^{\otimes m}(\underline{\Lambda}^m_y)$. 
Furthermore, the sequence is such that $q^{\otimes m}(\underline{\Lambda}_{y(m)}^{m})   \rightarrow 1-\epsilon$, and $\underline{\Lambda}_{y(m)}^{m}$ thus becomes a set of typical sequences for small $\epsilon$. 

Concerning the upper bounds we note that  the definition $D^{\epsilon}_{0}(q\Vert r) = -\log_2\inf_{q^{\otimes m}(\Omega)>1-\epsilon}r^{\otimes m}(\Omega)$ suggests that a method to obtain an upper bound is to search among sets $\Omega'$ with $q^{\otimes m}(\Omega')<1-\epsilon$, i.e., among less likely sets. As it so happens, we achieve the upper bound via a sequence of very \emph{atypical} sets $\overline{\Lambda}_{x(m)}^{m}$ in the sense that    $q^{\otimes m}(\overline{\Lambda}_{x(m)}^{m}) \rightarrow \epsilon$. 

As a side remark we note that for similar proofs for the smooth conditional max-entropy (albeit in the quantum case) \cite{Tomamichel09}  an upper bound can be obtained via Fannes' inequality \cite{Fannes}. There is indeed an analogue of Fannes' inequality for the relative Shannon entropy. However, the resulting upper bound appears not to be strong enough to establish the next to leading order expansion coefficient. This is the reason why we rather use the sets $\overline{\Lambda}_{x(m)}^{m}$.

\begin{Lemma}
\label{lowerboundD0}
Let $q\in \mathbb{P}(N)$ and $r\in \mathbb{P}^{+}(N)$ be such that $\sigma(q\Vert r)>0$,
and let $1> \epsilon >0$. Then 
\begin{eqnarray}
\label{mdvmvadm}
D^{\epsilon}_{0}(q^{\otimes m}\Vert r^{\otimes m}) & > &  mD(q\Vert r) + y \sigma\sqrt{m}  \\
\nonumber & &  -\log_2\Big[1-\Phi(y) + \frac{C\rho}{\sigma^{3}\sqrt{m}}\Big],
\end{eqnarray} 
for all pairs $y\in \mathbb{R}, m\in\mathbb{N}_{+}$ such that 
\begin{equation}
\label{condlowerboundD0}
\epsilon > \frac{C\rho}{\sigma^{3}\sqrt{m}} + \Phi(y),
\end{equation} 
where $\sigma:= \sigma(q\Vert r)$ and $\rho:=\rho(q\Vert r)$.
\end{Lemma}
\begin{proof}
Let $m\in \mathbb{N}_{+}$ and $y\in\mathbb{R}$. By the defining properties of the set $\underline{\Lambda}_y^{m}$ in Def.~\ref{Defupperlower}, it follows that
\begin{equation*}
\begin{split}
 r^{\otimes m}(\underline{\Lambda}^m_y) <  & 2^{-y\sigma \sqrt{m} - m\mu}  q^{\otimes m}(\underline{\Lambda}^m_y)\\
\leq & 2^{-y\sigma \sqrt{m} - m\mu}  \Big[1-\Phi(y) + \frac{C\rho}{\sigma^{3}\sqrt{m}}\Big],
\end{split}
\end{equation*}
where the second inequality is due to Eq.~(\ref{bnfgjksn}) in Lemma \ref{ymvmvls}.
Let us now restrict the pair $m\in \mathbb{N}_{+}$ and $y\in\mathbb{R}$ such that $\epsilon > \frac{C\rho}{\sigma^{3}\sqrt{m}} + \Phi(y)$. By  Eq.~(\ref{bnfgjksn}) in Lemma \ref{ymvmvls} it follows that 
\begin{equation}
q^{\otimes m}(\underline{\Lambda}^m_y)\geq 1- \Phi(y) -\frac{C\rho}{\sigma^{3}\sqrt{m}} >1-\epsilon.
\end{equation}
Hence, we can conclude that 
\begin{equation*}
\inf_{\Omega: q^{\otimes m}(\Omega)>1-\epsilon}r^{\otimes m}(\Omega) \leq  r^{\otimes m}(\underline{\Lambda}^m_y).
\end{equation*}
Since $1-\Phi(y) + \frac{C\rho}{\sigma^{3}\sqrt{m}}>0$ the statement of the lemma follows.
\end{proof}

\begin{Lemma}
\label{upperboundD0}
Let $q\in \mathbb{P}(N)$ and $r\in \mathbb{P}^{+}(N)$ be such that $\sigma(q\Vert r)>0$,
and let $1> \epsilon >0$. Then 
\begin{eqnarray}
 \label{nfblmfbx}
 D^{\epsilon}_{0}(q^{\otimes m}\Vert r^{\otimes m}) & \leq &  mD(q\Vert r) + x \sigma\sqrt{m} \\
\nonumber & &  -\log_2\Big[\Phi(x) -\epsilon- \frac{C\rho}{\sigma^{3}\sqrt{m}}\Big],
\end{eqnarray} 
for all pairs $x\in\mathbb{R}, m\in \mathbb{N}_{+}$ such that 
\begin{equation}
\label{condUpperboundD0}
\Phi(x) > \epsilon + \frac{C\rho}{\sigma^{3}\sqrt{m}},
\end{equation} 
where $\sigma:= \sigma(q\Vert r)$ and $\rho:=\rho(q\Vert r)$.
\end{Lemma}
The proof is  similar in spirit to the proof of  Lemma 11.8.1 in \cite{CoverThomas}.
\begin{proof}
Let $\Omega\subseteq \{1,\ldots, N\}^{\otimes m}$ be such that $q^{\otimes m}(\Omega)>1-\epsilon$.
By Eq.~(\ref{ycmvx}) in Lemma \ref{ymvmvls}, the set $\overline{\Lambda}_x^m$  satisfies $q^{\otimes m}(\overline{\Lambda}_x^m)\geq \Phi(x)-\frac{C\rho}{\sigma^3\sqrt{m}}$, for all pairs $x\in\mathbb{R}$ and $m\in\mathbb{N}_+$.
Consequently 
\begin{equation}
q^{\otimes m}(\Omega\cap\overline{\Lambda}^m_x) >  \Phi(x)-\epsilon -\frac{C\rho}{\sigma^{3}\sqrt{m}}.
\end{equation}
By combining this with the defining condition for the set $\overline{\Lambda}^m_x$ it follows  that
\begin{eqnarray*}
\nonumber r^{\otimes m}(\Omega) & \geq & r^{\otimes m}(\Omega\cap\overline{\Lambda}^m_x)\\
 \nonumber                 &  \geq &   2^{-x \sigma \sqrt{m}  - m\mu}q^{\otimes m}(\Omega\cap \overline{\Lambda}^m_x)\\
 \label{msdvmn}                  & > & 2^{-x \sigma \sqrt{m}  - m\mu}\Big[ \Phi(x)-\epsilon -\frac{C\rho}{\sigma^{3}\sqrt{m}}\Big],
\end{eqnarray*}
where $\mu = D(q\Vert r)$.
Assuming $\Phi(x)> \epsilon +\frac{C\rho}{\sigma^{3}\sqrt{m}}$, we can take the logarithm of both sides of the above inequality, and obtain the bound in Eq.~(\ref{nfblmfbx}) by taking the infimum of $\log_2r^{\otimes m}(\Omega)$ over all $\Omega$ such that $q^{\otimes m}(\Omega)>1-\epsilon$.
\end{proof}

\begin{Proposition}
\label{nextorderasymptotics}
Let $q\in \mathbb{P}(N)$ and $r\in \mathbb{P}^{+}(N)$, and let $1>\epsilon>0$. Then
\begin{equation}
\label{leadingorder}
\lim_{m\rightarrow\infty}\frac{D^{\epsilon}_{0}(q^{\otimes m}\Vert r^{\otimes m})}{m} = D(q\Vert r) 
\end{equation}
and
\begin{equation}
\label{asymptotics}
\lim_{m\rightarrow\infty}\frac{D^{\epsilon}_{0}(q^{\otimes m}\Vert r^{\otimes m})-mD(q\Vert r)}{\sqrt{m}} = \Phi^{-1}(\epsilon)\sigma(q\Vert r).
\end{equation}
\end{Proposition}
Eqs.~(\ref{leadingorder}) and (\ref{asymptotics}) implies that we have obtained the next to leading order expansion
\begin{equation}
\begin{split}
D^{\epsilon}_{0}(q^{\otimes m}\Vert r^{\otimes m}) = &  mD(q\Vert r) + \sqrt{m}\Phi^{-1}(\epsilon)\sigma(q\Vert r) \\
& + o(\sqrt{m}).
\end{split}
\end{equation}

\begin{proof}
We separate the two cases $\sigma(q\Vert r)=0$ and $\sigma(q\Vert r)\neq 0$.

Case $\sigma(q\Vert r)=0$: We first make a general observation concerning the properties of $\sigma$ and $D$. Given $q'\in \mathbb{P}(N)$ and $r'\in \mathbb{P}^{+}(N)$ it follows that $\sigma(q'\Vert r')=0$ if and only if there is a constant $1\geq c>0$ such that $r'_{n} = cq'_{n}$, for all $n$ in the support of $q'$. In this case it furthermore follows that $D(q'\Vert r') = -\log_2c$, and $D(q'\Vert r')\leq D^{\epsilon}_{0}(q'\Vert r') \leq  D(q'\Vert r') -\log_2(1-\epsilon)$.
Since $\sigma(q^{\otimes m}\Vert r^{\otimes m})^{2} = m\sigma(q\Vert r)^{2}$, it follows that $\sigma(q\Vert r)=0$ if and only if $\sigma(q^{\otimes m}\Vert r^{\otimes m}) = 0$.
By the above comments it follows that $mD(q\Vert r)\leq D^{\epsilon}_{0}(q^{\otimes m}\Vert r^{\otimes m}) \leq  mD(q\Vert r) -\log_2(1-\epsilon)$. This yields Eq.~(\ref{leadingorder}). Furthermore, 
one can see that Eq.~(\ref{asymptotics}) holds trivially in this case.

Case $\sigma(q\Vert r)\neq 0$: Let us choose $y$ in Lemma \ref{lowerboundD0} as
\begin{equation}
\label{ychoice}
y(m) := \Phi^{-1}\Big(\epsilon - \frac{C\rho}{\sigma^3\sqrt{m}} - \frac{1}{m}\Big).
\end{equation}
For sufficiently large $m$ we have $1>\epsilon - \frac{C\rho}{\sigma^3\sqrt{m}} - \frac{1}{m}>0$, and $y(m)$ thus well defined.
Furthermore, $\frac{C\rho}{\sigma^{3}\sqrt{m}} + \Phi\boldsymbol{(}y(m)\boldsymbol{)} = \epsilon - \frac{1}{m}<\epsilon$. Hence, for sufficiently large $m$ the conditions of Lemma \ref{lowerboundD0} are satisfied, and we find
\begin{equation}
\label{lowerynkvdf}
\begin{split}
D^{\epsilon}_{0}(q^{\otimes m}\Vert r^{\otimes m})  > & mD(q\Vert r) \\
  &  -\log_2\Big[1-   \epsilon  + \frac{1}{m} + 2\frac{C\rho}{\sigma^{3}\sqrt{m}}\Big]\\
   & +  \sigma\sqrt{m}   \Phi^{-1}\Big(\epsilon - \frac{C\rho}{\sigma^3\sqrt{m}} - \frac{1}{m}\Big).
 \end{split}
\end{equation} 
Let us choose $x$ in Lemma \ref{upperboundD0} as
\begin{equation}
\label{xchoice}
x(m) := \Phi^{-1}\Big(\epsilon + \frac{C\rho}{\sigma^3\sqrt{m}} + \frac{1}{m}\Big).
\end{equation}
For sufficiently large $m$ it is the case that $1>\epsilon + \frac{C\rho}{\sigma^3\sqrt{m}} + \frac{1}{m}>0$. Hence, $x(m)$  is well defined for sufficiently large $m$. Furthermore $\Phi\boldsymbol{(}x(m)\boldsymbol{)} = \epsilon + \frac{C\rho}{\sigma^3\sqrt{m}} + \frac{1}{m} > \epsilon + \frac{C\rho}{\sigma^3\sqrt{m}}$. Hence, the condition of Lemma \ref{upperboundD0} is satisfied, and we can conclude that
\begin{equation}
\label{upperynkvdf}
\begin{split}
D^{\epsilon}_{0}(q^{\otimes m}\Vert r^{\otimes m}) \leq &  mD(q\Vert r)   +\log_2(m)\\
& + \sigma\sqrt{m}\Phi^{-1}\Big(\epsilon + \frac{C\rho}{\sigma^3\sqrt{m}} + \frac{1}{m}\Big). 
 \end{split}
\end{equation} 
By combining Eqs.~(\ref{lowerynkvdf}) and (\ref{upperynkvdf}) one can prove the limits  (\ref{leadingorder}) and (\ref{asymptotics}).
\end{proof}

Note that if Eq.~(\ref{bnfgjksn}) is combined with Eq.~(\ref{ychoice}) it follows that $ q^{\otimes m}(\underline{\Lambda}_{y(m)}^m)\rightarrow 1-\epsilon$. Hence, for small $\epsilon$ the family of sets $\underline{\Lambda}_{y(m)}^m$ becomes typical.
Similarly, Eq.~(\ref{ycmvx}) with Eq.~(\ref{xchoice}) yields $q^{\otimes m}(\overline{\Lambda}_{x(m)}^m)\rightarrow \epsilon$. Hence the family $\overline{\Lambda}_{x(m)}^m$ becomes atypical.


\subsection{\label{expecvssingleshotiid}The expectation value setting vs. the iid single-shot case}

At first sight it might seem as the expectation value setting and the single-shot setting with iid states and non-interacting identical Hamiltonians describes the very same thing. However, this is not quite the case, and here we briefly point out the main difference. 

The expected work cost $\langle W(\mathcal{P},\mathcal{N})\rangle$ can be obtained as a sample average, i.e., we repeat the process (assuming access to an iid source of the system) and take the average of the recorded work costs. 
An alternative to this `serial' procedure would be to consider a `parallel' implementation on $K$ iid copies of the initial distribution, where we  naturally associate non-interacting and identical Hamiltonians to this collective system. On this system we record one collective work cost of the process, and divide the resulting work cost by $K$, thus obtaining a sample average to approximate the true expectation value. Apart from the division by $K$, this appears identical to what we did in Sec.~\ref{Sec:OptimalExpected}. Indeed, the initial distribution and the Hamiltonian are of course the very same. However, apart from question of what exactly it is we optimize, there is one significant difference. Namely, that in the single-shot iid case we allow arbitrary \emph{collective} processes $\mathcal{P}^{(m)}$ over the different copies of the systems, while in the parallel implementation of the expectation value setting we are restricted to \emph{iid processes} $\mathcal{P}^{\otimes m}$, where $\mathcal{P}$ is a single-copy process. In other words, in the single-shot iid setting we have much more freedom in the optimization procedure than we do in the corresponding expectation value case, as we can span a much larger class of processes. Although at first maybe tempting, we should thus not naively identify the expectation value setting with the iid single-shot case.


\subsection{\label{correlations} A class of state distributions}

As before we consider a collection of $m$ systems, e.g. $m$ $d$-level systems. For each $m$ we let $h^m\in\mathbb{R}^{d^m}$ be the set of energy levels of the total system. For each $m$ we choose a specific state $x(m)\in\{1,\ldots, d^{m}\}$. (This could be, e.g., the ground state of $h^{m}$.) In the following we shall assume that $h^{m}$ and $x(m)$ are chosen such that 
\begin{equation}
\label{nvsdkndv}
\lim_{m\rightarrow \infty} G_{x(m)}(h^{m}) = 0.
\end{equation}
Note that this is not a particularly strict condition.
For example, consider the special case of a collection of $m$ non-interacting systems with identical Hamiltonians, i.e., $h^{m} = h^{\oplus m}$, for $h\in\mathbb{R}^d$. For a specific element $s\in\{1,\ldots,d\}$, we let $x(m) := (s,\ldots,s)$, in which case 
$G_{x(m)}(h^{\oplus m}) = G_{s}(h)^{m}$. Since $G_{s}(h)<1$ (due to the fact that $h\in \mathbb{R}^{d}$) it follows that $\lim_{m\rightarrow \infty} G_{x(m)}(h^{\oplus m}) = 0$. Hence, in this case, the condition (\ref{nvsdkndv}) is always satisfied. (In the following we do not assume $h^{m} = h^{\oplus m}$, but only the condition (\ref{nvsdkndv}).)

Turning to the initial distribution $q^{m}$, we let $0\leq \nu \leq 1$ be independent of $m$ and define
\begin{equation}
q^{m}_{l}  := (1-\nu)\delta_{l,x(m)} +  \nu G_l(h^{m}). 
\end{equation}
In other words, with probability $1-\nu$ the system is in state $x(m)$, and with probability $\nu$ the system is Gibbs distributed. A direct calculation yields
\begin{equation*}
\begin{split}
\mathcal{A}(q^{m}, h^{m}) \sim  - (1-\nu ) \frac{\ln 2}{\beta} \log_{2}G_{x(m)}(h^{m}),
\end{split}
\end{equation*}
where we have made use of the assumption $\lim_{m\rightarrow \infty} G_{x(m)}(h^{m}) = 0$ (and hence $\lim_{m\rightarrow \infty}[ -\log_2G_{x(m)}(h^{m})] = +\infty$). Similarly, one finds 
\begin{equation*}
\begin{split}
\sigma(W^{m}_{\textrm{yield}})  = & \frac{\ln 2}{\beta} \sigma\boldsymbol{(}q^{m}\Vert G(h^{m})\boldsymbol{)}\\
\sim &   -\frac{\ln 2}{\beta} \sqrt{\nu(1  - \nu)}\log_{2}G_{x(m)}(h^{m}).
\end{split}
\end{equation*}
In the special case $\nu = 1/2$ we thus find that  $\mathcal{A}(q^{m},h^{m})$ and $\sigma(W^{m}_{\textrm{yield}})$ scale identically.

Let us now instead chose $\nu := \epsilon$ for a sufficiently small $\epsilon>0$, and compare the above with $\mathcal{A}^{\epsilon}(q^{m},h^{m})$.
By definition $D^{\epsilon}_{0}\boldsymbol{(}q^{m}\Vert G(h^{m})\boldsymbol{)}
=   -\log_{2}\min_{q^{m}(\Lambda^{m})> 1-\epsilon}\sum_{l\in \Lambda^{m}} G_{l}(h^{m})$,
 where $q^{m}(\Lambda^{m}) := \sum_{l\in \Lambda^{m}}q^{m}_{l}$. Assuming $\epsilon \leq 1/2$, the condition $q^{m}(\Lambda^{m})> 1-\epsilon$ and the choice $\nu = \epsilon$ implies that $x(m)\in \Lambda^{m}$, due to the construction of $q^{m}$. Furthermore, since $h^{m}\in \mathbb{R}^{d^m}$ it follows that  $G_{x(m)}(h^{m}) > 0$. Hence,  $q_{x(m)}(h^{m}) = 1-\epsilon + \epsilon G_{x(m)}(h^{m}) > 1-\epsilon$.
We can conclude that $x(m)$ must be an element in $\Lambda^{m}$, and that no other level has to be an element. Hence, $\{x(m)\}$ is the minimizing set, and thus $D^{\epsilon}_{0}(q^{m},h^{m}) = -\log_2 G_{x(m)}(h^{m})$.

Furthermore, due to the relation
$0 \leq \mathcal{A}^{\epsilon}(q^{m},h^{m}) -kT\ln(2)D^{\epsilon}_{0}\boldsymbol{(}q^{m}\Vert G(h^{m})\boldsymbol{)} \leq -kT\ln(1-\epsilon)$, it follows that 
\begin{equation*}
\begin{split}
 \mathcal{A}^{\epsilon}(q^{m},h^{m})\sim  -kT\ln(2)\log_{2}G_{x(m)}(h^{m}).
\end{split}
\end{equation*}

We obtain the special case mentioned in the main text if we consider a collection of non-interacting identical systems, $h^{m} = h^{\oplus m}$, and $x(m) = (s,\ldots, s)$. In this case $\mathcal{A}(q^{m}, h^{m})  \sim  - (1-\epsilon)mkT\ln(2) \log_{2}G_{s}(h)$, $\sigma(W^{m}_{\textrm{yield}}) \sim    -mkT\ln(2) \sqrt{\epsilon(1  - \epsilon)}\log_{2}G_{s}(h)$, and $\mathcal{A}^{\epsilon}(q^{m},h^{m})\sim  -mkT\ln(2)\log_{2}G_{s}(h)$. Hence, a linear scaling in $m$ for all three cases.


\subsection{\label{interaction} Maximally mixed state distribution}
As before, we consider $m$ $d$-level systems, but with a maximally mixed distribution of its initial state, i.e., $q^{m}_{l_1,\ldots,l_m} := d^{-m}$ for all $l_1,\ldots, l_m$.
 For the sake of simplicity, and without loss of generality, we may shift the set of energy levels such that $\sum_{l_1,\ldots,l_m} h^{m}(l_{1},\ldots, l_{m}) =0$.
In this case we have 
\begin{equation}
\label{mdglfnk1}
\mathcal{A}(q^{m},h^{m}) =    -kTm\ln(d)  -F(h^{m}).
\end{equation}
By Eqs.~(\ref{mvydlkm}) and (\ref{stddeviation}) it follows that 
\begin{equation}
\label{mdglfnk2}
\sigma^2(W^{m}_{\textrm{yield}}) =  d^{-m}\sum_{l_{1},\ldots,l_{m}}(h^{m}_{l_1,\ldots,l_m})^2.
\end{equation} 
Note that Eq.~(\ref{mdglfnk2}) does not contain any factor $kT$, since the factor $kT$ in Eq.~(\ref{stddeviation}) is canceled by a corresponding factor from the Gibbs distribution.

Since all states are equally likely, the minimization in the definition of $D^{\epsilon}_{0}$ is much simplified, as the condition $q^{m}(\Lambda)>1-\epsilon$ reduces to $|\Lambda|> (1-\epsilon)d^{m}$.
As a consequence
\begin{equation*}
\begin{split}
& D^{\epsilon}_{0}\boldsymbol{(}q^{m}\Vert G(h^{m})\boldsymbol{)} =   -\frac{\beta}{\ln(2)} F(h^{m})\\
&\quad\quad\quad\quad\quad  -\log_{2}\min_{ |\Lambda|> (1-\epsilon)d^{m}}\sum  e^{-\beta h^{m}_{l_1,\ldots,l_m} }.
\end{split}
\end{equation*}
Moreover, if we sort the energy levels in a non-decreasing order $h^{m}_1 \leq h^{m}_2\leq \cdots$, the above minimum is obtained if we remove sufficiently many of the levels of lowest energy. More precisely, 
\begin{equation}
\label{abvjvsb}
\begin{split}
D^{\epsilon}_{0}\boldsymbol{(}q^{m}\Vert G(h^{m})\boldsymbol{)}  = &  -\frac{\beta}{\ln(2)}F(h^{m})\\
& -\log_{2}\sum_{s=S+1}^{d^{m}}  e^{-\beta h^{m}_s},
\end{split}
\end{equation}
where $S$ is the largest integer such that $\epsilon d^{m} > S$.

To simplify the calculations we will in the following make an assumption on the family of energy levels $\{h^m_n\}_{n=1}^{d^m}$ and how they depend on $m$. Namely, for a sufficiently well behaved functions $g$ (we will in fact only use $x$, $x^2$, and $e^{-\beta x}$) we assume that 
\begin{equation}
\label{Assumption}
\frac{1}{d^m}\sum_{n=1}^{d^m}g(h^m_n) \sim \int_{-\infty}^{+\infty} g(x)f^{(m)}(x)dx,
\end{equation} 
where $f^{(m)}(x)\geq 0$ and $\int_{-\infty}^{+\infty}f^{(m)}(x)dx = 1$. (Or similarly, with modified limits of the summation and integration.)
In other words, we assume that in the limit of large $m$ the collection of energy levels can be replaced with a spectral density function. 

Applying this to Eq.~(\ref{mdglfnk1}) yields
\begin{equation}
\label{ueirbvaiie}
\mathcal{A}(q^{m},h^{m})  \sim    kT\ln\int_{-\infty}^{+\infty} e^{-\beta x} f^{(m)}(x)dx,
\end{equation} 
and similarly for Eq.~(\ref{mdglfnk1})
\begin{equation}
\sigma^2(W^{m}_{\textrm{yield}}) \sim   \int_{-\infty}^{+\infty} x^2 f^{(m)}(x)dx.
\end{equation} 
Next we consider the $\epsilon$-deterministic work content. 
If we let $F^{(m)}(x) = \int_{-\infty}^{x}f^{(m)}(x)dx$ denote the cumulative distribution function, then the condition $\epsilon d^{m} > S$ can be reformulated as $\epsilon = F^{(m)}(S)$. Assuming $F^{(m)}$ to be invertible we thus find  $S = {F^{(m)}}^{-1}(\epsilon)$. Hence, Eq.~(\ref{abvjvsb}) reduces to 
\begin{equation*}
\begin{split}
&  D^{\epsilon}_{0}\boldsymbol{(}q^{m}\Vert G(h^{m})\boldsymbol{)} \\
&\quad \sim   \frac{1}{\ln(2)}\ln\int_{-\infty}^{\infty} e^{-\beta x} f^{(m)}(x)dx \\
&\quad\quad  -\frac{1}{\ln(2)}\ln \int_{{F^{(m)}}^{-1}(\epsilon)}^{+\infty} e^{-\beta x} f^{(m)}(x)dx.
\end{split}
\end{equation*}

\subsubsection{The flat distribution}
As a simple example we consider the case of a flat distribution of the energy levels.
For  $a(m)>0$ we define
\begin{equation}
f^{(m)}(x) := \left\{ \begin{matrix} \frac{1}{2a(m)}, & |x|\leq a(m),\\
0, & |x|> a(m).
\end{matrix}\right.
\end{equation}
With this choice in Eq.~(\ref{ueirbvaiie}) we find
\begin{equation}
\mathcal{A}(q^{m},h^{m}) \sim a(m)
\end{equation}
and
\begin{equation}
\sigma(W^{m}_{\textrm{yield}}) \sim  \frac{1}{\sqrt{3}}a(m).
\end{equation}
Furthermore, since ${F^{(m)}}^{-1}(\epsilon) = (2\epsilon-1)a(m)$ it follows that 
\begin{equation*}
\begin{split}
D^{\epsilon}_0\boldsymbol{(}q^{m}\Vert G(h^{m})\boldsymbol{)} \sim  & 2\epsilon a(m) + \frac{1}{\beta}\ln \left(1 -e^{-2\beta a(m)} \right) \\
&  -\frac{1}{\beta}\ln \left(1-e^{-2\beta (1-\epsilon)a(m)}\right).
\end{split}
\end{equation*}

For $\epsilon \leq 1-1/\sqrt{2}$ we can now use the fact that $0\leq \mathcal{A}^{\epsilon}(q^{m},h^{m}) -kT\ln(2)D^{\epsilon}_0\boldsymbol{(}q^{m}\Vert G(h^{m})\boldsymbol{)} \leq -kT\ln(1-\epsilon)$ which yields
\begin{equation*}
\begin{split}
\mathcal{A}^{\epsilon}(q^{m},h^{m}) \sim 2\epsilon a(m).
\end{split}
\end{equation*}

\subsubsection{Wigner distribution}

Let $R(m)>0$ and define
\begin{equation*}
f^{(m)}(x) := \left\{\begin{matrix}
\frac{2}{\pi R(m)^2}\sqrt{R(m)^2 -x^{2}}, & |x|\leq R(m),\\
0, & |x|> R(m).
\end{matrix}
 \right.
\end{equation*}
By, e.g., Eq.~9 in Sect. 4.3.3.1 of \cite{HandbIntegr}, it follows that the above semi-circle law
 has the variance $\int_{-\infty}^{+\infty}x^2f^{(m)}(x)dx = R(m)^2/4$. Hence
\begin{equation}
\sigma(W^{m}_{\textrm{yield}}) \sim \frac{1}{2}R(m).
\end{equation}
Furthermore
\begin{equation*}
\begin{split}
\mathcal{A}(q^{m},h^{m})   \sim & \frac{1}{\beta}\ln I_1(m), \\
I_1(m) = & \frac{2}{\pi}\int_{-1}^{1}e^{-\beta R(m)y}\sqrt{1 -y^{2}}dy.
\end{split}
\end{equation*}
We determine upper and lower bounds to this integral, which gives us the asymptotic behavior.
On the interval $[-1,1]$ it is the case that $e^{-\beta R(m)y} \leq e^{\beta R(m)}$, which yields the upper bound $I_1(m)\leq e^{\beta R(m)}$.
Next, we determine a lower bound. To this end we we define the function
\begin{equation}
g(y) := \left\{ \begin{matrix}
y+1, & -1 \leq y\leq 0, \\
0, & \textrm{otherwise}.
\end{matrix}\right.
\end{equation}
We have $\sqrt{1 -y^{2}} \geq g(y)$, for $|y|\leq 1$. 
Hence
\begin{equation}
I_1(m) \geq \frac{2}{\pi}\frac{e^{\beta R(m)}}{\beta^2 R(m)^2}\Big(1  - [1+\beta R(m)] e^{-\beta R(m)}
      \Big).
\end{equation}
Combining the upper and lower bound we thus find 
\begin{equation}
\mathcal{A}(q^{m},h^{m})\sim R(m).
\end{equation}

Similarly as for  $\mathcal{A}$ we determine the asymptotic behavior of  $\mathcal{A}^{\epsilon}$ via upper and lower bounds. 
With $F(x) := \frac{2}{\pi}\int_{-1}^{x}\sqrt{1 -y^{2}}dy$, it follows that $F^{(m)}(z) = F\boldsymbol{(}z/R(m)\boldsymbol{)}$, and thus ${F^{(m)}}^{-1}(\epsilon) = R(m)F^{-1}(\epsilon)$.
We let 
\begin{equation*}
\begin{split}
 \mathcal{A}^{\epsilon}(q^m,h^m) \sim & \frac{1}{\beta}\ln I_1(m) - \frac{1}{\beta}\ln I_2(m)\\
 I_2(m) := & \int_{{F^{(m)}}^{-1}(\epsilon)}^{R(m)}  f^{(m)}(x)e^{-\beta x}dx 
\end{split}
\end{equation*}
and find $I_2(m) \leq  (1-  \epsilon)e^{-\beta R(m)F^{-1}(\epsilon)}$.

We again use the function $g$ to get a lower bound. Note that we here assume $\epsilon < 1/2$.
\begin{equation*}
\begin{split}
I_2(m) \geq  &  \frac{2}{\pi}\frac{e^{-\beta  R(m)F^{-1}(\epsilon) }}{\beta R(m)} \Big[1  -   F^{-1}(\epsilon)  -  e^{\beta R(m)F^{-1}(\epsilon)} \\
& + \frac{1}{\beta R(m)}e^{\beta R(m)F^{-1}(\epsilon)} -  \frac{1}{\beta R(m)}\Big].
\end{split}
\end{equation*}
We can conclude that 
\begin{equation}
\mathcal{A}^{\epsilon}(q^{m},h^{m})\sim c(\epsilon)R(m).
\end{equation}
where $c(\epsilon) := 1+F^{-1}(\epsilon)$.
Note that for sufficiently small $\epsilon$ we have 
\begin{equation}
\left(\frac{3\pi}{4\sqrt{2}}\epsilon \right)^{2/3} \leq c(\epsilon) \leq \left(\frac{3\pi}{4}\epsilon \right)^{2/3}.
\end{equation}


\section{\label{OtherCostFcns} Other cost functions, and the question of `self-erasure'}

In this investigation we have focussed on the expectation value and the $(\epsilon,\delta)$-deterministic value as cost functions. Clearly one could imagine other constructions.
Here we briefly point out one particular alternative, related to \cite{Dahlsten}. 

Given a real-valued random variable $X$, we can ask for the smallest upper bound to $X$ that is violated with probability smaller than $\epsilon$;  a very likely upper bound, so to speak.  More precisely
\begin{equation}
\textrm{Max}^{\epsilon}(X):=\inf\{x:P(X\leq x)>1-\epsilon\}.
\end{equation}
It is straightforward to see that 
\begin{equation}
\label{ineqDeltaMax}
\delta + \inf\!\Delta^{\epsilon}_{\delta}(X) \geq \textrm{Max}^{\epsilon}(X).
\end{equation}
Analogous to $\mathcal{C}^{\textrm{extr}}(q,h^{i},h^{f})$ and $\mathcal{C}^{\textrm{extr}}_{\epsilon}(q,h^i,h^f)$ we can use $\textrm{Max}^{\epsilon}$ as a cost function to define
\begin{equation*}
\mathcal{M}^{\textrm{extr}}_{\epsilon}(q,h^{i},h^{f}) := \inf_{\mathcal{P}\in\mathscr{P}(h^{i},h^{f})}\textrm{Max}^{\epsilon}\boldsymbol{(}W(\mathcal{P},\mathcal{N})\boldsymbol{)}.
\end{equation*}
This quantity is closely related to the approach in \cite{Dahlsten}.
In words, $\mathcal{M}^{\textrm{extr}}_{\epsilon}(q,h^{i},h^{f})$ is a kind of threshold quantity that asks for the minimal bound on the work cost that is violated with a probability at most $\epsilon$. Hence, we allow the work cost to be smaller than the value $\mathcal{M}^{\textrm{extr}}_{\epsilon}(q,h^{i},h^{f})$. This is to be compared with the $\epsilon$-deterministic value, where we require the extracted work to essentially precisely take the value $\mathcal{C}^{\textrm{extr}}_{\epsilon}(q,h^i,h^f)$, i.e., we neither allow larger nor smaller costs. Due to the threshold nature of $\mathcal{M}^{\textrm{extr}}_{\epsilon}(q,h^{i},h^{f})$  it is not clear how noisy an extraction process that yields the optimal value $\mathcal{M}^{\textrm{extr}}_{\epsilon}(q,h^{i},h^{f})$ could be. (Analogous to what we did in Sec.~\ref{fluctuations} for the optimal expected work extraction, we could ask for the noise properties of processes that achieve $\mathcal{M}^{\textrm{extr}}_{\epsilon}(q,h^{i},h^{f})$.) The problem is that the distribution of the work cost potentially could have a tail extending far below the threshold, which corresponds to some degree of unpredictability of this energy source. 

By Eq.~(\ref{ineqDeltaMax}) we can conclude that 
\begin{equation}
\mathcal{C}^{\textrm{extr}}_{\epsilon}(q,h^i,h^f) \geq \mathcal{M}^{\textrm{extr}}_{\epsilon}(q,h^{i},h^{f}).
\end{equation}
One could speculate whether $\mathcal{M}^{\textrm{extr}}_{\epsilon}(q,h^{i},h^{f})$ has almost the same value as $\mathcal{C}^{\textrm{extr}}_{\epsilon}(q,h^i,h^f)$ in general  thermodynamic models. (Whether they have the same value in the particular model we employ here is not clear.) The intuition is that we could `waste' energy in order to make the work cost variable more concentrated.
More precisely, imagine that we have found a process $\mathcal{P}$ such that $P\boldsymbol{(}W(\mathcal{P},\mathcal{N}) \leq w\boldsymbol{)}>1-\epsilon$, for some $w$.  Conditioned on the case that $W< w$, i.e., that the work cost is \emph{less} than $w$, one could imagine to make an additional dissipation of size $w-W$.  By this, all the weight of the probability distribution below $w$ is `shuffled' into a single peak at $w$. This would make $w$ into an $\epsilon$-deterministic value. However, this `wasting of energy' is conditional, i.e., the choice of process depends on the actual value of $W$ (as opposed to the processes generally considered in this investigation, which do not depend on the values of the random work variables per se). The analysis of such conditional processes would reasonably entail an explicit modeling of the control mechanisms that implements these conditional processes (regarding the conditional process as an unconditional process on a larger system). In view of Landauer's principle, we should not expect the resetting of this control mechanism to come for free. The question is, what does it cost us to `waste' energy? Another way to phrase the very same question would be to focus on the system that carries the extracted energy (which we do not model explicitly in this investigation). The `shuffling' of the work value described above, corresponds to a many-to-one map on the states of the energy reservoir, which again is an erasure process.
Loosely speaking, the question whether $\mathcal{M}^{\textrm{extr}}_{\epsilon}(q,h^{i},h^{f})$  coincides with $\mathcal{C}^{\textrm{extr}}_{\epsilon}(q,h^i,h^f)$ or not, can  be rephrased as the question whether the removal of the excess energy matches the cost of its own erasure. There are certainly many cases where systems appear to have this `self-erasure' property.  When the typical energy scale is much larger than $kT$, one can imagine scenarios when the system on a short time-scale more or less perfectly relaxes to some meta-stable state, thus implementing the above mentioned concentration of the distribution. The question is what happens in the general case. Since the model we use in this investigation does not include an explicit energy reservoir, we leave this as an open question.

\end{appendix}



\begin{thebibliography}{99}

\bibitem{ReviewFluctThm} C. Jarzynski, Equalities and Inequalities: Irreversibility and the second law of thermodynamics at the nanoscale. {\it Annu. Rev. Condens. Matter Phys.} {\bf 2}, 329 (2011). 
\bibitem{Procaccia} I. Procaccia and R. D. Levine, Potential work: A statistical-mechanical approach for systems in disequilibrium. {\it J. Chem. Phys.} {\bf 65}, 3357 (1967).
\bibitem{Lindblad1983} G. Lindblad, {\it Non-Equilibrium Entropy and Irreversibility} (Reidel, Lancaster, 1983).
\bibitem{Takara} K. Takara, H.-H. Hasegawa, and D. J. Driebe, Generalization of the second law for a transition between nonequilibrium states. {\it Phys. Lett. A} {\bf 375}, 88 (2010).
\bibitem{Esposito} M. Esposito and C. Van den Broeck, Second law and Landauer principle far from equilibrium.  {\it Euro. Phys. Lett.} {\bf 95}, 40004 (2011).
\bibitem{LeffRexI} H. S. Leff and A. F. Rex, {\it Maxwell's Demon: Entropy, Information, Computing} (Taylor and Francis, 1990).
\bibitem{LeffRexII} H. S. Leff and A. F. Rex, {\it Maxwell's Demon 2: Entropy, Classical and Quantum Information, Computing} (Taylor and Francis, 2002).
\bibitem{Dahlsten} O. Dahlsten, R. Renner, E. Rieper, and V. Vedral, Inadequacy of von Neumann entropy for characterizing extractable work. {\it New Journal of Physics}, {\bf 13}, 053015 (2011).
\bibitem{delRio} L. del Rio, J. {\AA}berg, R. Renner, O. Dahlsten, and V. Vedral, The thermodynamic meaning of negative entropy. {\it Nature} {\bf 474}, 61 (2011).
\bibitem{Faist} P. Faist, F. Dupuis, J. Oppenheim, and R. Renner, A quantitative Landauer's principle, arXiv:1211.1037 (2012).
\bibitem{Horodecki11} M. Horodecki and J. Oppenheim, Fundamental limitations for quantum and nano thermodynamics.  arXiv:1111.3834 (2011).
\bibitem{Egloff} D. Egloff, O. Dahlsten, R. Renner, and V. Vedral, Laws of thermodynamics beyond the von Neumann regime. arXiv:1207.0434 (2012).
\bibitem{EgloffThesis} D. Egloff, Master's thesis, ETH Zurich (2010).
\bibitem{CoverThomas} T. M. Cover and J. A. Thomas, {\it Elements of Information Theory, 2nd Ed.} (Wiley, Hoboken, 2006).
\bibitem{Shizume} K. Shizume, Heat generation required by information erasure. {\it Phys. Rev. E} {\bf 52}, 3495 (1995). 
\bibitem{Piechocinska} B. Piechocinska, Information erasure. {\it Phys. Rev. A} {\bf 61}, 062314 (2000).
\bibitem{Crooks98} G. E. Crooks, Nonequilibrium Measurements of free energy differences for microscopically reversible markovian systems. {\it J. Stat. Phys.} {\bf 90}, 1481 (1998).
\bibitem{CrooksTheorem} G. E. Crooks, Entropy production fluctuation theorem and the nonequilibrium work relation for free energy differences. {\it Phys. Rev. E} {\bf 60}, 2721 (1999). 
\bibitem{Wang1} L. Wang, R. Colbeck, and R. Renner, Simple channel coding bounds. {\it IEEE International Symposium on Information Theory, ISIT 2009}, 1804 (2009).
\bibitem{Datta} N. Datta, Min- and max-relative entropies and a new
entanglement monotone.  {\it IEEE Trans. Inf. Theor.} {\bf 55}, 2816 (2009).
\bibitem{Wang2} L. Wang and R. Renner, One-shot classical-quantum capacity and hypothesis testing.   {\it Phys. Rev. Lett.} {\bf 108}, 200501 (2012).
\bibitem{Brandao11} F. G. S. L. Brand\~ao, M. Horodecki, J. Oppenheim, J. M. Renes, and R. W. Spekkens,  The resource theory of quantum states out of thermal equilibrium. arXiv:1111.3882 (2011).
\bibitem{Berry} A. C. Berry, The accuracy of the gaussian approximation to 
 the sum of independent variates. {\it Trans. Amer. Math. Soc.} {\bf 49}, 122 (1941).
\bibitem{Esseen} C.-G. Esseen, On the Liapunoff limit of error in the theory of probability. {\it Ark. Mat. Astr. o. Fys.} {\bf 28A}, 1 (1942). 
\bibitem{Mehta} M. L. Mehta, {\it Random Matrices} (Elsevier, 2004).
\bibitem{Janzing} D. Janzing, P. Wocjan, R. Zeier, R. Geiss, and Th. Beth, Thermodynamic cost of reliability and low temperatures: Tightening Landauer's principle and the second law.
 {\it Int. J. Theor. Phys.} {\bf 39}, 2717 (2000).
 \bibitem{Liphardt02} J. Liphardt, S. Dumont, S. B. Smith, I. Jr Tinoco, and C. Bustamante, 
Equilibrium information from nonequilibrium measurements in an experimental test of Jarzynski's equality. {\it Science} {\bf 296}, 1832 (2002).
\bibitem{Collin05} D. Collin, F. Ritort, C. Jarzynski, S. B. Smith, I.Jr Tinoco, and C. Bustamante, Verification of the Crooks fluctuation theorem and recovery of RNA folding free energies.
{\it Nature} {\bf 437}, 231 (2005).
\bibitem{Toyabe10} S. Toyabe, T. Sagawa, M. Ueda, E. Muneyuki, and M. Sano, Experimental demonstration of information-to-energy conversion and validation of the generalized Jarzynski equality. {\it Nature Physics}, {\bf 6}, 988 (2010).
\bibitem{Holenstein11} T. Holenstein and R. Renner, On the randomness of independent experiments. {\it IEEE Trans. Inf. Theor.}, {\bf 57}, 1865 (2011). 
\bibitem{Tomamichel09} M. Tomamichel, R. Colbeck, and R. Renner, A fully quantum asymptotic equipartition property. {\it IEEE Trans. Inf. Theor.}, {\bf 55}, 5840 (2009).
\bibitem{Kawai} R. Kawai, J. M. R. Parrondo, and C. Van den Broeck, Dissipation: The phase-space perspective. {\it Phys. Rev. Lett.} {\bf 98}, 080602 (2007).
\bibitem{Horowitz} J. Horowitz and C. Jarzynski,  Illustrative example of the relationship between  dissipation and the relative entropy. {\it Phys. Rev. E.} {\bf 79}, 021106 (2009).
\bibitem{Zolfaghari} P. Zolfaghari, S. Zare, and B. Mizra, Generalized relation between relative entropy and dissipation for nonequilibrium systems. {\it Phys. Rev. E} {\bf 82}, 052104 (2010).
\bibitem{Vaikuntanathan} S. Vaikuntanathan and C. Jarzynski, Modeling Maxwell's demon with a microcanonical Szilard engine.  {\it Phys. Rev. E} {\bf 83}, 061120 (2011).
\bibitem{Allahverdyan} A. E. Allahverdyan, R. Balian, and Th. M. Nieuwenhuizen,  
Maximal work extraction from finite quantum systems. {\it Europhys. Lett.} {\bf 67}, 565 (2004).
\bibitem{Maroney} O. J. E. Maroney, Generalizing Landauer's principle {\it Phys. Rev. E} {\bf 79}, 031105 (2009).
\bibitem{Alicki} R. Alicki, M. Horodecki, P. Horodecki, and R. Horodecki, Thermodynamics of quantum information systems - Hamiltonian description. {\it Open Syst. Inf. Dyn.} {\bf 11}, 205 (2004).
\bibitem{Sagawa} T. Sagawa and M. Ueda, Minimal energy cost for the thermodynamic information processing: measurement and information erasure {\it Phys. Rev. Lett.} {\bf 102}, 250602 (2009).
\bibitem{Talkner} P. Talkner, E. Lutz, and P. H\"anggi, Fluctuation theorems: Work is not an observable, Phys. Rev. E {\bf 75}, 050102(R) (2007).
\bibitem{GemmerMahler} J. Gemmer, M. Michel , and G. Mahler, {\it Quantum Thermodynamics: Emergence of thermodynamic behavior within composite quantum systems}, Lecture notes in physics {784} (Springer, 2010). 
\bibitem{Goldstein} S. Goldstein, J. L. Lebowitz, R. Tumulka, and N. Zangh\`i, Canonical Typicality, Phys. Rev. Lett. {\bf 96}, 050403 (2006).
\bibitem{Popescu06} S. Popescu, A. J. Short, and A. Winter, Entanglement and the foundations of statistical mechanics, Nature Physics {\bf 2}, 754 (2006).
\bibitem{Partovi} M. H. Partovi, Entanglement versus Stosszahlansatz: Disappearance of the thermodynamic arrow in a high-correlation environment, Phys. Rev. E {\bf 77}, 021110 (2008).
\bibitem{Jennings} D. Jennings and T. Rudolph, Entanglement and the thermodynamic arrow of time, Phys. Rev. E {\bf 81}, 061130 (2010).
\bibitem{Palao01} J. P. Palao, R. Kosloff, and J. M. Gordon, Quantum thermodynamic cooling cycle. {\it Phys. Rev. E} {\bf 64}, 056130 (2001).
\bibitem{Youssef09} M. Youssef, G. Mahler, and A.-S. F. Obada, Quantum optical thermodynamic machines: Lasing as relaxation. {\it Phys. Rev. E} {\bf 80}, 061129 (2009).
\bibitem{Linden10} N. Linden, S. Popescu, and P. Skrzypczyk, How small can thermal machines be? The smallest possible refrigerator {\it Phys. Rev. Lett.} {\bf 105}, 130401 (2010).
\bibitem{Skrzypczyk11} P. Skrzypczyk, N. Brunner, N. Linden, and S. Popescu, The smallest refrigerators can reach maximal efficiency. {\it J. Phys. A: Math. Theor. {\bf 44}, 492002 (2011).}
\bibitem{Peres80} A. Peres, Measurement of time by quantum clocks. {\it Am. J. Phys.} {\bf 48}, 552 (1980).
\bibitem{Buzek} V. B\v uzek, R. Derka, and S. Massar, Optimal quantum clocks. {\it Phys. Rev. Lett.} {\bf 82}, 2207 (1999).
\bibitem{Gut} A. Gut, {\it Probability: A Graduate Course}, Springer Texts in Statistics (Springer, New York , 2005).
\bibitem{KLdiv} S. Kulback and R. A. Liebler, On information and sufficiency {\it Ann. Math. Stat.} {\bf 22}, 79 (1951). 
\bibitem{Landauer61} R. Landauer, Irreversibility and heat generation in the computing process. {\it IBM J. Res. Develop.} {\bf 5}, 183 (1961).
\bibitem{Bennett82} C. H. Bennett, The thermodynamics of computation - a review. {\it Int. J. Theor. Phys.} {\bf 21}, 905 (1982).
\bibitem{PlenioVitelli01} M. B. Plenio and V. Vitelli, The physics of forgetting: Landauer's erasure principle and information theory. {\it Contemporary physics} {\bf 42}, 25 (2001).
\bibitem{Maruyama09} K. Maruyama, F. Nori, and V. Vedral, Colloquium: The physics of Maxwell's demon and information.  {\it Rev. Mod. Phys.} {\bf 81}, 1 (2009).
\bibitem{Benett03} C. H. Bennett, Notes on Landauer's principle, reversible
computation, and Maxwell's Demon.  {\it Stud. Hist. Phil. Mod. Phys.} {\bf 34}, 501 (2003).
\bibitem{Shulman99} L. J. Shulman and U. V. Vazirani, in {\it Proc. 31st Annu. ACM Symp. Theory Comput.} (eds J. S. Vitter, L. Larmore, and T. Leighton) 322 (Association for Computing Machinery, 1999).
\bibitem{Curtis} J. H. Curtis, A note on the theory of moment generating functions {\it Ann. Math. Statist.} {\bf 13}, 430 (1942).
\bibitem{Mukherjea} A. Mukherjea, M. Rao, and S. Suen,  A note on moment generating functions {\it Statist. Prob. Lett.} {\bf 76}, 1185 (2006).
\bibitem{Ushakov} N. G. Ushakov and V. G. Ushakov, On convergence of moment generating functions {\it Statist. Prob. Lett.} {\bf 81}, 502 (2011).
\bibitem{Renyi} A. R\'enyi, On measures of entropy and information. {\it Proc. Fourth Berkeley Symp. on Math. Statist. and Prob.} {\bf 1}, 547 (1961).
\bibitem{Erven} T. van Erven and P. Harremo\"es, R\'enyi divergence and majorization. {\it International Symposium on Information Theory (ISIT) 2010}, 1335 (2010).
\bibitem{Bullen} P. S. Bullen, {\it Handbook of Means and their Inequalities}, Mathematics and its applications Vol. 560 (Kluwer, Dordrecht, 2003).
\bibitem{Petrov} V. V. Petrov, {\it Limit Theorems of Probability Theory} (Clarendon Press, Oxford, 1995).
\bibitem{RennerWolf} R. Renner and S. Wolf, Smooth R\'enyi entropy and applications. {\it Proc. ISIT 2004} (2004).
\bibitem{Jarzynski} C. Jarzynski, Nonequilibrium equality for free energy differences. {\it Phys. Rev. Lett.} {\bf 78}, 2690 (1997).
\bibitem{Copson} E. T. Copson, {\it Asymptotic Expansions}, Cambridge Tracts in Mathematics and Mathematical Physics No. 55 (Cambridge Univ. Press, London, 1965).
\bibitem{Tomamichel12} M. Tomamichel and M. Hayashi, A Hierarchy of information quantities for finite block length analysis of quantum tasks. arXiv:1208.1478 (2012). 
\bibitem{Korolev} V. Yu. Korolev and I. G. Shevtsova, On the upper bound for the absolute constant in the Berry-Esseen inequality.  {\it Theor. Prob. Appl.} {\bf 54}, 638 (2010). 
\bibitem{Fannes} R. Alicki and M. Fannes, Continuity of quantum conditional information {\it J. Phys. A: Math. Gen.} {\bf 37}, L55 (2004).
\bibitem{HandbIntegr} A. Jeffrey,  {\it Handbook of Mathematical Formulas and Integrals, 3rd Ed.} (Academic Press, Amsterdam, 2004).
\end{thebibliography}
\end{document}